\documentclass[12pt]{article}% Math and Physical Sciences Reference Style
\usepackage{authblk}
\usepackage{fullpage}
\usepackage{array}
\usepackage{pdflscape,afterpage}
\usepackage{hyperref}
\usepackage{bm}
\usepackage{xspace}
\usepackage{multicol}        % used for the two-column index
\usepackage{multirow}        % used for the two-column index
\usepackage{amsmath,amssymb,amsthm}
\usepackage{mathdots,mathtools}
\usepackage{xcolor}
\usepackage[defaultcolor=black]{changes}
\usepackage[algo2e,boxed,vlined,linesnumbered,titlenumbered,algosection]{algorithm2e}
\SetKw{KwFrom}{from}
\SetKwFor{For}{For}{do}{End for}
\SetKwFor{Forall}{For all}{do}{End for}
\SetKwFor{Foreach}{For each}{do}{End for}
\SetKw{KwRet}{Return}
\SetKw{KwGoto}{Goto}
\SetKw{KwErr}{Error}
\SetKw{KwAnd}{and}
\SetKw{KwOr}{or}
\SetKw{KwBreak}{break}
\SetKw{KwContinue}{continue}
\SetKw{KwTrue}{true}
\SetKw{KwFalse}{false}
\SetKw{KwFail}{fail}
\SetKw{KwNext}{next}
\SetKwFor{While}{While}{do}{End do}
\SetKwRepeat{Repeat}{Repeat}{until}
\SetKwIF{If}{ElseIf}{Else}{If}{then}{Else if}{Else}{End if}
\usepackage[noabbrev]{cleveref}
\crefname{subsection}{subsection}{subsections}
\newtheorem{theorem}{Theorem}[section]% meant for sectionwise numbers
\newtheorem{proposition}[theorem]{Proposition}% 
\newtheorem{lemma}[theorem]{Lemma}% 
\newtheorem{corollary}[theorem]{Corollary}% 

\theoremstyle{thmstyletwo}%
\newtheorem{example}[theorem]{Example}%
\newtheorem{remark}[theorem]{Remark}%

\theoremstyle{thmstylethree}%
\newtheorem{definition}[theorem]{Definition}%

\SetCommentSty{mycommfont}
\makeatletter
\@addtoreset{algocf}{section}
\makeatother

\usepackage{ifthen}

\newcommand\pare[1]{\left(#1\right)}

\newcommand\curl[1]{\left\{#1\right\}}
\newcommand\mcal{\mathcal}
\newcommand\mbf{\mathbf}
\newcommand\mbb{\mathbb}
\newcommand\msf{\mathsf}
\newcommand\mrm{\mathrm}

\newcommand\bx{\bm{x}}

\newcommand\bc{\mbf{c}}
\DeclareMathAlphabet{\mbfsf}{\encodingdefault}{\sfdefault}{bx}{n}
\newcommand\bi{\bm{i}}
\newcommand\bj{\bm{j}}

\newcommand\bs{\bm{s}}

\newcommand\bmf{\bm{f}}

\newcommand\bw{\bm{w}}

\newcommand\bphi{\bm{\vphi}}
\newcommand\bpsi{\bm{\psi}}

\newcommand\bmr{\bm{r}}
\newcommand\tiM{\tilde{M}}
\newcommand\tim{\tilde{m}}

\newcommand\seqw{w}
\newcommand\bseqw{\bw}

\newcommand\bone{\mbf{1}}

\newcommand\N{\mbb{N}}
\newcommand\Z{\mbb{Z}}
\newcommand\Q{\mbb{Q}}
\newcommand\Kbar{\overline{\mbb{K}}}
\newcommand\K{\mbb{K}}

\newcommand\bbQ{\mbb{Q}}

\newcommand\cG{\mcal{G}}
\newcommand\cH{\mcal{H}}
\newcommand\cI{\mcal{I}}
\newcommand\cJ{\mcal{J}}

\newcommand\cT[1][]{
  {\ifthenelse{\equal{#1}{}}{\mcal{T}}{\mcal{T}\!\!\pare{#1}}}}

\newcommand\rT{\mrm{T}}
\newcommand\h{\mrm{h}}
\newcommand\sM{\msf{M}}

\newcommand\vphi{\varphi}

\newcommand\ideal[1]{\left\langle#1\right\rangle}

\newcommand\gb{Gröbner basis\xspace}
\newcommand\gbs{Gröbner bases\xspace}
\newcommand\Fquatre{\textsc{\texorpdfstring{F\textsubscript{4}}{F4}}\xspace}
\newcommand\FquatreSAT{\textsc{\texorpdfstring{F\textsubscript{4}SAT}{F4SAT}}\xspace}
\newcommand\Fcinq{\textsc{\texorpdfstring{F\textsubscript{5}}{F5}}\xspace}
\newcommand\Ffour{\Fquatre}
\newcommand\FfourSAT{\FquatreSAT}

\newcommand\FGLM{\textsc{FGLM}\xspace}
\newcommand\spFGLM{\textsc{Sparse-FGLM}\xspace}
\newcommand\spFGLMcol{\textsc{Sparse-FGLM-colon}\xspace}
\newcommand\resp{\mbox{resp.}\xspace}

\DeclareMathOperator\DRL{\textsc{drl}}
\DeclareMathOperator\LEX{\textsc{lex}}
\newcommand\ldrl{\mathrel{\prec_{\DRL}}}

\newcommand\llex{\mathrel{\prec_{\LEX}}}

\newcommand\cGdrl{\cG_{\DRL}}
\newcommand\cGlex{\cG_{\LEX}}
\newcommand\cHlex{\cH_{\LEX}}
\newcommand\Sdrl{S_{\DRL}}

\newcommand\Tlex{T_{\LEX}}

\newcommand\adots{\mathinner{%
  \mkern1mu\raise1pt\hbox{.}%
  \mkern2mu\raise4pt\hbox{.}%
  \mkern2mu\raise7pt\vbox{\kern7pt\hbox{.}}\mkern1mu}} 

\newcommand\lcm{\textsc{lcm}\xspace}

\DeclareMathOperator\leadmon{\textsc{lm}}
\DeclareMathOperator\leadcoef{\textsc{lc}}
\DeclareMathOperator\leadterm{\textsc{lt}}
\newcommand\LM[1][\prec]{\leadmon_{#1}}
\newcommand\LC[1][\prec]{\leadcoef_{#1}}
\newcommand\LT[1][\prec]{\leadterm_{#1}}

\DeclareMathOperator\Staircase{Staircase}

\DeclareMathOperator\supp{supp}
\DeclareMathOperator\HS{HS}
\DeclareMathOperator\HP{HP}

\DeclareMathOperator\ComputeMaxDegree{\textsf{ComputeMaxDegree}}
\DeclareMathOperator\LinearizeColonIdeal{\textsf{LinearizeColonIdeal}}
\DeclareMathOperator\NF{NF}

\DeclareMathOperator\Berlekamp{Berlekamp--Massey}
\DeclareMathOperator{\spairop}{sp}
\newcommand{\spair}[3][\prec]{\spairop_{#1}\pare{{#2,#3}}}
\newcommand\nf[3]{\NF\pare{#1,#2,#3}}
\newcommand\Card[1]{\#\,#1}
\newcommand\elim{\ensuremath{g_n}\xspace}
\newcommand\param[1]{\ensuremath{g_{#1}}\xspace}

\newcommand\wrt{w.r.t.\xspace}

\newcommand\ie{i.e.\xspace}
\newcommand\msolve{\textsc{msolve}\xspace}
\newcommand\Maple{\textsc{Maple}\xspace}
\newcommand\Singular{\textsc{Singular}\xspace}

\newcommand\colid[2]{#1:\ideal{#2}}
\newcommand\satid[2]{\colid{#1}{#2}^{\infty}}
\newcommand\diff[2]{\frac{\partial #1}{\partial #2}}
\colorlet{diffRed}{red!80!black}

\newcommand\nameSigSigpDeg[4]{
  #1 &$#2$ &$#3$ &$#4$
}
\newcommand\FquatreSatMatrixSPFGLMTotalMapleBMapleFMapleT[8]{
  &$#1$ &$#2$ &$#3$ &$#4$ &$#5$ &$#6$ &$#7$ &$#8$
}

\title{New efficient algorithms for computing Gröbner bases of
  saturation ideals (\textsc{\texorpdfstring{F\textsubscript{4}SAT}{F4SAT}})
  and colon ideals (\textsc{Sparse-FGLM-colon})}

\author[1]{Jérémy Berthomieu}%\email{jeremy.berthomieu@lip6.fr}

\author[2]{Christian Eder}%\email{ederc@mathematik-uni-kl.de}

\author[1]{Mohab Safey El Din}%\email{mohab.safey@lip6.fr}

\affil[1]{Sorbonne Université, CNRS,
  LIP6,
  F-75005, Paris, France}

\affil[2]{Rheinland-Pfälzische Technische Universität
  Kaiserslautern Landau,
  Kaiserslautern, Germany}

\begin{document}

\maketitle


\abstract{This paper is concerned with linear algebra based methods for solving exactly
polynomial systems through so-called Gr\"obner bases, which allow one to compute
modulo the polynomial ideal generated by the input equations. This is a topical
issue in non-linear algebra and more broadly in computational mathematics
because of its numerous applications in engineering and computing sciences. Such
applications often require geometric computing features such as representing the
closure of the set difference of two solution sets to given polynomial systems.
Algebraically, this boils down to computing Gr\"obner bases of colon and/or
saturation polynomial ideals. In this paper, we describe and analyze new
Gr\"obner bases algorithms for this task and present implementations which are
more efficient by several orders of magnitude than the state-of-the-art
software.

%%% Local Variables:
%%% mode: latex
%%% TeX-master: "new_saturation"
%%% End:

}

%%================================%%
%% Sample for structured abstract %%
%%================================%%

%\keywords{Gr\"obner bases, Saturation ideals, Colon ideals, Algorithms}

%\pacs[MSC Classification]{13P10, 13P15, 13-04, 14Q20}

\section{Introduction}\label{s:intro}
% \paragraph*{Problem statement and motivations.}~ 
Let $\bmf = (f_1, \ldots, f_s)$ and $\varphi$ be polynomials in the polynomial
ring $\K[x_1, \ldots, x_n]$ where $\K$ is a field. Further, we denote by $I =
\langle \bmf \rangle = \langle f_1, \ldots, f_s \rangle$ the polynomial ideal
generated by $f_1, \ldots, f_s$ and by $V(I)\subset \Kbar^n$ the algebraic set
associated to $I$ (where $\Kbar$ is an algebraic closure of $\K$).

We consider the following computational problem: compute a \gb
associated to the colon (\resp saturated) ideal of $I$ by $\varphi$, i.e.
\[
  \colid{I}{\varphi} = \{h \mid h \varphi\in I\}\quad \text{ (\resp } \quad \satid{I}
  {\varphi} = \{h \mid \exists k \in \N \ h \varphi^k\in
  I\}\text{)} .
\]
By e.g.~\cite[Chap. 4]{CoxLOS2015}, the algebraic set
$V(\satid{I}{\varphi})\subset \Kbar^n$ is the Zariski closure of the
set difference 
$V(I) \setminus V(\varphi)$ and there exists $N\in \N$ such that
$\colid{I}{\varphi^N} = \colid{I}{\varphi^{N+1}} = \cdots = \satid{I}{\varphi}$.

Computing algebraic representations of saturated ideals arises in many
applications ranging from experimental mathematics to engineering sciences (see
e.g.~\cite{Steiner, garciafontan:hal-03499974, pascualescudero:hal-03070525})
since some natural algebraic modelings come with parasite solutions which one
excludes through some saturation process. For instance, modeling that some $p
\times q$ matrix with polynomial entries has rank $r$ through the simultaneous
vanishing of its $(r+1)$-minors will include those points at which the matrix
has rank less than $r$.

In the paper, we design new efficient algorithms for computing
\emph{\gbs} of such ideals, given as input $\bmf = (f_1, \ldots,
f_s)$, 
$\varphi$ and some admissible monomial order $\prec$ over the monomials
in $\K[x_1, \ldots, x_n]$. 

Recall that \gbs are finite generating sets of polynomial ideals
capturing their combinatorial and algebraic properties. They allow to compute
\emph{modulo} the ideal they generate, hence to decide the membership of a
polynomial to the ideal under consideration. \gbs algorithms are a
classical and versatile tool for polynomial system solving, non-linear algebra
and geometry, implemented in most of computer algebra systems.

\paragraph{Prior results.}~ The first algorithm for computing \gbs
is introduced by Buchberger in~\cite{bGroebner1965}. It is based on the
so-called Buchberger's criterion which provides an effective way to decide
whether a given polynomial sequence is already a \gb of the ideal it
generates. Modern algorithms such as Faugère's \Fquatre~\cite{Faugere1999} and
\Fcinq~\cite{Faugere2002} (see also~\cite{EF17}) algorithms actually use the
connection of \gb theories with Macaulay's constructions for the
multivariate resultant (see e.g. \cite{Lazard1983}) by considering the
finite-dimensional
vector spaces
\[
  E_d(\bmf) = \{h_1 f_1 + \cdots + h_s f_s \mid \deg(h_i f_i)\leq d \text{ for
    all } 1\leq i \leq s\}
\]
for which a basis with appropriate properties \wrt the given monomial order
is computed through the row echelonization of some Macaulay-like matrix. The way
these linear algebra constructions are generated at each degree $d, d+1, \ldots$
(and so on) plus a termination criterion is done via a connection to the
\gb theory and Buchberger's criterion in \Fquatre. The \Fcinq algorithm
poses a module theoretic view of \gb calculations which allows
one to generate Macaulay-like matrices of maximal rank under some genericity
assumptions as well as a module theoretic transposition of the notion of
critical pairs through the notion of \emph{signature} to handle termination
issues in this context. These two algorithms have been used to solve many
difficult applications and challenges of polynomial system solving (see e.g.
\cite{FaugereJ2003,FaugereJPT2010,Steel2015}). Such algorithms are
usually run with so-called
total degree monomial orders, i.e. those orders which filter monomials
first \wrt their total degrees.

When $I$ has dimension zero (i.e. $V(I)$ is a non-empty finite set) this linear
algebra view of \gb computations is often used in change of order
algorithms since the quotient ring $\K[x_1, \ldots, x_n] / I$ is a
finite-dimensional
vector space. Based on this, the
so-called \FGLM algorithm~\cite{FaugereGLM1993} reduces change of order
algorithms to kernel computations. Under some extra assumptions, the so-called
\spFGLM algorithm~\cite{FaugereM2011,FaugereM2017} makes the connection
with relation reconstructions.

Despite these developments, computing saturations of polynomial ideals is
currently done using the above \gb algorithms \emph{as a black box}.

Using Rabinowitsch trick~\cite{Rabinowitsch1930} and~\cite[Chap.~4, Sec.~4,
Th.~14, (ii)]{CoxLOS2015}, the saturated ideal $\satid{I}{\vphi}$ equals
$\pare{I+\ideal{1-t\vphi}}\cap\K[x_1,\ldots,x_n]$. Thus, computing a \gb of
$I+\ideal{1-t\vphi}$ for a monomial order eliminating $t$ % , \ie one that
% compares two monomials first by considering their degrees in $t$ and then that
% breaks ties using a monomial order $\prec$ on $x_1,\ldots,x_n$.
and keeping all polynomials not involving $t$ yields a \gb of $\satid{I}{\vphi}$,
see also~\cite[Chap.~3, Sec.~1, Th.~2 and Ex.~6]{CoxLOS2015}.

Moreover, if $I$ is homogeneous, \ie it is spanned by a set of
homogeneous polynomials, Bayer's algorithm~\cite{Bayer1982} allows us to
compute $\satid{I}{x_n}$. If it is not, then one can still recover a
\gb of $\satid{I}{x_n}$ using the algorithm below, still called
Bayer's algorithm:
\begin{enumerate}
\item Homogenize the input polynomials $f_1, \ldots, f_s$ with a new variable
  $x_0$ yielding homogeneous polynomials $\bmf^{\h} = (f^{\h}_1, \ldots, f^{\h}_s)$;
\item compute a \gb $G^{\h}$ for $\bmf^{\h}$ and a total degree monomial
  order (called graded reverse lexicographical order) where
  $x_n$ is smaller
  than the other variables $x_0,x_1, \ldots, x_{n-1}$;
\item factor out from all polynomials in $G^{\h}$ the highest possible power of
  $x_n$;
\item set $x_0$ to $1$ in these obtained polynomials and return the result. 
\end{enumerate}
When $\varphi \neq x_n$, one just introduces a slack variable $x_{n+1}$ and
computes the saturation of $I+\langle  x_{n+1}-\varphi \rangle$ \wrt $x_{n+1}$.

The above two approaches constitute the state-of-the-art algorithms for
computing saturations of ideals. Note that they do not take advantage of
intermediate data obtained during the \gb computations since these
are used as black boxes.

\paragraph{Main results.} In this paper, we propose new algorithms which
actually compute ``on the fly'' \gbs of saturated ideals through the
linear algebra approaches we sketched above. We design two families of efficient
algorithms which are the counterparts of the \Fquatre and the \FGLM algorithms.
We also present (publicly available) implementations of these algorithms which
are more efficient than the state-of-the-art software in computer algebra
systems by several orders of magnitude.

The first algorithm, named
\FquatreSAT, is a modification of
the \Fquatre algorithm to discover on the fly polynomials in
$\satid{I}{\varphi}$. The core idea is as follows. Recall that, on
input $\bmf = 
(f_1, \ldots, f_s)$ and $\varphi$ in $\K[x_1, \ldots, x_n]$ and a given total degree
monomial order $\prec$, the \Fquatre algorithm roughly computes bases $G_d$ of the
finite-dimensional vector spaces $E_d$, we introduced above, using $G_d$ to
generate a generating family for $E_{d+1}$ (using the notion of critical pairs,
see~\cite{Faugere1999}) and so on. Termination is ensured using Buchberger's
criterion.

We show that, during this process, one can search for polynomials $h$ of maximum
prescribed degree $\delta$ in the colon ideal $\colid{I}{\vphi}$ such
that $h \varphi\in E_d$ using \emph{(i)} the computation of \emph{normal forms} of
$m . \varphi$ where $m$ lies in a set of well-chosen monomials; and \emph{(ii)}
the computation of the kernel of some matrix which is built from the above
normal forms.

This algorithmic strategy allows us to discover on the fly new polynomials in
the colon ideal $\colid{I}{\vphi}$ which are then taken into account early in the whole
computation. Repeating this, with (maybe incomplete) generating sets
of $\colid{I}{\varphi}$ allows us to discover polynomials in
$\colid{\pare{\colid{I}{\vphi}}}{\vphi}=\colid{I}{\vphi^2}$ and so on.

We prove how to complete such a computation and how the above prescribed degree
$\delta$ can be chosen to ensure correctness of the algorithm. 

When the ideal $\satid{I}{\varphi}$ is known to be
zero-dimensional in advance, one can adapt \FGLM-like algorithms, assuming
we have precomputed a \gb for $I$ \wrt some monomial order, to
compute a \gb for $\satid{I}{\varphi}$ \wrt a
so-called lexicographical order (yielding a triangular basis). Here the main
difficulty to overcome is that since \emph{we do not assume that $I$ is
  zero-dimensional}, the vector space $\K[x_1, \ldots, x_n] / I$ is of infinite
dimension. We demonstrate how the change of order can still be realized through
linear algebra techniques which borrow from \FGLM the construction of matrices
representing multiplication operators in the quotient ring from which one can
extract the lexicographical \gb. We show how to use algebraic
properties to reduce the size of such matrices and state the complexity of
our approach when only the matrix representing the multiplication by
the last variable is needed. All in all, this new algorithm
reduces the change of order in this context to the computation of minimal
relations satisfied by sequences of scalars computed from the aforementioned
matrices as well as Hankel linear system solving.

Next, we present our implementation, which we wrote using the \texttt{C}
programming language and which is available in the \msolve
library~\cite{msolve,msolveweb}. We compare it against implementations based on the
Rabinowitsch trick using \gb engine in \msolve (which is one of
the fastest open source implementations), the one in \Maple (which is one
of the fastest in commercial computer algebra systems), and the leading software
for algebra and geometry \Singular. We carefully
analyze the practical behaviors of the new algorithms in this paper. Our
experiments show that on many examples the new algorithms are faster, often by
several orders of magnitude, than the state-of-the-art software alternatives.

\paragraph{Structure of the paper.} \Cref{s:prelim} is devoted to fix
some notation we use about \gbs and recall the basics of \Fquatre and
\FGLM algorithms needed to describe our new algorithms. \Cref{s:F4SAT}
describes the \FquatreSAT algorithm, and its correctness and termination proofs.
\Cref{s:SpFGLMcol} focuses on the \FGLM variant for saturation. Finally,
\Cref{s:implem} presents our implementations and compares it with the
state-of-the-art software. 

%%% Local Variables:
%%% mode: latex
%%% TeX-master: "new_saturation"
%%% End:

\section{Preliminaries}\label{s:prelim}
\subsection{\gbs}
We recall some basic definitions and properties of \gbs. We refer
to~\cite[Chap.~2, 3 and~5]{CoxLOS2015} for more details.

Throughout this paper, let $\K$ be a field and  $0\in\N$.
We denote by $\K[\bx]\coloneqq\K[x_1,\ldots,x_n]$ the polynomial ring in $n$
variables $x_1, \ldots, x_n$ with coefficients in $\K$. A polynomial $f \in
\K[\bx]$ is defined as $f=\sum_{\bs\in\N^n}f_{\bs}\bx^{\bs}$ such that $f_{\bs}
= 0$ for all but finitely many $\bs \in \N^n$. For $f \neq 0$ we define its
support $\supp f=\curl{\bs\in\N^n\,\middle\vert\,f_{\bs}\neq 0}$. Otherwise, by
convention, $\supp 0 =\curl{\mathbf{0}}$.

A monomial order $\prec$ on $\K[\bx]$ is a total order on the set of
monomials such that for all monomials $m, m'$ and
$s$, if $m\preceq m'$, then $m s\preceq m' s$. Furthermore, the monomial orders
in this paper are assumed to be well-orders, i.e. for all monomials $m$ we
have that $1 \preceq m$.

Fix a monomial order $\prec$.
Given a polynomial
$f\in\K[\bx]$, we  define its \emph{leading
  monomial}, denoted by $\LM(f)$, the largest monomial in
$f$ for $\prec$. The \emph{leading coefficient} of $f$, $\LC(f)$, is
the coefficient of $\LM(f)$ and the \emph{leading term} of $f$,
$\LT(f)$ is $\LC(f)\LM(f)$.
For a set $G\subseteq\K[\bx]$, we let
$\LM(G)=\{\LM(f)\,\vert\,\ f\in G\}$.
For an ideal $I \subset \K[\bx]$ we define $\LM(I)$ as the ideal generated by
leading monomials of all elements of $I$.
We recall briefly the definition of a \gb and of its associated staircase.

\begin{definition}
  A set of monomials $S$ is a \emph{staircase} if
  for two monomials $\mu_1$ and $\mu_2$ such that $\mu_1\mu_2\in S$,
  we have
  $\mu_1\in S$ and $\mu_2\in S$.
\end{definition}

\begin{definition}[{\cite[Chap.~2, Sec.~5, Def.~5 and Sec.~7,
    Def.~4]{CoxLOS2015}}]\label{def:staircase}
  Let $I$ be a nonzero ideal of
  $\K[\bx]$
  and let $\prec$ be
  a monomial order.
  A set $\cG\subseteq I$ is a \emph{\gb} of $I$ for $\prec$ if for all $f\in I$,
  there exists $g\in\cG$ such that $\LM(g)\mid\LM(f)$ or, equivalently,
  if $\ideal{\LM(\cG)}=\LM(I)$.
  It is \emph{reduced} if for any $g\in\cG$, $g$ is monic, \ie
  $\LC(g)=1$, and for any $g'\in\cG\setminus\curl{g}$ and any monomial
  $m\in\supp g'$,
  $\LM(g)\nmid m$.

  The \emph{staircase} associated to $\cG$
  is the set of monomials $\Staircase(\cG)$ which are not divisible by
  any $\LM(g)$ for $g\in\cG$, i.e. the complement of $\LM(I)$
  in the set of monomials.
\end{definition}
Once a monomial order $\prec$ is chosen, a monomial basis of the
quotient algebra $\K[\bx]/I$ can be canonically set: it is the set of
monomials that are not leading monomials of polynomials in $I$ \wrt $\prec$. In
other words, this is $\Staircase(\cG)$, where $\cG$ is a \gb of $I$
for $\prec$, see~\cite[Chap.~5, Sec.~3,
Prop.~1]{CoxLOS2015}. Furthermore, if $\K[\bx]/I$ is a
finite-dimensional $\K$-vector space, then 
$I$ is said to be \emph{zero-dimensional} of \emph{degree}
$\dim_{\K}\K[\bx]/I$, otherwise it is \emph{positive-dimensional}.

In this paper, we mainly deal with the lexicographic (LEX, $\llex$) and
degree reverse lexicographic (DRL, $\ldrl$) orders with the convention that
$x_n\prec\cdots\prec x_2\prec x_1$. They are
defined as below:

\begin{description}
\item[LEX:]
  $\bx^{\bi}\llex\bx^{\bj}$ if, and only if, there
  exists $1\leq p\leq n$ such that for all $q<p$,
  $i_q=j_q$ and $i_p<j_p$,
  see~\cite[Chap.~2, Sec.~2, Def.~3]{CoxLOS2015};
\item[DRL:]
  $\bx^{\bi}\ldrl\bx^{\bj}$ if, and only if,
  $i_1+\cdots+i_n<j_1+\cdots+j_n$ or both $i_1+\cdots+i_n=j_1+\cdots+j_n$
  and there exists
  $2\leq p\leq n$ such that for all $q>p$, $i_q=j_q$
  and $i_p>j_p$,
  see~\cite[Chap.~2, Sec.~2, Def.~6]{CoxLOS2015}. Observe that it is a
  total degree monomial order.
\end{description}

An important property of \gbs is that given a polynomial
$f\in\K[\bx]$ and $\cG=\curl{g_1,\ldots,g_t}$ a \gb of an ideal of
$\K[\bx]$ for $\prec$, there exist polynomials $h_0,h_1,\ldots,h_r$, with $h_0$
unique, such that $f=g_1 h_1+\cdots+g_t h_r+h_0$ and $\LM(h_0)$ is not
divisible by $\LM(g_1),\ldots,\LM(g_t)$. This polynomial $h_0$ is
called the \emph{normal form of $f$ with respect to $\cG$ for $\prec$}
and will be denoted by $\nf{f}{\cG}{\prec}$.

\begin{definition}[{\cite[Chap.~10, Sec.~2]{CoxLOS2015}}]
  Let $I$ be an ideal of $\K[\bx]$ spanned by homogeneous
  polynomials. Let $\K[\bx]_d$ (\resp $I_d$) be the subset of
  homogeneous polynomials of degree $d$, together with the zero
  polynomial, of
  $\K[\bx]$ (\resp $I$).
  
  The \emph{Hilbert series $\HS_{\K[\bx]/I}(t)$ of
    $\K[\bx]/I$} is the generating series in $t$ of 
  the sequence $\pare{\dim_{\K} \K[\bx]_d/I_d}_{d\in\N}$, \ie
  \[\HS_{\K[\bx]/I}(t)
    =\sum_{d\geq 0}\dim_{\K}\K[\bx]_d/I_d\,t^d
    =\frac{N(t)}{(1-t)^{\delta}},\]
  where $N(t)$ is coprime with $1-t$. The integer $\delta\geq
  0$ is called the \emph{dimension of $I$}.
\end{definition}

\begin{lemma}[see also~{\cite[Th.~15.26]{Eisenbud1995}}]\label{lm:dimId}
  Let $I$ be an ideal of $\K[\bx]$ spanned by homogeneous
  polynomials. Let $\prec$ be any total degree monomial order. Then, for any
  $d\in\N$, $\dim_{\K}\ideal{\LM(I)}_d=\dim_{\K} I_d$ and
  $\dim_{\K}\K[\bx]_d/\ideal{\LM(I)}_d=\dim_{\K}\K[\bx]_d/I_d$.
\end{lemma}

% \begin{proof}
%   \added{Let $d\in\N$, $r=\dim_{\K}\ideal{\LM(I)}_d$ and
%     $s=\dim_{\K}I_d$. We shall prove that $r=s$ by double
%     inequality.}
  
%   \added{Since $\dim_{\K} I_d=s$, there exist $s$ polynomials, say
%     $f_1,\ldots,f_s\in I_d$, that are
%     $\K$-linearly independant. Performing Gaussian elimination on
%     them, we obtain a triangular basis: that is polynomials
%     $g_1,\ldots,g_s\in I_d$, all with distinct leading
%     monomials. Hence $\dim_{\K}\ideal{\LM(I)}_d=r\geq s$.}
  
%   \added{Since $\dim_{\K}\ideal{\LM(I)}_d=r$, there exist $r$
%     distinct monomials in $\ideal{\LM(I)}_d$. Hence, there exist $r$
%     polynomials $h_1,\ldots,h_r$ in $I_d$, with distinct leading monomials for
%     $\prec$. Since $h_1,\ldots,h_r$ are necessarily $\K$-linearly
%     independant, we have $\dim_{\K} I_d=s\geq r$ and thus $r=s$.}
  
%   \added{The second equality is a direct consequence of the dimension
%     of a quotient vector space. This concludes the proof.}
% \end{proof}

\begin{corollary}\label{cor:sameHS}
  Let $I$ be an ideal of $\K[\bx]$ spanned by homogeneous
  polynomials. Let $\prec_1$ and $\prec_2$ be two total degree
  monomial orders. Then,
  \[\HS_{\K[\bx]/\ideal{\LM[\prec_1](I)}}(t)
    =\HS_{\K[\bx]/I}(t)
    =\HS_{\K[\bx]/\ideal{\LM[\prec_2](I)}}(t).\]
\end{corollary}

\begin{proof}
  For any $d\in\N$, by \Cref{lm:dimId},
  $\dim_{\K}\K[\bx]_d/\ideal{\LM[\prec_1](I)}_d=\dim_{\K}\K[\bx]_d/I_d
  =\dim_{\K}\K[\bx]_d/\ideal{\LM[\prec_2](I)}_d$. Hence, the Hilbert
  series are the same.
\end{proof}

\begin{proposition}[{\cite[Chap.~9, Sec.~3, Prop.~3]{CoxLOS2015}}
  and {\cite[Cor.~3.3]{Hartshorne1984}}]\label{prop:regindex}
  Let $I$ be an ideal of $\K[\bx]$ spanned by homogeneous
  polynomials of dimension $\delta$.
  Then, there exists a unique polynomial $\HP_I\in\Z[s]$, called the
  \emph{Hilbert polynomial of $I$}, and an integer
  $r\geq 0$, called the \emph{regularity index of $I$}, such that
  for all $d\geq r$, $\dim_{\K}\K[\bx]_d/I_d=\HP_I(d)$.
  Furthermore, there exist integers $m_0\geq\cdots\geq
  m_{\delta-1} > 0$ such
  that
  \[\HP_I(s)=\sum_{i=0}^{\delta-1}\pare{\binom{s+i}{i+1}-\binom{s+i-m_i}{i+1}},\]
  where the empty sum for $\delta=0$ is by convention $0$.
\end{proposition}

% \added{
%   \begin{theorem}\label{th:maxdeggb}
%     Let $I$ be an ideal of $\K[\bx]$ of dimension $\delta\leq 1$ spanned by homogeneous
%     polynomials. Let $\prec$ be a total degree
%     monomial orders and $\cG=\curl{g_1,\ldots,g_t}$ be the reduced \gb
%     of $I$ \wrt $\prec$. Let 
%     $r$ be the regularity index $I$. Then for all
%     $1\leq i\leq r$, $\deg g_i\leq d+1-\delta$.
%   \end{theorem}
% }

% \begin{proof}
%   % By Lemma~\ref{lm:dimId}, we know that for any $d\in\N$, $\dim_{\K}
%   % I_d=\dim_{\K}\ideal{\LM[\prec_1](I)}_d=\dim_{\K}\ideal{\LM[\prec_2](I)}_d$. In
%   % other words, there are as many leading monomials for $\prec_1$ of
%   % degree $d$ as there are for $\prec_2$ of degree $d$.
%   \added{If $I$ is zero-dimensional, then $d$ is minimal such that
%     $\dim_{\K}\K[\bx]/I_d=0$. It is then clear that no polynomial in
%     $\cG$ has degree higher than $d+1$.
%   }

%   \added{Now, if $I$ is positive-dimensional, then there exists a
%     $d\in\N$, such that for all $D\geq d$,
%     $\dim_{\K}\K[\bx]_D/I_D=\dim_{\K}\K[\bx]_d/I_d>0$. Furthermore,
%     each monomial in $\K[\bx]_{d+1}/I_{d+1}$ is obtained by
%     multiplying a monomial in $\K[\bx]_d/I_d$ by one $x_i$.
%   }
%   Observe that if
%   $\dim_{\K}\K[\bx]_d/I_d>\dim_{\K}\K[\bx]_{d+1}/I_{d+1}$, then there
%   is a polynomial of degree $d+1$ in $\cG$.
  
% \end{proof}

\begin{theorem}[{\cite[Sec.~4.5, Cor.]{MollerM1984}}]\label{th:maxdeggb}
  Let $I$ be an ideal of $\K[\bx]$
  % of dimension $\delta$
  spanned by homogeneous
  polynomials. Let $\prec$ be a total degree
  monomial order and $\cG=\curl{g_1,\ldots,g_t}$ be the reduced \gb
  of $I$ \wrt $\prec$. Let $r$ be the regularity index of $I$ and $m_0$
  be defined as in \Cref{prop:regindex}. Then, for all
  $1\leq i\leq t$, $\deg g_i\leq \max(r,m_0)$.
\end{theorem}

\begin{proposition}\label{prop:maxdeggbdim1}
  Let $I$ be an ideal of $\K[\bx]$
  of dimension $\delta\leq 1$
  spanned by homogeneous
  polynomials. Let $\prec$ be a total degree
  monomial order and $\cG=\curl{g_1,\ldots,g_t}$ be the reduced \gb
  of $I$ \wrt $\prec$. Let $r$ be the regularity index of $I$. Then, for all
  $1\leq i\leq t$, $\deg g_i\leq r+\delta$.
\end{proposition}
\begin{proof}
  Let us first consider the case $\delta=0$. By definition of $r$, all
  monomials of degree $r$ are leading monomials of $I$. Since a
  reduced \gb is minimal, its leading monomials cannot have degree
  greater than $r$.

  In the case $\delta=1$, by \Cref{prop:regindex},
  \[\HP_I(s)=\binom{s}{1}-\binom{s-m_0}{1}=m_0.\]
  That is, the number of monomials of degree $s$ in the staircase is
  constant equal to $m_0$ for $s\geq r$. Now, let us consider the sets
  of monomials of degree respectively $r$ and $r+1$ in the
  staircase. Since both sets have size $m_0$, there exists a
  bijection between them. We can build such a bijection as follows:
  enumerate the monomials $\mu$ of degree $r+1$ in increasing order for $\ldrl$
  and choose their image as follows: pick the largest index $k$
  such that $x_k|\mu$ and $\mu/x_k$, of degree $r$, has not been
  already chosen, then send $\mu$ onto $\mu/x_k$. We can build a
  similar bijection from the subset of monomials of degree
  $r+2$ of the staircase to the one of monomials of degree
  $r+1$. Now, by enumerating the possibilities, we see that if $\mu$
  of degree $r+2$ is sent onto $\mu/x_k$, then $\mu/x_k$ is sent
  onto $\mu/x_k^2$. Hence, the set of monomials of degree $r+1$
  outside the staircase, \ie in $\ideal{\LM(I)}$, completely spans the sets of all
  monomials of degree $d\geq r+1$ in $\ideal{\LM(I)}$, that is no
  polynomial of degree greater than $r+1$ lies in the reduced \gb of
  $I$ for a total degree ordering.
\end{proof}

% \begin{remark}
%   After performing a generic linear change of
%   variables, the bound of Theorem~\ref{th:maxdeggb} actually only
%   depends on the
%   degree of regularity, see~\cite[Th.~2]{Lazard1983} and
%   \cite[Sec.~3.9, Th.]{Giusti1984}.
% \end{remark}
\begin{theorem}[see~{\cite[Sec.~3.9, Th.]{Giusti1984}}
  and~{\cite[Th.~2]{Lazard1983}}]\label{th:maxdeggbgeneric}
  \added{
    Let $I=\ideal{f_1,\ldots,f_s}$ be an ideal of
    $\K[\bx]$ with
    $\deg f_i=d_i$ for all $1\leq i\leq s$. Assume $I$ is Cohen-Macaulay
    or has dimension at most $1$.
  }

  \added{
    Let $M$ be a generic invertible matrix of size $n$ and let
    $\tilde{f}_i=f_i(M\bx)$ for all $1\leq i\leq s$. Let $\tilde{\cG}$
    be the reduced \gb of
    $\tilde{I}=\ideal{\tilde{f}_1,\ldots,\tilde{f}_s}$ for $\ldrl$.
  }

  \added{
    Then, the index of regularity $r_{\tilde{I}}$ of $\tilde{I}$
    satisfies % the Macaulay bound
    \[r_{\tilde{I}}\leq\sum_{k=1}^s (d_i-1) + 1.\]
    Furthermore, for all $g\in\tilde{\cG}$,
    \[\deg g\leq r_{\tilde{I}}\leq\sum_{k=1}^s (d_i-1) + 1.\]
  }
\end{theorem}

\begin{example}
  % As an illustration to Theorem~\ref{th:maxdeggb}, consider the
  % zero-dimensional homogeneous ideal
  % $I=\ideal{x^2,y^2-x y}$ in $\K[x,y]$. Then, its \gb for $\prec_1$,
  % which is DRL with $x\prec_1 y$, is $\cG_1=\curl{x^2,y^2-x y}$. Its
  % \gb for $\prec_2$, which is DRL with $y\prec_2 x$, however, is
  % $\cG_2=\curl{x y-y^2,x^2,y^3}$.

  % Notice that the staircase $S_1$ for $\prec_1$ is
  % $\curl{1,y,x,x y}$, while the staircase $S_2$ for $\prec_2$ is
  % $\curl{1,y,x,y^2}$. This also illustrates
  % Corollary~\ref{cor:sameHS}.
  \added{
    As an illustration to \Cref{th:maxdeggb} and \Cref{prop:maxdeggbdim1}, consider the
    one-dimensional homogeneous ideal
    $I\subset\K[x,y]$ whose \gb for $\ldrl$ with $y\ldrl x$ is
    $\cG=\curl{x y-y^2,y^3}$.
  }
  
  \added{
    Notice that the associated staircase $S$ is
    $\curl{1,y,x,y^2,x^2,x^3,x^4,\ldots}$ so that
    $\HS_{\K[x,y]/I}(t)=1+2 t + 2 t^2+t^3+t^4+\cdots$. Hence
    $\HP_I(s) = 1 = m_0$, $r=3$ and their maximum bounds the degree of
    the polynomials in $\cG$.
  }

  \added{
    Taking a generic matrix $M=\pare{
      \begin{smallmatrix}
        a	&b\\c	&d
      \end{smallmatrix}
    }$ as in \Cref{th:maxdeggbgeneric} yields the ideal $\tilde{I}$ whose
    \gb for $\ldrl$ is $\tilde{\cG}=
    \curl{(a c - c^2) x^2 + (a d + b c - 2 c d) x y + (b d - d^2) y^2,
      c x^2 y + d y^3}$, under the generic assumption
    $(a - c) c\neq 0$. Thus, we have $\deg g\leq r=3\leq
    (2-1)+(3-1)+1=4$, for all $g\in\tilde{\cG}$.
  }
\end{example}

\subsection{\gb algorithms}
Buchberger developed the theory of \gbs and designed a first algorithm
to compute them in~\cite{bGroebner1965}. Since then, many efficient
\gb algorithms were developed. Here, we focus on Faugère's
\Fquatre algorithm~\cite{Faugere1999}.

\subsubsection{The \Fquatre algorithm}
\label{ss:F4}
In~\cite{bGroebner1965}, Buchberger's algorithm introduced the
concept of \emph{critical pairs}
for computing \gbs. For
two polynomials $f_1$ and $f_2$ in a set of generators of an ideal, the
critical pair $(f_1,f_2)$ leads to a normal form computation of the
\emph{S-polynomial}
\[\spair{f_1}{f_2}=\frac{\lcm(\LM(f_1),\LM(f_2))}{\LT(f_1)}f_1
  -\frac{\lcm(\LM(f_1),\LM(f_2))}{\LT(f_2)}f_2\]
\wrt the current intermediate basis.
The \emph{degree} of such a critical pair is defined as
$\deg\lcm (\LM(f_1),\allowbreak \LM(f_2))$. Notice that this bounds from above
$\deg\spair{f_1}{f_2}$.

In
\Cref{algo:F4} we state the pseudocode of Faugère's \Fquatre
algorithm, highlighting (in
red) the main differences to Buchberger's algorithm.

\begin{algorithm2e}[htbp!]
  \small
  \DontPrintSemicolon
  \caption{Faug\`ere's \Ffour\label{algo:F4}}
  \KwIn{A list of polynomials $f_1, \ldots, f_s$ spanning an ideal
    $I\subseteq\K[\bx]$ and a total degree monomial order $\prec$.}
  \KwOut {A \gb $\cG$ of $I$ for $\prec$.}
  $\cG\coloneqq\curl{f_1,\ldots,f_s}$.\;
  $P\coloneqq\curl{(f_i,f_j)\,\middle\vert\,1\leq i<j\leq s}$.\;
  \While{$P\neq\emptyset$}{
    \textcolor{diffRed}{Choose a subset $L$ of $P$}\nllabel{ln:F4while}.\;
    \textcolor{diffRed}{$P\coloneqq P \setminus L$}.\;
    \textcolor{diffRed}{$L \coloneqq \operatorname{SymbolicPreprocessing}(L,\cG)$}.\;
    \textcolor{diffRed}{$L \coloneqq \operatorname{LinearAlgebra}(L)$}.\;
    \For{\textcolor{diffRed}{$h \in L$ with $\LM[\prec](h)\notin
        \ideal{\LM[\prec](\cG)}$}}{
      $P\coloneqq P \cup \curl{(g,h)\,\middle\vert\,g\in\cG}$.\;
      $\cG\coloneqq \cG\cup\curl{h}$\nllabel{ln:F4for}.
    }
  }
  \KwRet $\cG$.
\end{algorithm2e}
Observe that the termination of the \Fquatre algorithm only relies on
Buchberger's first criterion: $\cG=\curl{g_1,\ldots,g_t}$ is a \gb of
an ideal $I$ for $\prec$ if for all $1\leq i,j\leq t$,
$\nf{\spair{g_i}{g_j}}{\cG}{\prec}=0$, see~\cite[Chap.~2, Sec.~6, Th.~6]{CoxLOS2015}.

We detail the differences to Buchberger's algorithm.
\begin{enumerate}
\item One can choose several critical pairs at a time, stored in a 
  subset $L \subseteq P$. The so-called \emph{degree strategy}
  chooses $L$ to be the set of \emph{all} critical pairs of minimal degree
  for a total degree monomial order, typically $\ldrl$.

\item For all terms of all the generators of the S-polynomials, one
  searches in the current intermediate \gb $\cG$ for possible reducers.
  One adds those to $L$ and again search for all of their terms for
  reducers in $\cG$. This is the
  $\operatorname{SymbolicPreprocessing}$ function.

\item Once all reduction data is collected from the last step,
  one generates a Macaulay-like matrix with columns corresponding to the 
  monomials appearing in $L$ (sorted by $\prec$) and rows corresponding to
  the coefficients of each polynomial in $L$. Gaussian Elimination is then 
  applied to reduce now all chosen S-polynomials at once. This is the
  $\operatorname{LinearAlgebra}$ function.
  
\item Finally, one adds those polynomials associated to rows in the updated 
    matrix to $\cG$ 
    whose leading monomials  are not already in $\LM(\cG)$.
\end{enumerate}
In order to optimize the algorithm one can now apply Buchberger's product and
chain criteria, see~\cite{bGroebnerCriterion1979,KollreiderB1978}.
Thus many useless critical pairs are removed before even being added
to $P$ and fewer zero
rows are computed during the linear algebra part of \Fquatre. Still, in general,
there are many zero reductions left.

Different selection strategies yield different behavior of the
algorithm. The degree strategy allows one to compute
\emph{truncated \gbs} of ideals in case of early terminations.
\begin{definition}\label{def:truncatedgb}
  Let $f_1,\ldots,f_s$ be polynomials in $\K[\bx]$ and $\prec$ be a
  monomial order. Let $\mu$ be a monomial and $F_{\mu}$ be the
  $\K$-vector subspace of $\ideal{f_1,\ldots,f_s}$ defined as
  \[F_{\mu}=\curl{\sum_{i=1}^s h_i f_i\,\middle\vert\,\forall\,1\leq i\leq s,
      \LT(h_i f_i)\preceq\mu}.\]

  Then, $\cG\subset F_{\mu}$ is a $\mu$-truncated \gb of
  $\ideal{f_1,\ldots,f_s}$ for $\prec$ if for all $p\in F_{\mu}$,
  there exists $g\in\cG$ such that $\LM(g)\mid\LM(p)$ and
  $p-\frac{\LT(p)}{\LT(g)}g$ is in $F_{\mu}$.
\end{definition}
Observe that taking a triangular basis of $F_{\mu}$ ordered
increasingly \wrt $\prec$ naturally yields a $\mu$-truncated \gb thereof.
\begin{proposition}\label{prop:truncatedgb}
  Let $f_1,\ldots,f_s$ be polynomials in $\K[\bx]$ and $\prec$ be a
  monomial order. Let $\mu$ be a monomial and $F_{\mu}$ be the
  $\K$-vector subspace of $\ideal{f_1,\ldots,f_s}$% defined as
  \[F_{\mu}=\curl{
      \sum_{i=1}^s h_i f_i\,
      \middle\vert
      \,\forall\,1\leq i\leq s, \LM(h_i f_i)\preceq\mu
    }.\]
  A subset $\cG=\curl{g_1,\ldots,g_t}\subset F_{\mu}$ is a $\mu$-truncated \gb of
  $\ideal{f_1,\ldots,f_s}$ for $\prec$ if, and only if,
  \[F_{\mu}\subseteq G_{\mu}=\curl{
      \sum_{j=1}^t h_j g_j\,
      \middle\vert
      \,\forall\,1\leq j\leq t, \LM(h_j g_j)\preceq\mu
    }.\]
  and for all $(g_i,g_j)\in\cG^2$ with $i\neq j$, if
  $\lcm(\LM(g_i),\LM(g_j))\preceq\mu$, then
  \[\nf{\spair{g_i}{g_j}}{\cG}{\prec}=0.\]
\end{proposition}
\begin{proof}
  This proof follows the proof of~\cite[Chap.~2, Sec.~6, Th.~6]{CoxLOS2015}.
  
  If $\cG$ is a $\mu$-truncated \gb of $\ideal{f_1,\ldots,f_s}$ for $\prec$, then
  observe that both $F_{\mu}$ and $G_{\mu}$ only contain polynomials
  with leading monomial less or equal to $\mu$ for $\prec$. 
  Let
  $p\in F_{\mu}$, then there exists $g\in\cG$ such that
  $\LM(g)\mid\LM(p)$, $\LM(g)$ is maximal and
  $p-\frac{\LT(p)}{\LT(g)}g$ is in $F_{\mu}$. Thus, $p = h g + q$ with
  $\LM(q)\prec\LM(p)\preceq\mu$, $\LM(h g)=\LM(p)\preceq\mu$ and
  $q\in F_{\mu}$. Repeating this division process on $q$ shows that
  $p\in G_{\mu}$. As a consequence $p$ reduces
  to $0$ \wrt $G$ and $\prec$.

  Now, let $g_i$ and $g_j$ be in $G$ and $i\neq j$. Let
  $m=\lcm(\LM(g_i),\LM(g_j))$, then $\spair{g_i}{g_j}=
  \frac{m}{\LT(g_i)}g_i-\frac{m}{\LT(g_j)}g_j$. Furthermore, if 
  $m\preceq\mu$, then $\LM\pare{\frac{m}{\LT(g_i)}g_i}=m\preceq\mu$,
  and likewise for $g_j$. Hence, $\spair{g_i}{g_j}\in G_{\mu}$
  and its normal form \wrt $G$ and $\prec$ is $0$.

  For the converse, assume that $G_{\mu}$ contains $F_{\mu}$ and that
  for all $(g_i,g_j)\in\cG^2$ with $i\neq j$, if
  $\lcm(\LM(g_i),\LM(g_j))\preceq\mu$, then
  $\nf{\spair{g_i}{g_j}}{\cG}{\prec}=0$.

  Let $p\in F_{\mu}$, since $F_{\mu}\subseteq G_{\mu}$, there exist
  $h_1,\ldots,h_t$ such that $\LM(h_i g_i)\preceq\mu$ for all $1\leq
  i\leq t$ and $p=h_1 g_1 +\cdots+ h_t g_t$. Let
  $m=\max_{1\leq i\leq t}\LM(h_i g_i)\preceq\mu$, then observe
  that $\LM(p)\preceq m \preceq\mu$. Assume
  that among all the possible such writings of $p$, the polynomials
  $h_1,\ldots,h_t$
  are chosen so that $m$ is minimal for $\prec$. Such a minimal
  monomial exists by the well-order property of $\prec$.

  Now, if $\LM(p)=m=\LM(h_i g_i)=\LM(h_i)\LM(g_i)$ for some
  $1\leq i\leq t$,
  then $\LM(g_i)$ divides $\LM(f)$, hence $\LM(p)\in\LM(\cG)$.

  Otherwise, $\LM(p)\prec m$. We will use the fact that if the
  critical pair $(g_i,g_j)$ satisfies $\lcm(\LM(g_i),\LM(g_j))\preceq\mu$
  implies
  $\nf{\spair{g_i}{g_j}}{\cG}{\prec}=0$ to build a new expression of
  $p$ that decreases $m$.

  Let us write
  \[p=\sum_{\substack{1\leq i\leq t\\\LM(h_i g_i)=m}}h_i g_i
    +\sum_{\substack{1\leq i\leq t\\\LM(h_i g_i)\prec m}}h_i g_i,\]
  Then,
  \[p=\sum_{\substack{1\leq i\leq t\\\LM(h_i g_i)=m}} \LM(h_i) g_i
    +\sum_{\substack{1\leq i\leq t\\\LM(h_i g_i)=m}}(h_i-\LM(h_i))g_i
    +\sum_{\substack{1\leq i\leq t\\\LM(h_i g_i)\prec m}}h_i g_i,\]
  Since the second and third sums only involve monomials smaller than
  $m$ for $\prec$, then the leading monomial of the first one is also
  smaller than $m$ for $\prec$. Observe, on the one hand, that
  \[s_{i,j}=\spair{g_i\LM(h_i)}{g_j\LM(h_j)}
    =\spair{g_i}{g_j}\frac{m}{\LM(g_i)\LM(g_j)}.\]
  Now, on the other hand,
  $\LM(g_i\LM(h_i))=\LM(g_j\LM(h_j))=m$, hence their lcm is
  $m$. Therefore,
  $\lcm(\LM(g_i),\LM(g_j))\preceq m\preceq\mu$.
  By~\cite[Chap.~2, Sec.~6, Lemma~5]{CoxLOS2015}, the first sum in the
  latter expression of $p$ is a
  linear combination of the $s_{i,j}$'s and $\LM(s_{i,j})\prec m$ for
  all $1\leq i<j\leq t$.

  Consider, one of these polynomials
  $s_{i,j}$. Since
  $\lcm(\LM(g_i),\LM(g_j))\preceq\mu$,
  then we know that the critical pair $(g_i,g_j)$ is such that
  $\nf{\spair{g_i}{g_j}}{\cG}{\prec}=0$, hence
  $\nf{s_{i,j}}{\cG}{\prec}\allowbreak =0$ and there exist
  $A_1,\ldots,A_t\in\K[\bx]$ such that
  \[s_{i,j}=A_1 g_1+\cdots+A_t g_t,\]
  and for all $1\leq i\leq t$, $\LM(A_i g_i)\preceq\LM(s_{i,j})\prec
  m$.

  Doing this for all the polynomials $s_{i,j}$, we show that the first sum of the latter
  expression of $p$ can be replaced by $B_1 g_1+\cdots+B_t g_t$, where
  for all $1\leq i\leq t$,
  $\LM(B_i g_i)\prec m$. This contradicts the minimality of $m$ for
  this property, hence $\LM(p)=m$ and $\LM(p)\in\ideal{\LM(\cG)}$. By
  \Cref{def:truncatedgb}, $\cG$ is a $\mu$-truncated \gb of
  $\ideal{f_1,\ldots,f_s}$ for $\prec$.
\end{proof}
\begin{remark}
  \begin{enumerate}
  \item If $\prec$ is a total degree monomial order, then for
    $d\in\N$, we can also define a \emph{$d$-truncated \gb} as a $\mu$-truncated
    \gb for $\mu$ the largest monomial of degree $d$ for $\prec$.
  \item If $\cG=\curl{g_1,\ldots,g_t}$ is a $\mu$-truncated \gb of
    $\ideal{f_1,\ldots,f_s}$ for $\prec$
    and
    \[\mu\succeq\max_{1\leq i<j\leq t}\lcm(\LM(g_i),\LM(g_j)),\]
    then $\cG$ is a \gb of $\ideal{f_1,\ldots,f_s}$ for $\prec$. Indeed, it spans
    the ideal and by
    \Cref{prop:truncatedgb}, all the S-polynomials reduce
    to $0$ \wrt $\cG$ and $\prec$. Hence, by Buchberger's first
    criterion~\cite[Chap.~2, Sec.~6, Th.~6]{CoxLOS2015}, it is a \gb of
    $\ideal{f_1,\ldots,f_s}$ for $\prec$.
  \item \Cref{def:truncatedgb} depends greatly on the set of
    generators of the ideal. Consider $f_1=x^n$, $f_2=(y-1)^n$ and
    $f_3=x y - y - 1$ for $n\geq 2$. By
    \Cref{prop:truncatedgb}, $\cG=\curl{f_1,f_2,f_3}$ is a
    $n$-truncated \gb of $\ideal{f_1,f_2,f_3}$ for $\ldrl$. Yet, this
    ideal is $\ideal{1}$ hence $\curl{1}$ is a $n$-truncated \gb of
    $\ideal{1}$.
  \end{enumerate}
\end{remark}
\begin{lemma}\label{lem:F4trunc}
  Let $f_1,\ldots,f_s\in\K[\bx]$ be the input polynomials of the \Fquatre
  algorithm. Let $d\in\N$. Assume that the \Fquatre algorithm uses the degree selection
  strategy and that, on line~\ref{ln:F4while}, $L$
  consists in all the critical pairs of degree $d$.

  If no new polynomial is added to $\cG$ on line~\ref{ln:F4for}, then $\cG$ is a
  $d$-truncated \gb of $\ideal{f_1,\ldots,f_s}$.
\end{lemma}
\begin{proof}
  By the degree selection strategy, only critical pairs of degree at
  least $d$ exist. Since no new polynomial is added at the end of the
  turn, this means that all S-polynomials coming from critical pairs
  of degree $d$ reduce to $0$ \wrt $\cG$ and $\prec$. Furthermore,
  $\cG$ contains $f_1,\ldots,f_s$, thus, by
  \Cref{prop:truncatedgb}, $\cG$ is a $d$-truncated \gb of
  $\ideal{f_1,\ldots,f_s}$ for $\prec$.
  % \qed
\end{proof}

\subsubsection{The \spFGLM algorithm}\label{ss:spfglm}
In this subsection, the input \gb, $\cGdrl$, is the reduced \gb of a
zero-dimensional ideal $I$ of degree $D$ for $\ldrl$. The
output is the reduced \gb, $\cGlex$, of $I$ for $\llex$.
In~\cite{FaugereM2011} and~\cite[Algorithm~3]{FaugereM2017},
using~\cite{MorenoSocias2003}, the authors observe that the map
\begin{align*}
  \K[\bx]/I	&\to	\K[\bx]/I\\
  f		&\mapsto \nf{x_n f}{\cGdrl}{\ldrl}
\end{align*}
given in the basis $S_{\DRL}=\Staircase(\cGdrl)$ is represented by a
matrix, $M_{x_n}$, with a special structure given in the following two lemmas.
\begin{lemma}\label{lem:buildingmatrixzerodim}
  Let $I$ be a zero-dimensional ideal of $\K[\bx]$ of degree $D$, $\cGdrl$ be its
  reduced \gb for $\ldrl$ and $\Sdrl=\{\sigma_0,\ldots,\sigma_{D-1}\}$
  be its associated staircase.
  Let $M_{x_n}$ be the matrix of the linear map $f\in\K[\bx]/I\mapsto
  \nf{x_n f}{\cGdrl}{\ldrl}\in\K[\bx]/I$.
  
  Then,
  one can build the matrix 
  $M_{x_n}=(m_{i,j})_{0\leq i,j<D}$ with the following
  procedure:
  \begin{itemize}
  \item if $x_n\sigma_j=\sigma_k$, then $m_{k,j}=1$ and for all
    $0\leq i<D$, $i\neq k$, $m_{i,j}=0$;
  \item otherwise for all $0\leq i<D$, $m_{i,j}$ is the
    coefficient of $\sigma_i$ in $\nf{x_n\sigma_j}{\cGdrl}{\ldrl}$.
  \end{itemize}
\end{lemma}
\begin{proof}
  By construction, the matrix $M_{x_n}$ has its $j$th column which is the
  vector of coefficients of
  $\nf{x_n\sigma_j}{\cGdrl}{\ldrl}$ in the basis $\Sdrl$.

  The former case is immediate.
  
  The latter case is obtained by linearity.
\end{proof}
\begin{lemma}[{\cite{FaugereM2011, FaugereM2017}
    using~\cite{MorenoSocias2003}}]
  Let $f_1,\ldots,f_n$ be generic polynomials of $\K[x_1,\ldots,x_n]$
  of degrees at most $d$. Let $\cGdrl$ be the reduced \gb of
  $\ideal{f_1,\ldots,f_n}$ for $\ldrl$. Then, the latter case of
  \Cref{lem:buildingmatrixzerodim} only happens if there exists
  $g\in\cGdrl$ such that $\LM[\ldrl](g)=x_n\sigma_j$. As a
  consequence, one has $\nf{x_n\sigma_j}{\cGdrl}{\ldrl}=x_n\sigma_j-g$.
\end{lemma}
\begin{proof}
  By the genericity assumption on $f_1,\ldots,f_n$, the ideal
  $\ideal{f_1,\ldots,f_n}$ is complete
  intersection and zero-dimensional. Then, in~\cite{MorenoSocias2003}, a description of
  $\Sdrl$ is given in that case: if $\sigma\in\Sdrl$, then either
  $x_n\sigma\in\Sdrl$ or $x_n\sigma\in\LM[\ldrl](\cGdrl)$, \ie there
  exists $g\in\cGdrl$ such that $\LM[\ldrl](g)=x_n\sigma$.
\end{proof}
Following, we can use
Wiedemann algorithm~\cite{Wiedemann1986} on $M_{x_n}$ to recover its minimal
polynomial. Furthermore, whenever the reduced \gb $\cGlex$ for $\llex$ is in
\emph{shape position}, \ie there exist
$g_n,g_{n-1},\ldots,g_1\in\K[x_n]$ such that
\[\cGlex=\curl{\elim(x_n),x_{n-1}-\param{n-1}(x_n),\ldots,x_1-\param{1}(x_n)},\]
and for all $1\leq k\leq n-1$, $\deg\param{k}<\deg\elim$, then
$\param{1},\ldots,\param{n-1}$ can be computed by solving Hankel
systems of size $D$. This can be done using the following two
algorithms, \Cref{algo:seqw,algo:param}.
\begin{algorithm2e}[htbp!]
  \small
  \DontPrintSemicolon
  \caption{Sequences for \spFGLM}\label{algo:seqw}
  \KwIn{A matrix $M\in\K^{D\times D}$, a row-vector $\bmr\in\K^D$ and
    $n$ column-vectors $\bc_0,\bc_1,\ldots,\bc_{n-1}\in\K^D$ 
  }
  \KwOut{$(\bmr M\bc_0)_{0\leq i<2 D},
    (\bmr M\bc_1)_{0\leq i<D},\ldots,(\bmr M\bc_{n-1})_{0\leq i<D}$,
    with $\bone=(1,0,\ldots,0)^{\rT}$.
  }
  $\seqw_0^{(0)}\coloneqq\bmr\bc_0,\seqw_0^{(1)}\coloneqq\bmr\bc_1,\ldots,
  \seqw_0^{(n-1)}\coloneqq\bmr\bc_{n-1}$.\;
  \For{$i$ \KwFrom $1$ \KwTo $D-1$}{
    $\bmr\coloneqq\bmr M$.\;
    $\seqw_i^{(0)}\coloneqq\bmr\bc_0,\seqw_i^{(1)}\coloneqq\bmr\bc_1,\ldots,
    \seqw_i^{(n-1)}\coloneqq\bmr\bc_{n-1}$.\;
  }
  \For{$i$ \KwFrom $D$ \KwTo $2 D-1$}{
    $\bmr\coloneqq\bmr M$.\;
    $\seqw_i^{(0)}\coloneqq\bmr\bc_0$\;
  }
  \KwRet $(\seqw_i^{(0)})_{0\leq i<2 D},(\seqw_i^{(1)})_{0\leq i<D},
  \ldots,(\seqw_i^{(n-1)})_{0\leq i<D}$
\end{algorithm2e}
\begin{proposition}\label{prop:seqw}
  Let $M\in\K^{D\times D}$ be a matrix with $s$ nonzero \mbox{coefficients},
  $\bmr\in\K^D$ be a row-vector and
  $\bc_0,\bc_1,\ldots,\bc_{n-1}\in\K^D$ be $n$ column-vectors. Then,
  \Cref{algo:seqw} is correct and computes the sequences
  $(\bmr M\bc_0)_{0\leq i<2 D}$ and
  $(\bmr M\bc_k)_{0\leq i<D}$ for $1\leq k\leq n-1$
  in $O(s D + n D^2)$ operations in $\K$.

  Furthermore, if the vectors $\bc_0,\bc_1,\ldots,\bc_{n-1}$ are vectors of
  the canonical basis, then this complexity drops to $O((s+n) D)$.
\end{proposition}
\begin{proof}
  The termination and the correctness of the algorithm are
  immediate. It remains to prove its complexity.
  
  Each vector matrix product $\bmr M$ accounts for $O(s)$ operations in
  $\K$, hence computing them all requires $O(s D)$ operations.

  Then, we need to perform the scalar products $\bmr\bc_k$
  for $0\leq k\leq n-1$ at each step. Each one needs $O(D)$
  operations. Hence a total of 
  $O(n D^2)$ operations.

  Observe that if $\bc_0,\bc_1,\ldots,\bc_{n-1}$ are vectors of the
  canonical basis, then these scalar products need, each, $O(1)$
  operations, hence a total of $O(n D)$ operations.
  This concludes the proof.
\end{proof}
\begin{algorithm2e}[htbp!]
  \small
  \DontPrintSemicolon
  \caption{Hankel system solving for \spFGLM}\label{algo:param}
  \KwIn{Sequences
    $(\seqw_i^{(0)})_{0\leq i<2 D-1}$ and
    $(\seqw_i^{(k)})_{0\leq i<D}$ for $1\leq k\leq n-1$
    with coefficients in $\K$.
  }
  \KwOut{$\gamma_{0,k},\ldots,\gamma_{D-1,k}$ for $1\leq k\leq n-1$
    such that for all $0\leq i<D$,
    $\seqw_i^{(k)}=\gamma_{D-1,k}\seqw_{D-1+i}^{(0)}+\cdots+\gamma_{0,k}\seqw_i^{(0)}$.
  }
  \For{$k$ \KwFrom $1$ \KwTo $n-1$}{
    Solve the Hankel linear system
    \[
      \begin{pmatrix}
        \seqw_0^{(0)}	&\seqw_1^{(1)}	&\cdots	&\seqw_{D-1}^{(0)}\\
        \seqw_1^{(0)}	&\seqw_2^{(1)}	&\cdots	&\seqw_D^{(0)}\\
        \vdots		&\vdots		&	&\vdots\\
        \seqw_{D-1}^{(0)}	&\seqw_D^{(1)}	&\cdots	&\seqw_{2 D-2}^{(0)}\\
      \end{pmatrix}
      \begin{pmatrix}
        \gamma_{0,k}\\\gamma_{1,k}\\\vdots\\\gamma_{D-1,k}
      \end{pmatrix}
      =
      \begin{pmatrix}
        \seqw_0^{(k)}\\\seqw_1^{(k)}\\\vdots\\\seqw_{D-1}^{(k)}
      \end{pmatrix}.
    \]\nllabel{ln:hankel}
  }
  \KwRet $\gamma_{i,k}$ for $0\leq i<D$ and $1\leq k\leq n-1$.
\end{algorithm2e}
\begin{proposition}\label{prop:param}
  Let $(\seqw_i^{(0)})_{0\leq i<2 D-1}
  (\seqw_i^{(1)})_{0\leq i<D},\ldots,(\seqw_i^{(n-1)})_{0\leq i<D}$
  be sequences such that $(\seqw_i^{(0)})_{0\leq i<2 D-1}$ is linear
  recurrent of order $D$, then \Cref{algo:param} is correct
  and computes, for all $1\leq k\leq n-1$,
  $\gamma_{0,k},\ldots,\gamma_{D-1,k}$ such that
  \[\forall\,0\leq i<D,\
    \seqw_i^{(k)}=\gamma_{D-1,k}\seqw_{D-1+i}^{(0)}
    +\cdots+\gamma_{0,k}\seqw_i^{(0)}\]
  in $O(\sM(D)(n+\log D))$ operations, where
  $\sM(D)$ denote a cost function for multiplying univariate polynomials
  of degree $D$ with coefficients in $\K$.
\end{proposition}
\begin{proof}
  The termination of the algorithm is
  immediate. Since $(\seqw_i^{(0)})_{0\leq i<2 D-1}$ is linear
  recurrent of order $D$, then the Hankel matrix on line~\ref{ln:hankel} is
  invertible, see~\cite{BrentGY1980} or for instance the proof
  of~\cite[Th.~3.2]{BerthomieuS2022}, thus the algorithm is correct.

  Concerning the complexity, using~\cite{BrentGY1980}, we
  can compute a representation of this inverse in $O(\sM(D)\log D)$
  operations in $\K$. Then, multiplying this
  representation of this inverse with the right-hand side member of
  the equality requires $O(\sM(D))$ operations in $\K$.
  Hence a total of
  $O(\sM(D)(n+\log D))$ operations in $\K$.
\end{proof}
We are now in a position to present the \spFGLM algorithm in the shape
position case.
\begin{algorithm2e}[htbp!]
  \small
  \DontPrintSemicolon
  \caption{\spFGLM\label{algo:spfglm}}
  \KwIn{The reduced \gb $\cGdrl$ of a zero-dimensional ideal $I$
    for $\ldrl$ and its associated staircase
    $\Sdrl$ of size $D$. 
  }
  \KwOut{The reduced \gb of $I$ for $\llex$, if it is in shape position.
  }
  Build the matrix $M$ as in \Cref{lem:buildingmatrixzerodim}.\;
  Pick $\bmr\in\K^D$ a row-vector at random.\;
  $\bone\coloneqq(1,0,\ldots,0)^{\rT}$.\tcp*{the column-vector of
    coefficients of $\nf{1}{\cGdrl}{\ldrl}$}
  \For{$k$ \KwFrom $1$ \KwTo $n-1$}{
    Build $\bc_k$ the column-vector of coefficients of
    $\nf{x_k}{\cGdrl}{\ldrl}$.
  }
  Compute $(\seqw_i^{(0)})_{0\leq i<2 D},
  (\seqw_i^{(1)})_{0\leq i<D},\ldots,(\seqw_i^{(n-1)})_{0\leq i<D})$
  with \Cref{algo:seqw} called on
  $M,\bmr,\bone,
  \bc_1,\ldots,\bc_{n-1}$.\;
  $g_n\coloneqq\Berlekamp (\seqw_0^{(0)},\ldots,\seqw_{2 D-1}^{(0)})$.\;
  \lIf{$\deg g_n<D$}{\KwRet ``Not in shape position or bad vector''.
  }
  Compute
  $g_1\coloneqq\gamma_{D-1,1}x_n^{D-1}+\cdots+\gamma_{0,1},\ldots,
  g_{n-1}\coloneqq\gamma_{D-1,n-1}x_n^{D-1}+\cdots+\gamma_{0,n-1}$
  with \Cref{algo:param} called on $(\seqw_i^{(0)})_{0\leq i<2 D-1},
  (\seqw_i^{(1)})_{0\leq i<D},\ldots,(\seqw_i^{(n-1)})_{0\leq i<D})$.\nllabel{ln:spparam}\;
  \KwRet $\{g_n(x_n),x_{n-1}-g_{n-1}(x_n),\ldots,x_1-g_1(x_n)\}$.
\end{algorithm2e}
\begin{theorem}\label{th:spfglm}
  Let $I$ be a zero-dimensional ideal of $\K[\bx]$ of degree $D$, $\cGdrl$ be its
  reduced \gb for $\ldrl$ and $\Sdrl$ be its associated staircase. Let
  $M_{x_n}$ be 
  the matrix of the map
  $f\in\K[\bx]/I\mapsto\nf{x_n f}{\cGdrl}{\ldrl}\in\K[\bx]/I$ in the
  monomial basis $\Sdrl$.

  Let us assume that there are $t$ monomials $\sigma$ in
  $\Sdrl$ such that $x_n\sigma\in\LT[\ldrl](I)$ and that
  $x_1,\ldots,x_{n-1}\in\Sdrl$, that $M_{x_n}$ is
  known and that the reduced \gb $\cGlex$ of $I$ for $\llex$ is in
  shape position.
  Then, one can compute
  $\cGlex$ in $O(t D^2+n\sM(D))$
  operations, where
  $\sM(D)$ denote a cost function for multiplying univariate polynomials
  of degree $D$ with coefficients in $\K$.
\end{theorem}
\begin{proof}
  Taking the
  column-vector $\bone=\pare{1,0,\ldots,0}^{\rT}$ so that for all $i\in\N$,
  $M_{x_n}^i\bone$ is the vector of
  coefficients in $S_{\DRL}$ of $\nf {x_n^i}{\cGdrl}{\ldrl}$, we can
  pick at random a row-vector $\bmr$ to compute the sequence
  $\bseqw^{(0)}=(\seqw^{(0)}_i)_{0\leq i<2 D}=\pare{\bmr M_{x_n}^i\bone}_{0\leq i<2 D}$.
  Generically, the linear recurrence
  relation of minimal order
  satisfied by this 
  sequence
  \[\forall\,i\in\N,\ \seqw^{(0)}_{i+d}+c_{d-1}\seqw^{(0)}_{i+d-1}
    +\cdots+c_0\seqw^{(0)}_i=0,\]
  is such that $\elim=x_n^d+c_{d-1}x_n^{d-1}+\cdots+x_0$ is
  the minimal
  polynomial of $M_{x_n}$.

  Let us assume that $\cGlex$ is in shape position, then there exist
  $\gamma_{0,k},\ldots,\gamma_{D-1,k}^{D-1}$ in $\K$ such that
  $x_k-\gamma_{D-1,k}^{D-1}x_n^{D-1}-\cdots-\gamma_{0,k}=0$ in
  $\K[\bx]/I$. Since $M_{x_n}^d\bc_k$ is the vector of coefficients of
  $\nf{x_n^d x_k}{\cGdrl}{\ldrl}$, by
  multiplying on the left $M_{x_n}^d\bc_k$ by $\bmr M_{x_n}^i$ for all
  $0\leq d\leq D-1$, we obtain
  \[\forall\,0\leq i<D,\
    \seqw_i^{(k)}=\gamma_{D-1,k}\seqw^{(0)}_{D-1+i}+\cdots
    +\gamma_{0,k}\seqw^{(0)}_i.\]
  Hence the algorithm is correct and terminates. Observe that if $\deg g_n<D$, then
  $\cGlex$ is not in shape position and the algorithm is still correct
  to return the error message.
  
  It remains to prove the complexity statement. By assumption,
  $x_1,\ldots,x_{n-1}\in\Sdrl$, hence $1\in\Sdrl$ and
  $\bc_0,\bc_1,\ldots,\bc_{n-1}$ are vectors of the canonical basis.
  Moreover,
  there are $t$ monomials $\sigma$ in $\Sdrl$ such that
  $x_n\sigma\in\LT[\ldrl](I)$, hence $M_{x_n}$ has at most $t D +
  (D-t) = O(t D)$ nonzero coefficients. Observe that among these $t$
  monomials, there must be a pure power of each $x_k$, for $1\leq
  k\leq n-1$, which
  is not $1$, hence $t>n$.
  Therefore,
  by \Cref{prop:seqw}, the call to
  \Cref{algo:seqw} requires $O(t D^2 + n D)=O(t D^2)$ operations.

  Now, using fast variants~\cite{BrentGY1980} of the Berlekamp--Massey
  algorithm~\cite{Berlekamp1968,Massey1969}, one recover the minimal
  linear recurrence relation in $O(\sM(D) \log D)$
  operations. Finally, by \Cref{prop:param}, we can compute
  $g_1,\ldots,g_n$ in $O(\sM(D)(n+\log D))$ operations.

  All in all, we have a complexity in $O(t D^2 + n\sM(D))$ operations
  in $\K$.
\end{proof}
Note that the Berlekamp--Massey algorithm and its faster variants
return a factor of $\elim$, so if the computed polynomial has degree $D$,
\ie it is the characteristic polynomial of $M_{x_n}$, then it is also
its minimal polynomial. Furthermore,
based on a deterministic
variant of Wiedemann's algorithm, one can also provide a
deterministic variant of this algorithm to recover
$\elim$~\cite[Algorithm~4]{FaugereM2017}.

\begin{remark}
  In~\cite{BostanSS2003,Hyun2020163}, the authors consider the case
  where an ideal $J$ is not in shape position but its radical
  $\sqrt{J}$ is. Let us recall that
  $\sqrt{J}=\curl{f\in\K[\bx]\,\middle\vert\,\exists k\in\N, f^k\in J}$,
  see~\cite[Chap.~4, Sec.~2, Def.~4]{CoxLOS2015}.  In that case, the
  lexicographic \gb of $\sqrt{J}$ can be computed in a similar
  fashion, it suffices to replace the call to
  \Cref{algo:param} on line~\ref{ln:spparam} by a call
  to~\cite[Algorithm~2]{Hyun2020163}. 
\end{remark}

%%% Local Variables:
%%% mode: latex
%%% TeX-master: "new_saturation"
%%% End: 

\section{The \FquatreSAT algorithm for saturated ideals}
\label{s:F4SAT}
This section is devoted to the design of an algorithm which on input $f_1,
\ldots, f_s$ and $\varphi$ in $\K[\bx]$ computes a \gb of $\satid{I}{\vphi}$ for
a total degree monomial order $\prec$, typically $\ldrl$, where $I = \langle
f_1, \ldots, f_s \rangle$. As explained earlier, this algorithm modifies the
\Fquatre algorithm~\cite{Faugere1999}, with the degree
selection strategy, to discover on the fly polynomials in
$\satid{I}{\vphi}$ as early as possible during the computation. The use of the
$\ldrl$ order allows us to obtain these polynomials of lowest possible degree
early in the computation. 

\subsection{Description of the \FquatreSAT algorithm}
From \Cref{lem:F4trunc}, after the first step of the \Ffour
algorithm in degree $d$, if no new polynomial of degree at most $d$ is
discovered, then the
current \gb $\cG$ is a $d$-truncated \gb of $I$ for $\prec$. Therefore, we have a
partial information on the staircase of $I$, and thus of
$\satid{I}{\vphi}$,
for $\prec$ since we know
monomials that are outside of this staircase.
The \FquatreSAT
algorithm searches for polynomials in $\satid{I}{\vphi}$ whose
supports are entirely included in the given staircase using the fact that
$\colid{\pare{\satid{I}{\vphi}}}{\vphi}=\satid{I}{\vphi}$. If new
polynomials are found, they are added to $\cG$ and the necessary
critical pairs are added to the set of pairs to handle. Then, we resume
the \Fquatre algorithm.

The search of new polynomials is done through linear algebra
computations and is sum up in \Cref{algo:linalgcolon}. From a \gb of an ideal $J$,
$I\subseteq J\subseteq\satid{I}{\vphi}$, we
compute a bound $B$ on the degree of the polynomials in the reduced \gb
of $\satid{J}{\vphi}=\satid{I}{\vphi}$ using the $\ComputeMaxDegree$
routine based on \Cref{th:maxdeggb}.
Then, we compute
$\nf{\sigma\vphi}{\cG}{\prec}$ for all monomials $\sigma$ in the
associated staircase $S$ of degree at most
$B$. Finally, we search for 
vanishing linear combinations thereof. Indeed, if
\[\nf{s\vphi}{\cG}{\prec}
  -\sum_{\substack{\sigma\in S_d\\\sigma\prec s}}
  c_{\sigma}\nf{\sigma\vphi}{\cG}{\prec}
  = 0,\]
then $\pare{s-\sum_{\sigma\in S_d,\sigma\prec s}
  c_{\sigma}\sigma}\vphi\in J$.
This yields \Cref{algo:F4SAT}.
\begin{algorithm2e}[bhtp!]
  \small
  \DontPrintSemicolon
  \caption{\FquatreSAT\label{algo:F4SAT}}
  \KwIn{A list of polynomials $f_1,\ldots,f_s$ spanning an ideal
    $I\subseteq\K[\bx]$, a polynomial $\vphi\in\K[\bx]$ and a total
    degree monomial order $\prec$.}
  \KwOut{
    A \gb $\cG$ of $\satid{I}{\vphi}$ for $\prec$.
  }
  $\cG\coloneqq\curl{f_1,\ldots,f_s}$.\;
  $b\coloneqq\text{\KwTrue}$\tcp*{tracks if $\cG$ has changed}
  $P\coloneqq\curl{(f_i,f_j)\,\middle\vert\, 1\leq i<j\leq s}$\;
  \While{$P\neq\emptyset$\nllabel{ln:F4SATwhile}}{
    Choose a subset $L$ of $P$\nllabel{ln:F4SATwhileL}.\;
    $P\coloneqq P\setminus L$.\;
    $L \coloneqq \operatorname{SymbolicPreprocessing}(L,\cG)$\nllabel{ln:F4linalg}.\;
    $L \coloneqq \operatorname{LinearAlgebra}(L)$.\;
    \For{$h\in L$ with $\LM[\prec](h)\not\in\ideal{\LM[\prec](\cG)}$}{
      $P\coloneqq P\cup\curl{(g,h)\middle\vert g\in \cG}$,
      $\cG\coloneqq \cG\cup\curl{h}$\nllabel{ln:F4SATfor},
      $b\coloneqq\text{\KwTrue}$.\;
    }
    % \If(\tcp*[f]{New information on $\ideal{\LM(\satid{I}{\vphi})}$}){
    %   $\cG$ was augmented}
    \If(\tcp*[f]{new information on
      $\ideal{\LM(\satid{I}{\vphi})}$}){$b$}{
      $b\coloneqq\text{\KwFalse}$.\;
      % $B\coloneqq\ComputeMaxDegree
      % (\cG,\vphi)$.\nllabel{ln:degreebound}
      $B\coloneqq\max_{\ell\in L}\deg \ell$
      \tcp*{bounds the degrees of the new polynomials}
      $\cH\coloneqq\LinearizeColonIdeal(\cG,B,\prec)$\;
      \Foreach{$h\in\cH$}{
        $P\coloneqq P\cup\curl{(g,h)\,\middle\vert\,g\in\cG}$,
        $\cG\coloneqq \cG\cup\curl{h}$\nllabel{ln:F4SATfor2},
        $b\coloneqq\text{\KwTrue}$.\;
      }
    }
  }
  $B\coloneqq\ComputeMaxDegree (\cG,\vphi,\prec)$.\nllabel{ln:degreebound}
  \tcp*{bounds the degrees in the sought \gb}
  $\cH\coloneqq\LinearizeColonIdeal(\cG,B,\prec)$\;
  \If(\tcp*[f]{The ideal $\ideal{\cG}$ is not saturated}){$\cH\neq\emptyset$}{
    \Foreach{$h\in\cH$}{
      $P\coloneqq P\cup\curl{(g,h)\,\middle\vert\,g\in\cG}$,
      $\cG\coloneqq \cG\cup\curl{h}$\nllabel{ln:F4SATfor3},
      $b\coloneqq\text{\KwTrue}$.\;
    }
    \KwGoto line~\ref{ln:F4SATwhile}\tcp*{resumes the whole \gb computation}
  }
  \KwRet $\cG$.
\end{algorithm2e}
\begin{algorithm2e}[bhtp!]
  \small
  \DontPrintSemicolon
  \caption{$\LinearizeColonIdeal$\label{algo:linalgcolon}}
  \KwIn{A list $\cG$ of polynomials in $\K[\bx]$, a polynomial
    $\vphi$, a degree bound $B$
    and a total degree monomial order $\prec$.}
  \KwOut{
    An auto-reduced family of polynomials $\cH$ such that for all
    $h\in\cH$, $\deg h<B$ and $\nf{h\vphi}{\cG}{\prec}=0$.
  }
  \For{$\sigma\not\in\LM[\prec](\cG)$ \KwAnd
    $\deg\sigma\leq B$}{
    $q_{\sigma}\coloneqq\nf{\sigma\vphi}{\cG}{\prec}$.\nllabel{ln:nfmultipliers}
    \nllabel{ln:F4SATnf}\;
  }
  Build the matrix $M$ whose rows are given by polynomials $q_{\sigma}$ and columns by
  each monomials in their support in decreasing order.\;
  Compute a lower triangular basis $K$ of the left-kernel of
  $M$.\nllabel{ln:F4SATnew}\;
  $\cH\coloneqq\emptyset$\;
  \Foreach{$k\in K$}{
    $\cH\coloneqq\cH\cup\curl{\sum_{\sigma\not\in\ideal{\LM[\prec](\cG)}}
      k_{\sigma}\sigma}$.\tcp*{the polynomial whose vector of coefficients is $k$}
  }
  \KwRet $\cH$.
\end{algorithm2e}

\begin{theorem}\label{th:F4SAT}
  Let $f_1,\ldots,f_s$ be a generating family of an ideal
  $I\subseteq\K[\bx]$, $\vphi\in\K[\bx]$ be a polynomial and $\prec$
  be a total degree monomial order. Then,
  \Cref{algo:F4SAT} terminates and returns a \gb of
  $\satid{I}{\vphi}$ for $\prec$.
\end{theorem}

\subsection{Proof of termination and correctness}
\added{
  We start by proving that the algorithm returns a \gb.
}
\added{
\begin{lemma}\label{lem:gb}
  Let $\cG$ be the set of polynomials in line~\ref{ln:degreebound} of
  \Cref{algo:F4SAT} called for $\prec$. Then, $\cG$ is a \gb for $\prec$.
\end{lemma}
}
\added{
\begin{proof}
  When reaching line~\ref{ln:degreebound}, the set $P$ of critical
  pairs of polynomials in $\cG$ is empty. Thus, by Buchberger's first
  criterion~\cite[Chap.~2, Sec.~6, Th.~6]{CoxLOS2015}, $\cG$ is a \gb.
\end{proof}
}

\added{
\begin{theorem}\label{th:gb}
  Let $\cG$ be the output of \Cref{algo:F4SAT} called for
  $\prec$. Then, $\cG$ is a \gb for $\prec$.
\end{theorem}
}
\begin{proof}
  \added{
    By \Cref{lem:gb}, at line~\ref{ln:degreebound}, $\cG$ is a \gb for $\prec$. If,
    at the following line, $\cH$ is empty, then $\cG$ is not changed and
    the algorithm outputs $\cG$, which is a \gb for $\prec$.
  }

  \added{
    Otherwise, $\cG$ is augmented, $P$ is augmented and thus not empty
    and the algorithm returns to line~\ref{ln:F4SATwhile}.
  }

  \added{
    All in all, the algorithm can only reach the return instruction is
    $P$ is empty, hence $\cG$ is a \gb for $\prec$, and $\cH$ is empty
    as well, hence $\cG$ is not augmented. This concludes the proof that
    the returned $\cG$ is a \gb for $\prec$.
  }
\end{proof}

\begin{lemma}\label{lem:staircase}
  Let $I$ and $J$ be two ideals of $\K[\bx]$ such that $I\subseteq J$ and $\cG$
  and $\cH$ be their respective reduced \gbs for a common monomial order
  $\prec$. Let $S$ and $T$ be the associated staircases to $\cG$ and $\cH$.
  Then, $T\subseteq S$. Furthermore, there exist $h_1,\ldots,h_r$, such that for
  all $i$, $\supp h_i\subseteq S$ and $J=I+\ideal{h_1,\ldots,h_r}$.
\end{lemma}
\begin{proof}
  Let $\cG=\curl{g_1,\ldots,g_t}$.
  By definition of a \gb,
  for all $f\in I$, there exists $1\leq i\leq r$ such that 
  $\LM(g_i)\mid\LM(f)$. Since $I\subseteq J$, then $f\in J$
  and there also exists
  $h\in\cH$ such that $\LM(h)\mid\LM(f)$, hence
  $\LM(I)\subseteq\LM(J)$. By definition, $S$ (\resp $T$) is the complement of
  $\LM(I)$ (\resp $\LM(J)$) in the set of monomials, hence $T\subseteq S$.

  Since $I\subseteq J$, there exist $f_1,\ldots,f_r$ such that
  $J=I+\ideal{f_1,\ldots,f_r}$. Thus,
  $J=\ideal{g_1,\ldots,g_t,f_1,\ldots,f_r}$.
  By the definition of
  $\cG$ being a \gb of $I$ for $\prec$, for all $1\leq j\leq r$, we have 
  $f_j=q_{j,1} g_1+\cdots+q_{j,t} g_t+\nf{f_j}{\cG}{\prec}\in
  J$. Since $\nf{f_j}{\cG}{\prec}$ has no monomial divisible by
  $\LM(g)$, for $g\in\cG$, its support is a subset of $S$. Thus,
  taking $h_j=\nf{f_j}{\cG}{\prec}$, we have $\supp h_j\in S$
  and
  $J=I+\ideal{h_1,\ldots,h_r}$.
\end{proof}

We will apply \Cref{lem:staircase} with $J=\colid{I}{\vphi}$, thus, by
definition of $\colid{I}{\vphi}$, we also know that for all $1\leq j\leq r$,
$h_j\vphi$ is in $I$. Moreover, $h_j$ is a polynomial whose support is
in $S$, thus
it can be written as $h_j=s-\sum_{\sigma\in S,\sigma\prec s}c_{\sigma}\sigma$,
with $s \in S$ and $c_{\sigma}$'s in $\K$. Since $h_j\vphi\in I$, we know that
\[\nf{h_j\vphi}{\cG}{\prec}
  =0 \quad\text{and}\quad
    \nf{s\vphi}{\cG}{\prec}
  =\sum_{\substack{\sigma\in S\\\sigma\prec s}}
  c_{\sigma}\nf{\sigma\vphi}{\cG}{\prec}.
\]

From \Cref{lem:staircase}, we deduce a superset of the support of
the polynomials in the reduced \gb of $\satid{I}{\vphi}$ for $\prec$
from the staircase $S$ of $I$ for $\prec$. Yet, if $S$ is not finite,
it is not clear up to which degree we need to search for
polynomials in $\satid{I}{\vphi}$, \ie the bound $B$ on
line~\ref{ln:degreebound}. This is given by the following routine:
$\ComputeMaxDegree$.

\begin{algorithm2e}[bhtp!]
  \small
  \DontPrintSemicolon
  \caption{$\ComputeMaxDegree$\label{algo:maxdeg}}
  \KwIn{A \gb \added{$\cG$} of an ideal
    $I\subseteq\K[\bx]$, a polynomial $\vphi\in\K[\bx]$ and a total
    degree monomial order $\prec$.}
  \KwOut{
    A bound on the degree of the polynomials in the reduced \added{\gb
    of} $\satid{I}{\vphi}$ for $\prec$.
  }
  $\psi\coloneqq\nf{\vphi}{\cG}{\prec}$.\;
  $d\coloneqq\deg\psi$.\;
  % $J^{\h}\coloneqq\ideal{\LM(\cG)}+\ideal{x_{n+1}^d}\subseteq\K[x_0,\bx,x_{n+1}]$.\;
  % Compute $S$ and $P$, the Hilbert series and polynomial of
  % $J^{\h}$.\;
  % Compute $r$, the index of regularity, and $m_0$ for
  % $J^{\h}$ using $S$ and $P$.\tcp*{see
  %   Proposition~\ref{prop:regindex}}
  \added{$I'\coloneqq\ideal{\LM(\cG)}+\ideal{x_{n+1}^d}\subseteq\K[\bx,x_{n+1}]$}.\;
  Compute $S$ and $P$, the Hilbert series and polynomial of
  \added{$I'$}.\;
  Compute $r$, the index of regularity, and $m_0$ for
  \added{$I'$} using $S$ and $P$.\tcp*{see
    \Cref{prop:regindex}}
  \KwRet $\max (r,m_0)$.
\end{algorithm2e}

\added{
  \begin{remark}
    In generic situations, such as those encountered in \Cref{s:implem},
    the required maximum degree is much smaller than the one returned by
    \Cref{algo:maxdeg}. Furthermore, in \Cref{ss:bettermaxdeg}, we
    provide variants
    % of \Cref{algo:maxdeg},
    thereof,
    which also return better
    degree bounds, under different assumptions on $I$ or $I'$.
  \end{remark}
}

In order to prove the correctness of $\ComputeMaxDegree$, or
\Cref{algo:maxdeg}, we start by proving the following lemma,
corresponding to the special case $\vphi=x_n$.

\begin{lemma}\label{lem:F4SATCrit}
  Let $\prec$ be a total degree monomial order.
  Let $\curl{g_1,\ldots,g_r}$ be a \gb
  of an ideal $I\subseteq\K[\bx]$, for $\prec$. Let $B$ be the
  bound defined in \Cref{th:maxdeggb} for
  $I^{\h}=\ideal{f^{\h}\,\middle\vert\,f\in I}$, the
  homogenization of $I$.

  Then, \added{any polynomial
    in the reduced \gb of 
    $\satid{I}{x_n}$ for $\prec$ has a degree
    which is bounded from above by
    $B$.}
\end{lemma}
\begin{proof}
  By~\cite[Chap.~8, Sec.~4, Proof of Th.~4]{CoxLOS2015}, homogenizing
  $g_1,\ldots,g_r$ with variable $x_0$ yields a 
  homogeneous \gb $\cG_1=\curl{g_1^{\h},\ldots,g_r^{\h}}$ for
  $\prec$ with $x_0\prec x_n\prec\cdots\prec x_1$ of the homogeneous ideal
  $I^{\h}$. From $\cG_1$, we
  compute the Hilbert series $\HS_{\K[x_0,\bx]/I^{\h}}(t)$ of
  $\K[x_0,\bx]/I^{\h}$, using~\cite[Chap.~10, Sec.~2,
  Th.~6]{CoxLOS2015} and thus the bound $B$.
  % Thus, the Hilbert series 
  % $\HS_{\K[x_0,\bx]/I^{\h}}(t)$ of $\K[x_0,\bx]/I^{\h}$ equals 
  % the Hilbert series $\HS_{\K[x_0,\bx]/I}(t)$ of 
  % $\K[x_0,\bx]/I$ divided by $1-t$:
  % \[\HS_{\K[x_0,\bx]/I^{\h}}(t) = \frac{\HS_{\K[x_0,\bx]/I}(t)}{1-t}
  %   =\HS_{\K[x_0,\bx]/I}(t)\sum_{i\geq 0}t^i.\]

  Let $\prec_2$ be a total degree monomial order such that
  $x_n\prec_2\cdots\prec_2 x_1\prec_2 x_0$ and let
  $\cG_2=\curl{\tilde{g}_1^{\h},\ldots,\tilde{g}_s^{\h}}$ be a \gb of
  $I^{\h}$ for $\prec_2$.
  % By~\cite[Chap.~10, Sec.~2, Prop.~8]{CoxLOS2015},
  By \Cref{cor:sameHS},
  the Hilbert
  series of the quotient ring $\K[x_0,\bx]/I^{\h}$ only depends on
  $I^{\h}$ and not on the chosen total degree monomial order. Hence
  the bound $B$ also applies to the degree of the polynomials in $\cG_2$:
  \[\max_{1\leq j\leq s} \deg \tilde{g}_j^{\h}\leq B.\]

  Finally, using Bayer's algorithm~\cite[p.~120]{Bayer1982}, we compute
  a \gb of $\satid{I^{\h}}{x_n}$ for $\prec_2$ from $\cG_2$ as follows: for each
  $g\in\cG_2$, find the largest integer $k$ such that $x_n^k$ divides
  $g$ and
  take $\frac{g}{x_n^k}$. Then, we  obtain a \gb
  of $\satid{I}{x_n}$ for $\prec_2$ by
  dehomogenizing the resulting polynomials, \ie by setting $x_0$ to
  $1$. Thus \added{any polynomial 
    in this \gb has degree at most
    $\max_{1\leq j\leq s}\deg\tilde{g}_j^{\h}\leq B$.
  }
\end{proof}
\begin{theorem}\label{th:F4SATCrit}
  Let $\prec$ be a total degree monomial order.
  Let $\cG=\curl{g_1,\ldots,g_r}$ be a \gb
  of an ideal $I\subseteq\K[\bx]$ for $\prec$.
  Let
  $\vphi\in\K[\bx]$.

  Then,
  \Cref{algo:maxdeg} returns a bound on the degree of the
  polynomials in the reduced \gb of $\satid{I}{\vphi}$ for $\prec$.
\end{theorem}
\begin{proof}
  By the definition of a \gb, there exist polynomials $q_1,\ldots,q_r$
  such that
  $\vphi=q_1 g_1+\cdots+q_r g_r + \psi$ and $\psi=\nf{\vphi}{\cG}{\prec}$. Then,
  a polynomial
  $h$ is in $\colid{I}{\vphi}$ if, and only if, $h \vphi$ is in
  $I$. Thus, this is
  equivalent to requiring that $h\psi$ is in
  $I$.
  In other words,
  $\colid{I}{\vphi}=\colid{I}{\psi}$ and $\satid{I}{\vphi}=\satid{I}{\psi}$.
  
  Now, let us denote
  $g_{r+1}=x_{n+1}^d-\psi$, where $x_{n+1}$ is a new indeterminate
  and $d=\deg\psi$. Then, its
  leading monomial for $\prec$ with $x_n\prec\cdots\prec x_1\prec x_{n+1}$ is
  $x_{n+1}^d$. Since $g_1,\ldots,g_r$ do not involve
  $x_{n+1}$, their leading monomials are exactly the same as those for
  $\prec$ with $x_n\prec\cdots\prec x_1$. By Buchberger's second
  criterion~\cite[Chap.~2, Sec.~9, Prop.~4]{CoxLOS2015}, adding
  $g_{r+1}$ to $\cG$ does not create new critical pairs. Since $\cG$
  is already a \gb of $I$ for $\prec$ with $x_n\prec\cdots\prec x_1$,
  by Buchberger's first criterion~\cite[Chap.~2, Sec.~6, Th.~6]{CoxLOS2015},   
  $\cH=\curl{g_1,\ldots,g_r,g_{r+1}}$ is also a \gb of
  $\cI=I+\ideal{x_{n+1}^d-\psi}$ for $\prec$ with
  $x_n\prec\cdots\prec x_1\prec x_{n+1}$.

  Now, let $I^{\h}$ and $\cI^{\h}$ be the respective
  homogenizations of $I$ and $\cI$ with variable $x_0$. Since
  $\LM(\cI^{\h})=\ideal{\LM(I^{\h})}+\ideal{\LM(x_{n+1}^d-\psi)}
  =\ideal{\LM(\cG)}+\ideal{x_{n+1}^d}=\ideal{\LM(\cG\cup\{x_{n+1}^d\})}$,
  these polynomials allow us to compute the Hilbert series and
  polynomial of $\cI^{\h}$. By \Cref{prop:regindex}, from
  the Hilbert series and the
  Hilbert polynomial, we
  compute $r$, the index of regularity, and $m_0$ for $\cI^{\h}$.

  Furthermore, saturating $I$ by $\psi$ is
  equivalent to saturating $\cI$ by $\psi$ (\resp $x_{n+1}^d$, \resp
  $x_{n+1}$) and then eliminating $x_{n+1}$.
  Thus, using \Cref{lem:F4SATCrit} on the ideal
  $\cI$, the \gb $\cH$, and the
  monomial order
  $\prec$ with $x_n\prec\cdots\prec x_1\prec x_{n+1}$, we
  deduce that \added{any polynomial in the \gb $\cH'$ of
    $\satid{\cI}{\vphi}=\satid{\cI}{x_{n+1}}$ has
    degree at most
    $B$. Now, the bound $B$ also applies to the degree of the polynomials
    in $\cH'$ for a monomial order eliminating
    $x_{n+1}$ on $\satid{\cI}{\vphi}$. Finally, setting $x_0$ to $1$
    also yields polynomials of degree at most $B$.
  }
\end{proof}

\begin{remark}
  Knowing the Hilbert series and polynomial of $I^{\h}$
  allows us to directly compute those of $\cI^{\h}$. Indeed,
  by~\cite[Chap.~10, Sec.~2, Ex.~4.a]{CoxLOS2015}, we have
  \begin{align*}
    \HS_{\K[x_0,\bx,x_{n+1}]/\cI^{\h}}(t)
    &=\HS_{\K[x_0,\bx]/I^{\h}}(t) \HS_{\K[x_{n+1}]/\ideal{x_{n+1}^d}}(t)\\
    &=\HS_{\K[x_0,\bx]/I^{\h}}(t) \frac{1-t^d}{1-t}\\
    &=\HS_{\K[x_0,\bx]/I^{\h}}(t)\pare{1+t+\cdots+t^{d-1}}.
  \end{align*}
  Hence, for all $s\geq d-1$,
  $\HP_{\cI^{\h}}(s)=\HP_{I^{\h}}(s)
  +\cdots+\HP_{I^{\h}}(s-d+1)$. Let $r_{I^{\h}}$ (\resp $r_{\cI^{\h}}$) be the index
  of regularity of $I^{\h}$ (\resp $\cI^{\h}$). Since for all $s\geq r_{I^{\h}}$,
  $\HP_{I^{\h}}(s)=\dim_{\K}\K[x_0,\bx]_s/I^{\h}_s$, we have for all
  $s\geq r_{I^{\h}}+d-1$,
  $\HP_{\cI^{\h}}(s)=\dim_{\K}\K[x_0,\bx,x_{n+1}]_s/\cI^{\h}_s$. In
  other words, $r_{\cI^{\h}}\leq r_{I^{\h}}+d-1$. Assuming $I^{\h}$ has dimension
  $\delta$, then so does
  $\cI^{\h}$ and by \Cref{prop:regindex}, there exist
  integers $m_{0,I^{\h}},\ldots,m_{\delta-1,I^{\h}}$ such that
  \[\HP_{I^{\h}}(s)=\sum_{i=0}^{\delta-1}\pare{\binom{s+i}{i+1}
      -\binom{s+i-m_{i,I^{\h}}}{i+1}}.\]
  
  Therefore,
  \begin{align*}
    \HP_{\cI^{\h}}(s)=\sum_{j=0}^{d-1}
    \sum_{i=0}^{\delta-1}\pare{\binom{s-j+i}{i+1}-\binom{s-j+i-m_{i,I^{\h}}}{i+1}}.
  \end{align*}
  Now, it suffices to compute the decomposition, by algebraic manipulation,
  \[\HP_{\cI^{\h}}(s)=\sum_{i=0}^{\delta-1}\pare{\binom{s+i}{i+1}
      -\binom{s+i-m_{i,\cI^{\h}}}{i+1}},\]
  in order to ensure that any polynomial in a \gb of $\cI^{\h}$ has
  degree at most $B=\max(r_{\cI^{\h}},m_{0,\cI^{\h}})$.
\end{remark}

\begin{lemma}\label{lem:F4SATtrunc}
  Let $f_1,\ldots,f_s\in\K[\bx]$ be the input polynomials of the \FquatreSAT
  algorithm. Let $d\in\N$. Assume that the \FquatreSAT algorithm
  uses the degree selection 
  strategy and that, on line~\ref{ln:F4SATwhileL}, $L$
  consists in all the critical pairs of degree $d$.
  
  If no new polynomial is added to $\cG$ on both lines~\ref{ln:F4SATfor}
  and~\ref{ln:F4SATfor2}, then $\cG$ is a
  $d$-truncated \gb of $\colid{\ideal{f_1,\ldots,f_s}}{\vphi}$.
\end{lemma}
\begin{proof}
  The first part of the proof follows the one of \Cref{lem:F4trunc}.
  By the degree selection strategy, only critical pairs of degree at
  least $d$ exist. Since no new polynomial is added at the end of the
  first for loop, this means that all S-polynomials coming from critical pairs
  of degree $d$ reduce to $0$ \wrt $\cG$ and $\prec$. Furthermore,
  $\cG$ contains $f_1,\ldots,f_s$, thus, by
  \Cref{prop:truncatedgb}, $\cG$ is a $d$-truncated \gb of
  $\ideal{f_1,\ldots,f_s}$ for $\prec$.
  
  Now, for $\mu$ the largest monomial of degree
  $d$, seeing $F_{\mu}$ of \Cref{def:truncatedgb} as a
  vector subspace of $\colid{\ideal{f_1,\ldots,f_s}}{\vphi}$ instead
  of $\ideal{f_1,\ldots,f_s}$, by \Cref{prop:truncatedgb}
  and the fact that no polynomial is added at the end of the second
  for loop, we have that $\cG$ is also a $d$-truncated \gb of
  $\colid{\ideal{f_1,\ldots,f_s}}{\vphi}$.
\end{proof}

We are now in a position to prove \Cref{th:F4SAT}.
\begin{proof}\hspace*{-2mm} \emph{of \Cref{th:F4SAT}}
  % Looks strange but it's the best way of getting it done correctly.
  At the first round of the while loop, $\cG$ contains a generating
  family of $I\subseteq \satid{I}{\vphi}$.

  Let us assume that at each round of the while loop, $\cG$ starts by
  containing a generating family of an ideal $J\subseteq
  \satid{I}{\vphi}$. Then, at the end of the round, it contains a
  generating family of an ideal $K$ with $J\subseteq K\subseteq
  \colid{J}{\vphi}$. Since $J\subseteq \satid{I}{\vphi}$, then
  $K\subseteq\colid{J}{\vphi}\subseteq
  \colid{(\satid{I}{\vphi})}{\vphi}=\satid{I}{\vphi}$.
  
  Thus, by recurrence and the fact that
  $\K[\bx]$ is Noetherian, this sequence of ideals must stabilize
  to an ideal that is included in $\satid{I}{\vphi}$
  so that the loop terminates. Furthermore, by
  % Buchberger's first
  % criterion~\cite[Chap.~2, Sec.~6, Th.~6]{CoxLOS2015},
  \Cref{th:gb},
  the output
  family is a \gb for $\prec$ of the ideal it spans.
  
  By the correctness of the \Fquatre algorithm, the \FquatreSAT algorithm
  computes a \gb for $\prec$ of an ideal $J$ containing $I$. Moreover,
  by \Cref{th:F4SATCrit}, we know that if $\cG$ is 
  a \gb, then the given bound is enough to retrieve a \gb of
  $\satid{\ideal{\cG}}{\vphi}$ and thus of
  $\colid{\ideal{\cG}}{\vphi}$.
  Furthermore, at each
  round of the loop, if $\cG$ has been enlarged, then the
  algorithm looks for new polynomials in the saturated
  ideal.
  Otherwise, by \Cref{lem:F4SATtrunc}, $\cG$ is a $d$-truncated
  \gb of $J=\colid{\ideal{\cG}}{\vphi}$, which satisfies $I\subseteq
  J\subseteq \colid{J}{\vphi}\subseteq\satid{I}{\vphi}$.

  Therefore, the algorithm can only terminate
  if $J$ is saturated by $\vphi$, that is $\colid{J}{\vphi}=J$.

  Since $\satid{I}{\vphi}$ is the smallest ideal containing $I$ and
  saturated by $\vphi$, we conclude that $J=\satid{I}{\vphi}$.
\end{proof}

% \begin{theorem}
%   Let $f_1,\ldots,f_s$ be a generating family of an ideal
%   $I\subseteq\K[\bx]$, $\vphi\in\K[\bx]$ be a polynomial and $\prec$
%   be a total degree monomial order. Then,
%   \Cref{algo:F4SAT} terminates and returns a \gb of
%   $\satid{I}{\vphi}$ for $\prec$.
% \end{theorem}

\subsection{\added{Better degree bounds}}\label{ss:bettermaxdeg}
\added{
  In this section, we provide variants of \Cref{algo:maxdeg} that return
  better degree bounds, assuming the ideals satisfy some properties.
  We start with the case when $\satid{I}{\vphi}$ is
  $0$-dimensional. Then, 
  relying on \Cref{th:maxdeggbgeneric}, we provide better degree bounds for
  the polynomials that are computed by \Cref{algo:F4SAT} in some generic
  situations.
}

\added{
  \begin{theorem}
    Let $\cG$ be the \gb for $\ldrl$ in line~\ref{ln:degreebound} of
    \Cref{algo:F4SAT}. Assume that the ideal $\ideal{\cG}$ has dimension $0$.
    Then, on line~\ref{ln:degreebound} of \Cref{algo:F4SAT}, one can
    call the variant \Cref{algo:maxdegdim0} of $\ComputeMaxDegree$.
  \end{theorem}
}
\added{
  \begin{proof}
    By the proof of \Cref{th:F4SAT}, the ideal
    $J=\ideal{\cG}$ satisfies $I\subseteq
    J\subseteq\satid{I}{\vphi}$, so that
    $\satid{J}{\vphi}=\satid{I}{\vphi}$.
    Thus,
    applying \Cref{lem:staircase} on $J$ and $\satid{I}{\vphi}$, we
    want to compute some polynomials $h_1,\ldots,h_r$ whose supports
    are in the staircase associated to $\cG$ such that
    $\satid{I}{\vphi} = \ideal{\cG} + \ideal{h_1,\ldots,h_r}$. Since $J$ is
    $0$-dimensional, this staircase is finite and it suffices to
    consider all the monomials in this staircase. This can be done by
    taking as the degree bound, the maximum degree of the monomials in
    this staircase.
  \end{proof}
}

\begin{algorithm2e}[bhtp!]
  \small
  \DontPrintSemicolon
  \caption{$\ComputeMaxDegree$ in dimension $0$\label{algo:maxdegdim0}}
  \KwIn{A \gb $\cG$ of a $0$-dimensional ideal
    $I\subseteq\K[\bx]$, a polynomial $\vphi\in\K[\bx]$ and a total
    degree monomial order $\prec$.}
  \KwOut{
    A bound on the degree of the polynomials in the reduced \added{\gb
      of} $\satid{I}{\vphi}$ for $\prec$.
  }
  Build the staircase $S$ associated to $\cG$ for $\prec$.\;
  \KwRet $\max_{\sigma\in S}\deg \sigma$.
\end{algorithm2e}
\added{
  Recall that testing if $\cG$ spans a $0$-dimensional ideal
  is easy, it suffices to check that for each variable there is a
  polynomial whose leading monomial is a pure power thereof.
}

% \begin{lemma}
%   \added{
%     Let $I\subseteq\K[\bx]$ be an ideal and let $\vphi$ be a polynomial
%     of degree $d$.
%     Let $\cG$ be the reduced \gb of $I+\ideal{x_{n+1}^d-\vphi}$ for
%     DRL.
%     Let $B$ be the maximum degree of the polynomials in
%     $\cG$ and let $\cH$ be the reduced \gb of $\sat{I}{\vphi}$. Then,
%     any polynomial in $\cH$ has degree at most $B$.
%   }
% \end{lemma}
% \begin{proof}
%   Following the proof of \Cref{lem:F4SATCrit}, $B$ is a bound on the
%   degree of the polynomials in the reduced \gb of the homogenization
%   of $I+\ideal{x_{n+1}^d-\vphi}$ for DRL.
% \end{proof}

\added{
  \begin{definition}\label{def:Noetherpos}
    Let $I$ be an ideal of $\K[\bx]$ of dimension $\delta$. $I$ is in
    \emph{Noether position} if $x_{n-\delta+1},\ldots,x_n$ are free in $\K[\bx]/I$.
  \end{definition}
  Observe that this property can be easily tested: it suffices to
  compute a \gb of $I$ for a monomial order $\prec$ such that
  $x_n\prec\cdots\prec x_1$ and to check that no leading monomial is
  purely in $x_{n-\delta+1},\ldots,x_n$.
}

\added{
  \begin{lemma}\label{lem:CMmaxdegree}
    Let $I\subseteq\K[\bx]$ be an ideal and $\vphi$ be a polynomial of
    degree $d$. Assume that $\cI=I+\ideal{x_{n+1}^d-\vphi}$ is
    Cohen-Macaulay or has dimension at most $1$ and that it is in
    Noether position. Let $\cG$ be the reduced \gb of this ideal for $\ldrl$ and
    $r$ be its index of regularity. Then, any polynomial in the
    reduced \gb of $\satid{\cI}{x_{n+1}}$ for $\ldrl$ has degree at most $r$.
  \end{lemma}
}
\begin{proof}
  \added{
    By \Cref{th:maxdeggbgeneric}, the index of regularity $r$ of
    $\cI=I+\ideal{x_{n+1}^d-\vphi}$ bounds the degree of the polynomials
    in $\cG$ for $\ldrl$ and thus of its
    homogenization for $\ldrl$ with variable $x_0$ as the smallest.
  }
  \added{
    By the Noether position assumption, $r$ also bounds the
    degree of the \gb of the homogenization for $\ldrl$ where $x_0$ is
    now the largest variable and $x_{n+1}$ the smallest. Recall that the
    saturation process consists now in dividing the polynomials of
    this \gb by $x_{n+1}$ as much as possible and then in
    dehomogenizing them. Therefore, the degrees of the polynomials can
    only decrease and
    thus remain bounded from above by $r$.
  }
\end{proof}

\added{
  We are now in a position to design another variant of
  \Cref{algo:maxdeg}, namely \Cref{algo:maxdegdim1}, when the input
  \gb spans a $1$-dimensional
  ideal in Noether position. Recall that testing if $\cG$ spans such
  a $1$-dimensional ideal
  in Noether position is easy. It suffices to check that no polynomial
  in $\cG$ has for leading monomial a pure power of $x_n$ and that no
  pair of variables is free, i.e. for each $1\leq i<j\leq n$,
  there exist $k,\ell$ such that $x_i^k x_j^{\ell}$ is divisible by a
  leading monomial of $\cG$.
}
\begin{algorithm2e}[bhtp!]
  \small
  \DontPrintSemicolon
  \caption{$\ComputeMaxDegree$ in dimension $1$ in Noether position
    \label{algo:maxdegdim1}}
  \KwIn{A \gb $\cG$ of a $1$-dimensional ideal
    $I\subseteq\K[\bx]$ in Noether position,
    a polynomial $\vphi\in\K[\bx]$ and a total
    degree monomial order $\prec$.}
  \KwOut{
    A bound on the degree of the polynomials in the reduced \added{\gb
      of} $\satid{I}{\vphi}$ for $\prec$.
  }
  $\psi\coloneqq\nf{\vphi}{\cG}{\prec}$.\;
  $d\coloneqq\deg\psi$.\;
  % $J^{\h}\coloneqq\ideal{\LM(\cG)}+\ideal{x_{n+1}^d}\subseteq\K[x_0,\bx,x_{n+1}]$.\;
  % Compute $S$ and $P$, the Hilbert series and polynomial of
  % $J^{\h}$.\;
  % Compute $r$, the index of regularity, and $m_0$ for
  % $J^{\h}$ using $S$ and $P$.\tcp*{see
  %   Proposition~\ref{prop:regindex}}
  \added{$I'\coloneqq\ideal{\LM(\cG)}+\ideal{x_{n+1}^d}\subseteq\K[\bx,x_{n+1}]$}.\;
  Compute $S$ and $P$, the Hilbert series and polynomial of
  \added{$I'$}.\;
  Compute $r$, the index of regularity using $S$ and $P$.\tcp*{see
    \Cref{prop:regindex}}
  \KwRet $r$.
\end{algorithm2e}

\added{
  \begin{theorem}\label{th:F3SATdim1}
    Let $\cG$ be the \gb for $\ldrl$ in line~\ref{ln:degreebound} of
    \Cref{algo:F4SAT} and let $\vphi$ be a polynomial of degree
    $d$. Assume that the ideal
    $\cJ=\ideal{\cG}+\ideal{x_{n+1}^d-\vphi}$ is in Noether position
    and furthermore either that it is Cohen-Macaulay or that it has
    dimension $1$.
    Then, on line~\ref{ln:degreebound} of \Cref{algo:F4SAT}, one
    can call the variant \Cref{algo:maxdegdim1} of $\ComputeMaxDegree$.
  \end{theorem}
}
\begin{proof}
  \added{
    By \Cref{lem:CMmaxdegree}, $r$ bounds the degree of any polynomial
    in the reduced \gb of $\satid{\cJ}{x_{n+1}}$ for $\ldrl$. Now,
    mimicking the proof of \Cref{th:F4SATCrit}, with $B=r$ as the
    degree bound on the polynomials in the reduced \gb of
    $\satid{\cJ}{x_{n+1}}=\satid{\cJ}{\vphi}$ for $\ldrl$, we have that any
    polynomial in the reduced \gb of
    $\satid{J}{\vphi}=\satid{I}{\vphi}$ for $\ldrl$ has degree at most $r$.
  }
\end{proof}

\added{
  \begin{remark}
    Empirically, in the experiments of \Cref{s:implem}, the first time
    \Cref{algo:F4SAT} computes a \gb $\cG$ of an ideal $J$ for $\ldrl$, with
    $I\subseteq J\subseteq\satid{I}{\vphi}$, $J$ is actually
    $\satid{I}{\vphi}$. Therefore, the saturation step using
    $\LinearizeColonIdeal$ is useless.
  \end{remark}
}
% \noindent \Cref{th:F3SATdim1} can even be strengthened whenever the computed
% ideal has dimension $0$.

\subsection{Practical optimization}
\label{ss:F4SAToptim}
\paragraph{\added{Early termination.}}
As we shall see in \Cref{s:implem}, the most expensive step of
\FquatreSAT is the last saturation step, that is checking on
line~\ref{ln:F4SATnew} that no new
polynomial in the saturated ideal can be formed thanks to the
monomials in the discovered staircase of the computed \gb $\cG$. To bypass
this, whenever we detect that $\satid{I}{\vphi}$ is
zero-dimensional, we rely on
the following trick
to determine if we 
have computed a \gb of $\satid{I}{\vphi}$: From the geometric point of
view, the variety defined by $\satid{I}{\vphi}$ is the Zariski closure
of the variety defined by $I$ to which the variety defined
by $\vphi$ is removed. If this resulting variety is a finite set of points, then
none of them can lie on the hypersurface defined by $\vphi$. Thus the
intersection of this set of points and this hypersurface is
empty. Algebraically, this means that $\pare{\satid{I}{\vphi}} +
\ideal{\vphi}=\ideal{1}$. Hence, we add $\vphi$ to $\cG$
and run
the \Ffour algorithm. If the output
is indeed $1$, then the saturation has already been computed.

However, if the current computed \gb spans a positive-dimensional
ideal, this trick does not work. By \Cref{th:maxdeggb,th:F4SATCrit}, we may
derive a bound on the degree of the sought polynomials in the reduced \gb of
$\satid{I}{\vphi}$ which ensures the termination of
\FquatreSAT. Yet, as this bound might be highly pessimistic, for
efficiency issues, it is crucial to provide 
an early termination criterion. To do so, we rely on Rabinowitsch
trick~\cite{Rabinowitsch1930} and~\cite[Chap.~4, Sec.~4, Th.~14,
(ii)]{CoxLOS2015}. Whenever Buchberger's criterion ensures us that
$\cG$ is a \gb, we check that it spans $\satid{I}{\vphi}$ by
computing a \gb of $\ideal{\cG}+\ideal{1-t\vphi}$ for a monomial
order eliminating $t$.

\paragraph{\added{Delayed kernel computation.}}
Furthermore, the search of new polynomials in the saturated ideal need
not be performed as soon as possible. As an optimization, we can
decide to perform it after a given number of steps of the \Fquatre
algorithm, so that the new information on the staircase increases the 
probability to find new polynomials in the saturated ideal.
\added{For instance, in \Cref{s:implem},
  we search for new elements in the saturation only if \Fquatre has
  added new elements to the basis in three distinct linear algebra steps.}

Furthermore, if we target
specifically small degree polynomials, we can require to only compute
the $q_{\sigma}$ for small degree $\sigma$'s compared to the degrees of
the polynomials in $\cG$ on line~\ref{ln:F4SATnf}. Then, when no new
polynomials are found and the 
set of critical pairs is empty, we can compute all the $q_{\sigma}$ to
ensure the correctness of the algorithm and the output.
\added{As such, in \Cref{s:implem}, $q_{\sigma}$ is computed if $\deg\sigma$ is at
  most $2/3$ of the maximal degree in the current basis $\cG$, in order to speed
  intermediate steps of the algorithm up. Moreover, when the
  set of critical pairs is empty, all $q_{\sigma}$ up to the maximal degree of the
  basis are taken into account.
}

\paragraph{\added{Tracer over $\Q$.}}
When the base field is the field of rational numbers, a practical
efficient implementation of the \FquatreSAT 
algorithm requires a multi-modular approach. Like for the \Fquatre algorithm,
we can use a tracing algorithm~\cite{msolve} where we \emph{learn}
from the first modular steps
and \emph{apply} optimal computations in the following modular steps.

In contrast to \Fquatre we cannot learn all information needed for optimal runs 
of \FquatreSAT in the first modular run:
Observe that on line~\ref{ln:nfmultipliers}, the normal form
$q_{\sigma}$ can be
computed iteratively to increase the usage of already pre-reduced data from 
lower degrees: If $q_{\sigma}$ was computed in a
previous turn, we reduce it \wrt to the new $\cG$ and $\prec$. Though,
modulo the first prime $p_1$, we cannot learn how these iterative reductions are
performed, since there might be useless saturation steps. Since these are 
skipped in the following modular steps, we cannot predict, during the 
computation modulo $p_1$, the normal form computations modulo other primes. However, these reductions can be learned from
the computations modulo $p_2$, the second modular step. Here we know exactly when 
we apply useful saturation steps, thus the normal form computations and its 
information stabilizes. This might also have an impact on the overall \Fquatre 
computation, thus we can only learn when to apply useful saturation steps modulo 
$p_1$. Only in the second modular step can we learn all the information for the 
complete computation.

The \emph{tracing} of \FquatreSAT can now be described by three main steps:
\begin{enumerate}
    \item In the \emph{first} modular computation \emph{learn}
    \begin{enumerate}
        \item when to apply a useful saturation step and
        \item which \Fquatre matrices give new information for the basis.
    \end{enumerate}
Since we cannot learn anything further for the \Fquatre computation we can apply 
probabilistic linear algebra to accelerate this step.
\item In the \emph{second} modular computation \emph{learn}
    \begin{enumerate}
        \item all polynomial data that is needed in the \Fquatre matrices to 
            generate the non-zero information \wrt each corresponding matrix 
            and
        \item all polynomial data needed to apply the normal form computations 
            in the useful saturation steps.
    \end{enumerate}
In order to learn this data we have to apply exact linear algebra.
\item For all \emph{successive} modular computations
    we can just \emph{apply} the learned data, no need of handling critical pairs or 
symbolic preprocessing. There will be no reductions to zero, all computational 
steps will be useful from now on.
\end{enumerate}
Note that the last saturation step, which just ensures us that
the ideal is saturated and so does not provide any new polynomial, should
only be run in the first modular 
computation.

In \Cref{s:implem}, the first learning phase, modulo $p_1$, is
denoted by \emph{learn~1}, the second one, modulo $p_2$, is denoted by
\emph{learn~2}. The apply phase, modulo $p_3,\ldots$, is denoted by \emph{apply}.

%%% Local Variables:
%%% mode: latex
%%% TeX-master: "new_saturation"
%%% End:

\section{Change of order algorithm for colon ideals}
\label{s:SpFGLMcol}
In this section, let $I$ be an ideal of $\K[\bx]$, let $\cGdrl$ be its
reduced \gb for $\ldrl$ and $\Sdrl$ be its associated staircase and
let $\varphi$ be
a polynomial. We assume that the colon ideal $\colid{I}{\vphi}$ is zero-dimensional,
thus $\vphi\notin I$,
and that its lexicographic reduced \gb $\cHlex$ is in \emph{shape
  position}:
There exist $h_1,\ldots,h_n\in\K[x_n]$,
with $\deg h_k<\deg h_n$ for $1\leq k\leq n-1$,
such that
\[\cHlex=\curl{h_n(x_n),x_{n-1}-h_{n-1}(x_n),\ldots,x_1-h_1(x_n)}.\]
We design a new algorithm to compute $\cHlex$ from $\cGdrl$, even
when $I$ is positive-dimensional. Our approach is to build a matrix $\tiM_{x_n}$
so that applying Wiedemann's algorithm allows us to recover
$\cHlex$,
similarly to the \spFGLM
algorithm~\cite{FaugereM2011,FaugereM2017}.

Firstly, we discuss the situation if $I$ is
zero-dimensional~(\Cref{ss:zerodim}). Next, we handle the
case when $I$ is
positive-dimensional~(\Cref{ss:posdim}).
\Cref{ss:matrix} focuses on the construction of the matrix
$\tiM_{x_n}$, followed by possible optimizations in \Cref{ss:optim}.
Finally, in
\Cref{ss:nonshape}, we discuss how to handle the situation if the assumption
that $\cHlex$ is in shape position is dropped.

\subsection{The case where $I$ is zero-dimensional ideal}
\label{ss:zerodim}
Let $I$ be zero-dimensional of degree $D$.
Thus, $\colid{I}{\vphi}$ is also zero-dimensional, of
degree $D'\leq D$.  Let $\Sdrl$ be the staircase of $I$
associated to $\cGdrl$ and
let $\bphi$ be the vector of coefficients of $\nf{\vphi}{\cGdrl}{\ldrl}$
in the basis given by $\Sdrl$. Further, let $M_{x_n}$ be the matrix of the map
$f\in\K[\bx]/I\mapsto \nf{x_n f}{\cGdrl}{\ldrl}\in\K[\bx]/I$ in the
basis $\Sdrl$.
\begin{lemma}\label{lem:elimcolon}
  Let $I$ be a zero-dimensional ideal of $\K[\bx]$ of degree $D$, $\cGdrl$ be its
  reduced \gb for $\ldrl$ and $\Sdrl$ be the associated staircase. Let
  $\vphi\in\K[\bx]\setminus I$ be such that $\colid{I}{\vphi}$ is in
  shape position and let
  $\cHlex=\{h_n(x_n),x_{n-1}-h_{n-1}(x_n),\ldots,x_1-h_1(x_n)\}$ be
  the reduced \gb of $\colid{I}{\vphi}$ for $\llex$, with $\deg h_n=D'$.

  Let $M_{x_n}$ be the matrix of the map
  $f\in\K[\bx]/I\mapsto\nf{x_n f}{\cGdrl}{\ldrl}\in\K[\bx]/I$ in the
  basis $\Sdrl$ and $\bphi$ be the vector of coefficients of $\vphi$
  in the basis $\Sdrl$.

  Then, for all $i\in\N$,
  $M_{x_n}^i\bphi$ is the vector of coefficients of
  $\nf{x_n^i\vphi}{\cGdrl}{\ldrl}$.

  Furthermore, let $\bmr\in\Kbar^D$ be a row-vector and
  $\bseqw=(\seqw_i)_{i\in\N} =(\bmr M_{x_n}^i\bphi)_{i\in\N}$. Let $d\in\N$ be
  minimal such that there exist $c_0,\ldots,c_{d-1}\in\K$ such that
  \[\forall i\in\N,\quad
    \seqw_{i+d}+c_{d-1}\seqw_{i+d-1}+\cdots+c_0\seqw_i=0.\]
  Then, the polynomial $x_n^d+c_{d-1} x_n^{d-1}+\cdots+c_0$ divides
  $h_n$. Furthermore, if $\bmr$ is generic in $\Kbar^D$, then $d=D'$ and
  $h_n=x_n^d+c_{d-1}x_n^{d-1}+\cdots+c_0$.
\end{lemma}
\begin{proof}
  Since $\bphi$ is the vector of coefficients of $\vphi$ in
  $\K[x_n]/I$, then by construction of $M_{x_n}$, $M_{x_n}\bphi$ is
  the vector of coefficients of $\nf{x_n\vphi}{\cGdrl}{\ldrl}$ in
  $\K[x_n]/I$. By recurrence, we obtain the first statement.

  Now, since $\K[\bx]/I$ is finite-dimensional, there exists a smallest
  integer $b$ such that $\bphi, M_{x_n}\bphi,\ldots,M_{x_n}^b\bphi$ are
  not linearly independent. We let $a_0,\ldots,a_{b-1}\in\K$ such that
  \[M_{x_n}^b\bphi+a_{b-1}M_{x_n}^{b-1}\bphi+\cdots+a_0\bphi=0.\]
  Thus,
  $\nf{\pare{x_n^b+a_{b-1}x_n^{b-1}+\cdots+a_0}\vphi}{\cGdrl}{\ldrl}=0$
  and the polynomial $(x_n^b+a_{b-1}x_n^{b-1}+\cdots+a_0)\vphi$
  is in $I$. Hence,
  $\pare{x_n^b+a_{b-1}x_n^{b-1}+\cdots+a_0}\in\colid{I}{\vphi}$.

  By the minimality of $b$, this ensures that
  $h_n=x_n^b+a_{b-1}x_n^{b-1}+\cdots+a_0$.

  Now, multiplying the vector equality above by $\bmr M_{x_n}^i$ on
  the left yields
  \[\forall\,i\in\N,\ \seqw_{i+b}+a_{b-1}\seqw_{i+b-1}+\cdots+a_0\seqw_i=0.\]
  Thus, $\bseqw$ is linearly recurrent of order at most $b$ and $d\leq
  b$. Since linear recurrences are in one-to-one correspondence with
  polynomials, these polynomials define an ideal of $\K[x_n]$ spanned
  by $x_n^d+c_{d-1}x_n^{d-1}+\cdots+c_0$ that contains $h_n$. Hence
  the former divides the latter.

  Following
  the proof of
  Wiedemann's algorithm~\cite{Wiedemann1986}, it suffices to take $\bmr$
  outside finitely many vector subspaces of $\bar{\K}^D$ to recover
  the minimal polynomial of $M_{x_n}$ instead of a proper factor
  thereof. Thus, for $\bmr$ generic, we actually compute $h_n$.
\end{proof}

\begin{lemma}\label{lem:paramcolon}
  Let $I$ be a zero-dimensional ideal of $\K[\bx]$ of degree $D$, $\cGdrl$ be its
  reduced \gb for $\ldrl$ and $\Sdrl$ be the associated staircase.  Let
  $\vphi\in\K[\bx]\setminus I$ be such that $\colid{I}{\vphi}$ is in
  shape position and let
  $\cHlex=\{h_n(x_n),x_{n-1}-h_{n-1}(x_n),\ldots,x_1-h_1(x_n)\}$ be
  the reduced \gb of $\colid{I}{\vphi}$ for $\llex$.

  Let $M_{x_n}$ be the matrix of the map
  $f\in\K[\bx]/I\mapsto\nf{x_n f}{\cGdrl}{\ldrl}\in\K[\bx]/I$ in the
  basis $\Sdrl$ and $\bphi$ be the vector of coefficients of $\vphi$
  in the basis $\Sdrl$.

  Let $\bmr\in\Kbar^D$ be a generic row-vector and
  $\bseqw=(\seqw_i)_{i\in\N} =(\bmr M_{x_n}^i\bphi)_{i\in\N}$. Let $d\in\N$ be
  minimal such that there exist $c_0,\ldots,c_{d-1}\in\K$ such that
  \[\forall i\in\N,\quad
    \seqw_{i+d}+c_{d-1}\seqw_{i+d-1}+\cdots+c_0\seqw_i=0,\]
  then $d=D'$ and $h_n=x_n^d+c_{d-1}x_n^{d-1}+\cdots+c_0$.

  For all $1\leq k\leq n-1$, let $\bpsi_k$
  be the vector of coefficients of
  $\nf{x_k\vphi}{\cGdrl}{\ldrl}$, then $M_{x_n}^i\bpsi_k$ is the
  vector of coefficients of
  $\nf{x_n^ix_k\vphi}{\cGdrl}{\ldrl}$ and
  there exist unique $ \gamma_{0,k},\ldots,\gamma_{D'-1,k}\in\K$ such that
  \[
    \begin{pmatrix}
      \seqw_0	&\seqw_1	&\cdots	&\seqw_{D'-1}\\
      \seqw_1	&\seqw_2	&\cdots	&\seqw_{D'}\\
      \vdots	&\vdots		&	&\vdots\\
      \seqw_{D'-1}	&\seqw_{D'}	&\cdots	&\seqw_{2 D'-2}\\
    \end{pmatrix}
    \begin{pmatrix}
      \gamma_{0,k}\\\gamma_{1,k}\\\vdots\\\gamma_{D'-1,k}
    \end{pmatrix}
    =
    \begin{pmatrix}
      \bmr M_{x_n}^0 \bpsi_k\\\bmr M_{x_n}^1 \bpsi_k\\\vdots\\
      \bmr M_{x_n}^{D'-1} \bpsi_k\\
    \end{pmatrix},
  \]
  and $h_k=\gamma_{D'-1,k}x_n^{D'-1}+\cdots+\gamma_{0,k}$.
\end{lemma}
\begin{proof}
  The proof of the first statement is a direct consequence of the
  definition of $M_{x_n}$.

  Now, since $\bmr$ is generic, then by \Cref{lem:paramcolon},
  $d=D'$ and $\bseqw$ does not satisfy any linear recurrence
  relation of order less than $D'$. Thus, there is no vector in the
  kernel of the above Hankel matrix, see~\cite{BrentGY1980}.

  Let $1\leq k\leq n-1$ and $h_k(x_n)=\alpha_{D'-1,k}x_n^{D'-1}+\cdots+\alpha_{0,k}$, then
  $(x_k-h_k(x_n))\vphi$ is in $I$ and
  $\nf{\pare{x_k-h_k(x_n)}\vphi}{\cGdrl}{\ldrl}=0$,
  hence
  \[\bpsi_k=\alpha_{D'-1,k}M_{x_n}^{D'-1}\bphi+\cdots+\alpha_{0,k}\bphi.\]
  Now, multiplying this equality by $\bmr M_{x_n}^i$ for $0\leq i\leq
  D'-1$ shows that $(\alpha_0,\ldots,\alpha_{D'-1})^{\rT}$ is a solution
  of the above linear system. Since the matrix has full
  rank, the solution is unique and this ends the proof.
  % \qed
\end{proof}

From this, we deduce the following algorithm.
\begin{algorithm2e}[htbp!]
  \small
  \DontPrintSemicolon
  \caption{\spFGLM for colon ideals\label{algo:spfglmcolonzero}}
  \KwIn{The reduced \gb $\cGdrl$ of a zero-dimensional ideal $I$
    for $\ldrl$, its associated staircase $\Sdrl$ of size $D$ and a
    polynomial $\vphi\in\K[\bx]$.
  }
  \KwOut{The reduced \gb of $\colid{I}{\vphi}$ for $\llex$, if it is in shape position.
  }
  Build the matrix $M$ as in \Cref{lem:buildingmatrixzerodim}.\;
  Pick $\bmr\in\K^D$ a row-vector at random.\;
  Build $\bphi$ the column-vector of coefficients of
  $\nf{\vphi}{\cGdrl}{\ldrl}$.\;
  \For{$k$ \KwFrom $1$ \KwTo $n-1$}{
    Build $\bpsi_k$ the column-vector of coefficients of
    $\nf{x_k\vphi}{\cGdrl}{\ldrl}$.
  }
  Compute $(\seqw_i^{(0)})_{0\leq i<2 D},
  (\seqw_i^{(1)})_{0\leq i<D},\ldots,(\seqw_i^{(n-1)})_{0\leq i<D})$
  with \Cref{algo:seqw} called on
  $M,\bmr,\bphi,
  \bpsi_1,\ldots,\bpsi_{n-1}$\nllabel{ln:spcol0seqw}.\;
  $h_n\leftarrow\Berlekamp (\seqw_0^{(0)},\ldots,\seqw_{2
    D-1}^{(0)})$, $D'\coloneqq\deg h_n$\nllabel{ln:spcol0bm}.\;
  \lIf{$\nf{h_n\vphi}{\cGdrl}{\ldrl}\neq 0$}{\KwRet ``Bad vector''.
  }
  Compute
  $h_1\coloneqq\gamma_{D'-1,1}x_n^{N-1}+\cdots+\gamma_{0,1},\ldots,
  h_{n-1}\coloneqq\gamma_{D'-1,n-1}x_n^{N-1}+\cdots+\gamma_{0,n-1}$
  with \Cref{algo:param} called on $(\seqw_i^{(0)})_{0\leq i<2 D'-1},
  (\seqw_i^{(1)})_{0\leq i<D'},\ldots,(\seqw_i^{(n-1)})_{0\leq
    i<D'})$.\nllabel{ln:spcol0param}\;
  \For{$k$ \KwFrom $1$ \KwTo $n-1$}{
    \lIf{$\nf{(x_k-h_k(x_n))\vphi}{\cGdrl}{\ldrl}\neq 0$}{\KwRet ``Not in
      shape position''.
    }
  }
  \KwRet $\{h_n(x_n),x_{n-1}-h_{n-1}(x_n),\ldots,x_1-h_1(x_n)\}$.
\end{algorithm2e}
\begin{remark}\label{rk:spfglmzerocomp}
  Observe that line~\ref{ln:spcol0seqw} of
  \Cref{algo:spfglmcolonzero} can lead to a large computational overhead
  whenever $D$ is much larger than $D'$. This is the bottleneck of the algorithm.

  Mixing lines~\ref{ln:spcol0seqw} and~\ref{ln:spcol0bm} so that the minimal
  linear recurrence relation is computed online during the
  computation of the terms $\seqw^{(0)}_i$ ensures that only $O(D')$
  of them are computed.
  Then, we can also compute only $O(D')$ terms
  $\seqw^{(k)}_i$ for each $1\leq k\leq n-1$.
\end{remark}
\begin{theorem}\label{th:spfglmcolonzero}
  Let $I$ be a zero-dimensional ideal of $\K[\bx]$ of degree $D$, $\cGdrl$ be its
  reduced \gb for $\ldrl$ and $\Sdrl$ be its associated staircase. Let
  $M_{x_n}$ be
  the matrix of the map
  $f\in\K[\bx]/I\mapsto\nf{x_n f}{\cGdrl}{\ldrl}\in\K[\bx]/I$ in the
  monomial basis $\Sdrl$. Let $\vphi$ be a polynomial not in $I$, so
  that $\colid{I}{\vphi}$ is zero-dimensional of degree $D'$.

  Let us assume that there are $t$ monomials $\sigma$ in
  $\Sdrl$ such that $x_n\sigma\in\LT[\ldrl](I)$, that $M_{x_n}$ is
  known and that the reduced \gb $\cHlex$ of $I$ for $\llex$ is in
  shape position.
  Then, one can compute
  $\cHlex$ by computing $n$ normal forms \wrt $\cGdrl$ and $\ldrl$ and
  $O((t+n) D D')$
  operations in $\K$.
\end{theorem}
\begin{proof}
  The proof follows the proof of \Cref{th:spfglm}.

  Taking the
  column-vector $\bphi$ so that for all $i\in\N$,
  $M_{x_n}^i\bphi$ is the vector of
  coefficients in $S_{\DRL}$ of $\nf {x_n^i}{\cGdrl}{\ldrl}$, we can
  pick at random a row-vector $\bmr$ to compute the sequence
  $\bseqw=(\seqw^{(0)}_i)_{0\leq i<2 D}=\pare{\bmr M_{x_n}^i\bone}_{0\leq i<2 D}$.
  Generically, the linear recurrence
  relation of minimal order
  satisfied by this
  sequence
  \[\forall\,i\in\N,\ \seqw^{(0)}_{i+d}+c_{d-1}\seqw^{(0)}_{i+d-1}
    +\cdots+c_0\seqw^{(0)}_i=0,\]
  is such that $h_n=x_n^d+c_{d-1}x_n^{d-1}+\cdots+x_0$ is
  the minimal
  polynomial of $x_n$ in the quotient algebra
  $\K[\bx]/(\colid{I}{\vphi})$, hence $d=D'$.

  Let us assume that $\cHlex$ is in shape position, then there exist
  $\gamma_{0,k},\ldots,\gamma_{D'-1,k}$ in $\K$ such that
  $x_k-\gamma_{D'-1,k}x_n^{D'-1}-\cdots-\gamma_{0,k}=0$ in
  $\K[\bx]/(\colid{I}{\vphi})$. Since $M_{x_n}^j\bc_k$ is the vector of coefficients of
  $\nf{x_n^j x_k}{\cGdrl}{\ldrl}$, by
  multiplying on the left $M_{x_n}^j\bc_k$ by $\bmr M_{x_n}^i$ for all
  $0\leq j\leq D'-1$, we obtain
  \[\forall\,0\leq i<D,\
    \seqw_i^{(k)}=\gamma_{D'-1,k}\seqw^{(0)}_{D'-1+i}+\cdots
    +\gamma_{0,k}\seqw^{(0)}_i.\]
  Hence the algorithm is correct and terminates. Observe that if
  $h_n\vphi$ is not in $I$, then $h_n$ is not correctly computed and
  the algorithm correctly
  returns an error message. Likewise, if $\cHlex$ is not in shape
  position, one of the computed polynomial $x_k-h_k(x_n)$ is not in
  $\colid{I}{\vphi}$, hence multiplied by $\vphi$, it is not in $I$
  and the algorithm correctly returns an error message.

  By assumption, there are $t$ monomials $\sigma$ in $\Sdrl$ such that
  $x_n\sigma\in\LT[\ldrl](I)$, hence $M_{x_n}$ has at most $t D +
  (D-t) = O(t D)$ nonzero coefficients.
  Therefore,
  by \Cref{prop:seqw}, the call to
  \Cref{algo:seqw} requires $O((t+n) D^2)$ operations.

  Now, using fast variants~\cite{BrentGY1980} of the Berlekamp--Massey
  algorithm~\cite{Berlekamp1968,Massey1969}, one recover the minimal
  linear recurrence relation in $O(\sM(D) \log D)$
  operations. Finally, by \Cref{prop:param}, we can compute
  $h_1,\ldots,h_n$ in $O(\sM(D)(n+\log D))$ operations.

  All in all, we have a complexity in $O((t+n) D^2)$ operations
  in $\K$.

  Using the modification of \Cref{rk:spfglmzerocomp}, we can
  only compute $O(D')$ sequence terms $\seqw^{(k)}_i$ for
  $0\leq k\leq n-1$ in $O((t+n) D D')$ operations. As a trade-off, the minimal
  recurrence relation is computed in $O(D'^2)$ operations but this is
  not the bottleneck.
\end{proof}

\subsection{The case where $I$ is positive-dimensional ideal}
\label{ss:posdim}
Now, let $I$ be positive-dimensional, i.e.
$\K[\bx]/I$ is an infinite-dimensional vector space. Still, we assume
that $\colid{I}{\vphi}$ is zero-dimensional and in shape position with
$\cHlex=\curl{h_n(x_n),x_{n-1}-h_{n-1}(x_n),\ldots,x_1-h_1(x_n)}$
such that $\deg h_n=D'$.

To compute the polynomials
$h_n,h_{n-1},\ldots,h_1\in\K[x_n]$, we shall show that we
can rely on linear algebra routines in
a finite-dimensional vector subspace of $\K[\bx]/I$. We start by
defining such a vector subspace by giving a monomial basis thereof.

\begin{lemma}\label{lem:sigma}
  Let $\prec$ be a monomial order and $I$ be an ideal of $\K[\bx]$.
  Let $\cG$ be the reduced \gb of $I$ for $\prec$ and $S$ be its associated
  staircase. Let $\vphi\in\K[\bx]\setminus I$ be a polynomial and
  $1\leq k\leq n$ such that
  $J=\pare{\colid{I}{\vphi}}\cap\K[x_k,\ldots,x_n]$ is
  zero-dimensional with staircase $T$ for another monomial order $\prec_2$.

  Then,
  the set
  \[\Sigma=\curl{\sigma\in S\,\middle\vert\,
      \exists s\in\bigcup_{\tau\in T}
      \supp\nf{\tau\vphi}{\cG}{\prec},\sigma\mid s},
  \]
  is a finite subset of $S$, which is a staircase as well, and is such that for all
  $(i_k,\ldots,i_n)\in\N^{n-k+1}$,
  $\supp\nf{x_k^{i_k}\cdots x_n^{i_n}\vphi}{\cG}{\prec}
  \subseteq\Sigma$.
\end{lemma}
\begin{proof}
  Since $T$ is finite, then $\bigcup_{\tau\in T}
  \supp\nf{\tau\vphi}{\cG}{\prec}$ is a finite set of monomials. Since a
  monomial admits finitely many divisors, then $\Sigma$ is finite.

  Let $t$ be a monomial not in $T$, then for each $\tau\in T$ such
  that $\tau\prec_2 t$, there exists $c_{\tau}\in\K$ such that
  $h=t-\sum_{\substack{\tau\in T\\\tau\prec_2 t}}c_{\tau}\tau\in
  J$, that is
  $\nf{h\vphi}{\cG}{\prec}=0$. Then, by linearity
  \[\nf{t\vphi}{\cG}{\prec}=
    \sum_{\tau\in T}c_{\tau}\nf{\tau\vphi}{\cG}{\prec},\]
  and $\supp\nf{t\vphi}{\cG}{\prec}\subseteq\Sigma$.
\end{proof}

\begin{remark}
  By our assumptions, we can take $T=\curl{1,x_n,\ldots,x_n^{D'-1}}$
  for $\llex$
  so that
  \[\Sigma=\curl{\sigma\in S\,\middle\vert\,\exists
      s\in\bigcup_{i=0}^{D'-1}\supp\nf{x_n^i\vphi}{\cGdrl}{\ldrl},
      \sigma\mid s}.\]
\end{remark}

\begin{proposition}\label{prop:posdimlincomb}
  Let $I$ be a positive-dimensional ideal of $\K[\bx]$, $\cGdrl$ be its
  reduced \gb for $\ldrl$ and $\Sdrl$ be its associated staircase. Let
  $\vphi\in\K[\bx]\setminus I$ such that $\colid{I}{\vphi}$ is
  zero-dimensional and in shape position. Let
  $\cHlex=\{h_n(x_n),x_{n-1}-h_{n-1}(x_n),\ldots,x_1-h_1(x_n)\}$ be
  the reduced \gb of $\colid{I}{\vphi}$ for $\llex$, with $\deg h_n=D'$.
  % and $\Tlex$
  % be the associated staircase.

  Let $\Sigma=\curl{\sigma\in\Sdrl\,\middle\vert\,
    \exists s\in\bigcup_{i=0}^{D'-1}\supp\nf{x_n^i\vphi}{\cGdrl}{\ldrl},
    \sigma\mid s}$.

  Then, there exist unique
  $c_0,\ldots,c_{D'-1}\in\K$ and, for all $1\leq k\leq n-1$,
  unique $\gamma_{0,k},\ldots,\gamma_{D'-1,k}\in\K$ such that
  \begin{align*}
    \nf{x_n^{D'}\vphi}{\cGdrl}{\ldrl}
    &=-c_{D'-1}\nf{x_n^{D'-1}\vphi}{\cGdrl}{\ldrl}\\
    &\qquad-\cdots
      -c_0\nf{\vphi}{\cGdrl}{\ldrl},\\
    \nf{x_k\vphi}{\cGdrl}{\ldrl}
    &=\gamma_{D'-1,k}\nf{x_n^{D'-1}\vphi}{\cGdrl}{\ldrl}\\
    &\qquad+\cdots+
      \gamma_{0,k}\nf{\vphi}{\cGdrl}{\ldrl},
  \end{align*}
  and all these vectors lie in the vector space spanned by $\Sigma$.

  Furthermore, $h_n=x_n^{D'}+c_{D'-1}x_n^{D'-1}+\cdots+c_0$ and for
  all $1\leq k\leq n-1$, we have
  $h_k=\gamma_{D'-1,k}x_n^{D'-1}\allowbreak +\cdots+\gamma_{0,k}$.
\end{proposition}
\begin{proof}
  By assumption, $\cHlex$ is in shape position, hence
  $\Tlex=\curl{1,x_n,\ldots,x_n^{D'-1}}$ is a monomial basis
  of $\K[\bx]/(\colid{I}{\vphi})$. Thus, there exist unique
  $c_0,\ldots,c_{D'-1}\in\K$ and for all $1\leq k\leq D'-1$,
  unique $\gamma_{0,k},\ldots,\gamma_{D'-1,k}$ such that
  $h_n=x_n^{D'}+c_{D'-1}x_n^{D'-1}+\cdots+c_0$ and for
  all $1\leq k\leq n-1$, $h_k=\gamma_{D'-1,k}x_n^{D'-1}\allowbreak
  +\cdots+\gamma_{0,k}$.

  Now, multiplying $h_n$ and $x_k-h_k$ by $\vphi$ makes these
  polynomials lie in $I$. Hence the equality on the normal forms. Now,
  they all have their support in $\Sigma$, hence they lie in the
  vector space spanned by $\Sigma$.
  % \qed
\end{proof}
As a consequence, we can compute $h_n,h_{n-1},\ldots,h_1$ by means of
Gaussian elimination in the vector space spanned by $\Sigma$.

\begin{remark}\label{rk:iterative_sigma}
  Observe that whenever $\supp\vphi$ is much larger than
  $\bigcup_{g\in\cGdrl}\supp g$, its support might already be large
  enough to define $\Sigma$, or a subset $\Sigma'$ thereof such that
  the vector space it spans allows us to recover
  $h_n,h_{n-1},\ldots,h_1$ by Gaussian elimination. Furthermore, testing
  effectively the correctness of the computed polynomials
  can be done via multiplying $h_n(x_n)$ (\resp $x_k-h_k(x_n)$
  for $1\leq k\leq n-1$) by $\vphi$ and checking if its normal form \wrt
  $\cGdrl$ and $\ldrl$ is $0$.

  Otherwise, $\Sigma'$ was too small. We can enlarge it by
  adding the missing monomials of $\supp
  \nf{x_n\vphi}{\cGdrl}{\ldrl}$ and their divisors,
  iterating this process further.
\end{remark}

\subsection{Building the multiplication matrix}
\label{ss:matrix}
Let $W$ be the vector subspace of $\K[\bx]/I$ whose monomial basis is
$\Sigma$. The goal is to build a matrix $\tiM_{x_n}$ of a linear map
from $W$ to itself allowing us to compute
$h_n,h_{n-1},\ldots,h_1$. The image of the map $f\in W\mapsto
\nf{x_n f}{\cGdrl}{\ldrl}\in\K[\bx]/I$ need not be in $W$. Thus we compose
it with the projection $\pi_{W}$ onto $W$, which discards any monomial of
$\supp\nf{x_n f}{\cGdrl}{\ldrl}\subseteq\Sdrl$ not in $\Sigma$.

\begin{lemma}\label{lem:buildingmatrix}
  Let $I$ be a positive-dimensional ideal of $\K[\bx]$, $\cGdrl$ be its
  reduced \gb for $\ldrl$ and $\Sdrl$ be its associated staircase. Let
  $\vphi\in\K[\bx]\setminus I$ such that $\colid{I}{\vphi}$ is
  zero-dimensional and in shape position. Let
  $\cHlex=\{h_n(x_n),x_{n-1}-h_{n-1}(x_n),\ldots,\allowbreak x_1-h_1(x_n)\}$ be
  the reduced \gb of $\colid{I}{\vphi}$ for $\llex$, with $\deg h_n=D'$.

  Let $\Sigma=\curl{\sigma\in\Sdrl\,\middle\vert\,
    \exists s\in\bigcup_{i=0}^{D'-1}\supp\nf{x_n^i\vphi}{\cGdrl}{\ldrl},
    \sigma\mid s}\subset \Sdrl$ be a
  finite staircase and $W$ be the vector subspace of $\K[\bx]/I$ it
  spans.

  Let $\tiM_{x_n}$ be the matrix of the linear map $f\in W\mapsto
  \pi_W\big(\nf{x_n f}{\cGdrl}{\ldrl}\big)\in W$.

  Assuming $\Sigma=\{\sigma_0,\ldots,\sigma_{N-1}\}$ with for all
  $0\leq i<N-1$, $\sigma_i\ldrl\sigma_{i+1}$,
  one can build the matrix
  $\tiM_{x_n}=(\tim_{i,j})_{0\leq i,j<N}$ with the following
  procedure:
  \begin{itemize}
  \item if $x_n\sigma_j=\sigma_k$, then $\tim_{k,j}=1$ and for all
    $0\leq i<N$, $i\neq k$, $\tim_{i,j}=0$;
  \item if $x_n\sigma_j\in\Sdrl\setminus\Sigma$, then for all $0\leq
    i<N$, $\tim_{i,j}=0$;
  \item otherwise for all $0\leq i<N$, $\tim_{i,j}$ is the
    coefficient of $\sigma_i$ in $\nf{x_n\sigma_j}{\cGdrl}{\ldrl}$.
  \end{itemize}
\end{lemma}
\begin{proof}
  By construction, the matrix $\tiM_{x_n}$ has its $j$th column which is the
  vector of coefficients of the projection of
  $\nf{x_n\sigma_j}{\cGdrl}{\ldrl}$ onto $W$ in the basis $\Sigma$.

  The first case is immediate.

  Observe that in the second case, $x_n\sigma_j\in\Sdrl$, hence it is
  its own normal form \wrt $\cGdrl$ and $\ldrl$. Yet, since it is not
  in $\Sigma$, its projection is $0$.

  The last case is obtained by linear combination of the first two.
\end{proof}

\begin{remark}\label{rk:normalform}
  While the first two cases of \Cref{lem:buildingmatrix} require
  no computation whatsoever, a priori, the last one needs a normal form
  computation.

  Since $\cGdrl$ is a reduced \gb, if
  $x_n\sigma_j=\LT[\ldrl](g)$ for some $g\in\cGdrl$,
   then
  $\nf{x_n\sigma_j}{\cGdrl}{\ldrl}=x_n\sigma_j - g$. Therefore, only
  the case $x_n\sigma_j\in\LT[\ldrl](I)\setminus\LT[\ldrl](\cGdrl)$
  requires a nontrivial normal form computation.
\end{remark}

\begin{proposition}\label{prop:spfglmcolon}
  Let $I$ be a positive-dimensional ideal of $\K[\bx]$, $\cGdrl$ be its
  reduced \gb for $\ldrl$ and $\Sdrl$ be the associated staircase. Let
  $\vphi\in\K[\bx]\setminus I$,
  $\cHlex=\curl{h_n(x_n),x_{n-1}-h_{n-1}(x_n),\ldots,x_1-h_1(x_n)}$ be
  the reduced \gb of $\colid{I}{\vphi}$, with $\deg h_n=D'$.

  Let $\Sigma=\curl{\sigma\in\Sdrl\,\middle\vert\,
    \exists s\in\bigcup_{i=0}^{D'-1}\supp\nf{x_n^i\vphi}{\cGdrl}{\ldrl},
    \sigma\mid s}$ be of size $D$ and $W$ be the vector subspace of $\K[\bx]/I$
  spanned by $\Sigma$.

  Let $\tiM_{x_n}$ be the matrix of the map
  $f\in W\mapsto\pi_W\pare{\nf{x_n f}{\cGdrl}{\ldrl}}\in W$ in the
  basis $\Sigma$ and $\bphi$ be the vector of coefficients of $\vphi$
  in the basis $\Sigma$.

  Let $\bmr\in\Kbar^D$ be a row-vector and
  $\bseqw=(\seqw_i)_{i\in\N} =(\bmr \tiM_{x_n}^i\bphi)_{i\in\N}$. Let $d\in\N$ be
  minimal such that there exist $c_0,\ldots,c_{d-1}\in\K$ such that
  \[\forall i\in\N,\quad
    \seqw_{i+d}+c_{d-1}\seqw_{i+d-1}+\cdots+c_0\seqw_i=0.\]
  Then, the polynomial $x_n^d+c_{d-1}x_n^{d-1}+\cdots+c_0$ divides
  $h_n$. Furthermore, if $\bmr$ is generic in $\Kbar^D$, then $d=D'$ and
  $h_n=x_n^d+c_{d-1}x_n^{d-1}+\cdots+c_0$.

  In addition, for all $1\leq
  k\leq n-1$, let $\bpsi_k$
  be the vector of coefficients of\\
  $\pi_W\pare{\nf{x_k\vphi}{\cGdrl}{\ldrl}}$. Then, for all $1\leq
  k\leq n-1$,
  there exist unique $ \gamma_{0,k},\ldots,\allowbreak
  \gamma_{D'-1,k}\in\K$ such that
  \[
    \begin{pmatrix}
      \seqw_0	&\seqw_1	&\cdots	&\seqw_{D'-1}\\
      \seqw_1	&\seqw_2	&\cdots	&\seqw_{D'}\\
      \vdots	&\vdots		&	&\vdots\\
      \seqw_{D'-1}	&\seqw_{D'}	&\cdots	&\seqw_{2 D'-2}\\
    \end{pmatrix}
    \begin{pmatrix}
      \gamma_{0,k}\\\gamma_{1,k}\\\vdots\\\gamma_{D'-1,k}
    \end{pmatrix}
    =
    \begin{pmatrix}
      \bmr \tiM_{x_n}^0 \bpsi_k\\\bmr \tiM_{x_n}^1 \bpsi_k\\\vdots\\
      \bmr \tiM_{x_n}^{D'-1} \bpsi_k\\
    \end{pmatrix},
  \]
  and $h_k=\gamma_{D'-1,k}x_n^{D'-1}+\cdots+\gamma_{0,k}$.
\end{proposition}
\begin{proof}
  The proof is similar to the proofs of \Cref{lem:elimcolon,lem:paramcolon}.

  For the first statement, by \Cref{prop:posdimlincomb} and the definition of
  $\tiM_{x_n}$, there exists a smallest integer
  $b$ such that $\bphi, \tiM_{x_n}\bphi,\ldots,\tiM_{x_n}^b\bphi$ are
  not linearly independent. We let $a_0,\ldots,a_{b-1}\in\K$ such that
  \[\tiM_{x_n}^b\bphi+a_{b-1}\tiM_{x_n}^{b-1}+\cdots+a_0\bphi=0.\]
  Thus,
  $\pi_W\pare{\nf{\pare{x_n^b+a_{b-1}x_n^{b-1}+\cdots+a_0}\vphi}{\cGdrl}{\ldrl}}=0$. Since
  the support of this normal form is included in $\Sigma$, the normal
  form actually lies in $W$. Hence, we conclude that
  the polynomial $(x_n^b+a_{b-1}x_n^{b-1}+\cdots+a_0)\vphi$
  is in $I$ and thus that
  $x_n^b+a_{b-1}x_n^{b-1}+\cdots+a_0$ is in $\colid{I}{\vphi}$.

  By minimality of $b$, this ensures that
  $h_n=x_n^b+a_{b-1}x_n^{b-1}+\cdots+a_0$.

  Now, multiplying the vector equality above by $\bmr \tiM_{x_n}^i$ on
  the left yields
  \[\forall\,i\in\N,\ \seqw_{i+b}+a_{b-1}\seqw_{i+b-1}+\cdots+a_0\seqw_i=0.\]
  Thus, $\bseqw$ is linearly recurrent of order at most $b$ and $d\leq
  b$. Since linear recurrences are in one-to-one correspondence with
  polynomials, these polynomials define an ideal of $\K[x_n]$ spanned
  by $x_n^d+c_{d-1}x_n^{d-1}+\cdots+c_0$ that contains $h_n$. Hence
  the former divides the latter.

  For the second statement, since $d=D'$, then $\bseqw$ does not
  satisfy any linear recurrence
  relation of order less than $D'$. Thus, there is no vector in the
  kernel of the above Hankel matrix, see~\cite{BrentGY1980}.

  Let $1\leq k\leq n-1$ and $h_k(x_n)=\alpha_{D'-1,k}x_n^{D'-1}+\cdots+\alpha_{0,k}$, then
  $(x_k-h_k(x_n))\vphi$ is in $I$ and
  $\nf{\pare{x_k-h_k(x_n)}\vphi}{\cGdrl}{\ldrl}=0$. Projecting onto
  $W$, we have
  \[\bpsi_k=\alpha_{D'-1,k}\tiM_{x_n}^{D'-1}\bphi+\cdots+\alpha_{0,k}\bphi.\]
  Now, multiplying this equality by $\bmr \tiM_{x_n}^i$ for $0\leq i\leq
  D'-1$ shows that $(\alpha_0,\ldots,\alpha_{D'-1})^{\rT}$ is a solution
  of the above linear system. Since the matrix has full
  rank, the solution is unique and this ends the proof.
  % \qed
\end{proof}

We obtain the following
\Cref{algo:spfglmcolon}, the so-called \spFGLMcol algorithm.

\begin{algorithm2e}[htbp!]
  \small
  \DontPrintSemicolon
  \caption{\spFGLMcol\label{algo:spfglmcolon}}
  \KwIn{The reduced \gb $\cGdrl$ of a generic ideal
    for $\ldrl$, a polynomial $\vphi\in\K[\bx]$, and a finite staircase
    $\Sigma$ of size $N$ containing
    $\supp\nf{x_n^k\vphi}{\cGdrl}{\ldrl}$ for all $k\in\N$.
  }
  \KwOut{The reduced \gb of $\colid{I}{\vphi}$ for $\llex$, if it is in shape position.
  }
  Build the matrix $\tiM$ as in \Cref{lem:buildingmatrix}.\;
  Pick $\bmr\in\K^N$ a row-vector at random.\;
  Build $\bphi$ the column-vector of coefficients of $\vphi$
  restricted to $\Sigma$.\;
  \For{$k$ \KwFrom $1$ \KwTo $n-1$}{
    Build $\bpsi_k$ the column-vector of coefficients of
    $\nf{x_k\vphi}{\cGdrl}{\ldrl}$ restricted to $\Sigma$.\;
  }
  Compute $(\seqw_i^{(0)})_{0\leq i<2 N},
  (\seqw_i^{(1)})_{0\leq i<N},\ldots,(\seqw_i^{(n-1)})_{0\leq i<N})$
  with \Cref{algo:seqw} called on
  $\tiM,\bmr,\bphi,
  \bpsi_1,\ldots,\bpsi_{n-1}$.\nllabel{ln:spcolseqw}\;
  $h_n\leftarrow\Berlekamp (\seqw_0^{(0)},\ldots,\seqw_{2
    D-1}^{(0)})$, $D'\coloneqq\deg h_n$.\nllabel{ln:spcolbm}\;
    \lIf{$\nf{h_n\vphi}{\cGdrl}{\ldrl}\neq 0$}{\KwRet ``Bad vector''.
  }
  Compute
  $h_1\coloneqq\gamma_{D'-1,1}x_n^{N-1}+\cdots+\gamma_{0,1},\ldots,
  h_{n-1}\coloneqq\gamma_{D'-1,n-1}x_n^{N-1}+\cdots+\gamma_{0,n-1}$
  with \Cref{algo:param} called on $(\seqw_i^{(0)})_{0\leq i<2 D'-1},
  (\seqw_i^{(1)})_{0\leq i<D'},\ldots,(\seqw_i^{(n-1)})_{0\leq
    i<D'})$.\nllabel{ln:spcolparam}\;
  \For{$k$ \KwFrom $1$ \KwTo $n-1$}{
    \lIf{$\nf{(x_k-h_k(x_n))\vphi}{\cGdrl}{\ldrl}\neq 0$}{\KwRet ``Not in
      shape position''.
    }
  }
  \KwRet $\{h_n(x_n),x_{n-1}-h_{n-1}(x_n),\ldots,x_1-h_1(x_n)\}$.
\end{algorithm2e}
Observe that \Cref{rk:spfglmzerocomp} applies also to
\Cref{algo:spfglmcolon}. 
\begin{theorem}\label{th:spfglmcolon}
  Let $I$ be a positive-dimensional ideal of $\K[\bx]$, let $\cGdrl$ be its
  reduced \gb for $\ldrl$ and $\Sdrl$ be the associated staircase. Let
  $\vphi\in\K[\bx]\setminus I$ such that $\colid{I}{\vphi}$ is
  zero-dimensional of degree $D'$ and in
  shape position.

  Let
  $\Sigma$ be a finite staircase of size $N$ containing
  $\supp\nf{x_n^i\vphi}{\cGdrl}{\ldrl}$ for all $i\in\N$.

  Let $t$ be the number of monomials $\sigma$ in $\Sigma$ such that
  $x_n\sigma\in{\LM[\ldrl](\cGdrl)}$ and let $u$ be the number of monomials
  $\sigma$ in $\Sigma$ such that
  $x_n\sigma\in\ideal{\LM[\ldrl](I)}\setminus\LM[\ldrl](\cGdrl)$.

  Then, for a generic choice of vector $\bmr\in\K^N$,
  \Cref{algo:spfglmcolon} terminates and returns the
  reduced \gb of $\colid{I}{\vphi}$ for $\llex$. To do so, it requires at most $u+n$
  normal form computations \wrt $\cGdrl$ and $\ldrl$
  plus $O\pare{(t+u+n) N D'}$ operations in $\K$.
\end{theorem}
\begin{proof}
  The proof follows the proofs of \Cref{th:spfglm,th:spfglmcolonzero}.

  By the terminations of the normal form computations and the
  Berlekamp--Massey algorithm, \Cref{algo:spfglmcolon} terminates.

  We now prove the correctness of the algorithm.
  By \Cref{prop:spfglmcolon}, we know that the polynomial
  returned by the Berlekamp--Massey algorithm, on line~\ref{ln:spcolbm} is
  a divisor of the
  minimal univariate polynomial of the reduced \gb of
  $\colid{I}{\vphi}$ for $\llex$. Thus, it suffices to check that,
  multiplied by $\vphi$, it lies in $I$ to ensure that this is the
  correct polynomial.

  By \Cref{prop:spfglmcolon} also, we know that if a
  polynomial $x_k-h_k(x_n)$ is in $\colid{I}{\vphi}$, then calling
  \Cref{algo:param}
  allows us to compute
  $h_k$. It suffices then to multiply $x_k-h_k(x_n)$ by $\vphi$ and to
  check that it is in $I$ to ensure that $x_k-h_k(x_n)$ is in
  $\colid{I}{\vphi}$ and thus that $h_k$ is correct. If it is not,
  then $\colid{I}{\vphi}$ is actually not in shape position. This proves the
  correctness of the algorithm.

  Finally, let us prove the complexity of the algorithm.
  To build the matrix $\tiM$, we need to compute normal forms of monomials
  $x_n\sigma\not\in\Sdrl$.
  By \Cref{rk:normalform}, the only normal forms which are not
  free to compute are those of
  $x_n\sigma\in\LT[\ldrl](I)\setminus\LT[\ldrl](\cGdrl)$, by
  assumption, there are $u$ of them.
  Then, we also need to compute the $n$ normal forms of
  $\bphi,\bpsi_1,\ldots,\bpsi_{n-1}$ \wrt $\cGdrl$ and $\ldrl$.

  By \Cref{rk:spfglmzerocomp}, it suffices to compute $O(D')$
  terms for each sequence $\seqw^{(k)}_i$ for $0\leq k\leq n-1$ and
  make a call to the Berlekamp--Massey algorithm in
  $O((t+u+n) N D')$ operations.

  All in all, we have a cost of $u+n$ normal forms computations plus
  $O((t+u+n) N D')$ operations in $\K$.
  % \qed
\end{proof}

\subsection{Reduction of the size of $\Sigma$}
\label{ss:optim}
In many applications, see \Cref{s:implem}, the size of the
chosen
$\Sigma$
is much larger than the degree of $\colid{I}{\vphi}$. This contrasts
greatly with the original \spFGLM where, by definition, the size of the
staircase $\Sdrl$ is the degree of the ideal $I$. Therefore, in order to
speed the computation up, one needs to reduce the size of
$\Sigma$ as much as possible. This can be done either before any computation with $\tiM$
or after.

In particular, we shall prove that the zero columns in
$\tiM$ will make us remove many monomials in $\Sigma$
such that after this reduction, $\tiM$ does not have any zero columns
left.

\begin{lemma}
  Let $I$ be an ideal of
  $\K[\bx]$,
  $\cGdrl$ be its reduced \gb for $\ldrl$ and $\Sdrl$ be its associated
  staircase. Let $\vphi\in\K[\bx]$ be a polynomial such that
  $\pare{\colid{I}{\vphi}}$ is
  zero-dimensional with staircase $T=\curl{1,x_n,\ldots,x_n^{D'-1}}$ for $\llex$.

  Let
  \[\Sigma=\curl{\sigma\in\Sdrl\,\middle\vert\,
      \exists s\in\bigcup_{\tau\in T}
      \supp\nf{\tau\vphi}{\cG}{\prec},\sigma\mid s}\subset\Sdrl,
  \]
  and
  \[\Sigma'=\Sigma\setminus\curl{\sigma\in\Sigma\,\middle\vert\,
      \exists\,i\in\N,\ x_n^i\sigma\in\Sdrl\setminus\Sigma}.\]
  Let $W'$ be the vector subspace of $\K[\bx]/I$ spanned by
  $\Sigma'$. Let $\tiM_{x_n}'$ be the matrix of the map $f\in
  W'\mapsto\pi_{W'}\pare{\nf{x_n\sigma}{\cGdrl}{\ldrl}}\in W'$ in the
  basis $\Sigma'=\curl{\sigma_0,\ldots,\sigma_{N'-1}}$. Then,
  $\tiM_{x_n}'$ can be built using the same procedure as in
  \Cref{lem:buildingmatrix}.

  Furthermore, if its $j$th column is zero, then $x_n\sigma_j\in\LT[\ldrl](I)$.
\end{lemma}
\begin{proof}
  As $\tiM_{x_n}'$ is defined in a similar fashion as $\tiM_{x_n}$,
  the procedure of \Cref{lem:buildingmatrix} still applies.

  By construction of $\tiM_{x_n}'$, the $j$th column is $0$ if, and only if,
  \[\pi_{W'}\pare{\nf{x_n\sigma_j}{\cGdrl}{\ldrl}}=0.\]
  This can only happen in
  two cases.
  Either when $x_n\sigma_j$ is its own normal form and
  $\pi_{W'}(x_n\sigma_j)=0$, that is
  $x_n\sigma_j\in\Sdrl\setminus\Sigma'$. Or when
  $x_n\sigma_j\neq\nf{x_n\sigma_j}{\cGdrl}{\ldrl}$, that is
  $x_n\sigma_j\in\LT[\ldrl](I)$, and then the projection onto $W'$ is
  $0$.

  By assumption on $\Sigma'$, the multiplication of $\sigma_j$ by
  $x_n$ cannot reach a monomial in $\Sdrl$ not in $\Sigma$, hence if
  the projection of its normal form is $0$, this means that
  $x_n\sigma_j$ is not its own normal form, \ie $x_n\sigma_j\in\LT[\ldrl](I)$.
  % \qed
\end{proof}

Observe that $\Sigma'$ need not be a staircase: indeed, $1$ may have
even been removed.
\begin{proposition}\label{prop:reduce_sigma}
  Let $I$ be a positive-dimensional ideal of $\K[\bx]$, let $\cGdrl$ be its
  reduced \gb for $\ldrl$ and $\Sdrl$ be the associated staircase. Let
  $\vphi\in\K[\bx]\setminus I$ such that $\colid{I}{\vphi}$ is
  zero-dimensional and in
  shape position.

  Let
  \[\Sigma=\curl{\sigma\in\Sdrl\,\middle\vert\,
      \exists s\in\bigcup_{\tau\in T}
      \supp\nf{\tau\vphi}{\cG}{\prec},\sigma\mid s}\subset\Sdrl,
  \]
  and
  \[\Sigma'=\Sigma\setminus\curl{\sigma\in\Sigma\,\middle\vert\,
      \exists\,i\in\N,\ x_n^i\sigma\in\Sdrl\setminus\Sigma}.\]

  Let $W$ (\resp $W'$) be the vector subspace of $\K[\bx]/I$ spanned
  by $\Sigma$ (\resp $\Sigma'$). Let $\bphi$ (\resp $\bphi'$) be the
  vector of coefficients of the projection of
  $\nf{\vphi}{\cGdrl}{\ldrl}$ onto $W$ (\resp $W'$).

  Let $d\in\N$ be minimal such that there exist
  $c_0,\ldots,c_{d-1}\in\K$ such that
  \[\forall\,i\in\N,\quad
    \tiM_{x_n}^{i+d}\bphi
    +
    c_{d-1}\tiM_{x_n}^{i+d-1}\bphi+\cdots+c_0\tiM_{x_n}^i\bphi=0.\]
  Let $b\in\N$ be minimal such that there exist
  $a_0,\ldots,a_{b-1}\in\K$ such that
  \[\forall\,i\in\N,\quad
    \tiM_{x_n}^{\prime i+b}\bphi'
    +
    a_{b-1}\tiM^{\prime i+b-1}_{x_n}\bphi+\cdots
    +a_0\tiM_{x_n}^{\prime i}\bphi=0.\]
  Then, $b=d$ and $a_0=c_0,\ldots,a_{b-1}=c_{d-1}$.
\end{proposition}
\begin{proof}
  By assumption, the sequence $\pare{\tiM_{x_n}^{\prime i}\bphi}_{i\in\N}$
  is linear recurrent of order $d$. The linear recurrences it
  satisfies are then in one-to-one correspondence with the ideal
  $\ideal{h_n}\in\K[\bx]$, where $h_n=x_n^d+c_{d-1} x_n^{d-1}+\cdots+c_0$.

  Now, there exist a smallest integer $\beta$ such that there exist unique
  $\alpha_0,\ldots,\alpha_{\beta-1}$ such that
  \[\forall\,i\in\N,\quad
    \tiM_{x_n}^{i+\beta+1}\bphi
    +
    \alpha_{\beta-1}\tiM_{x_n}^{i+\beta-1}\bphi+\cdots+\alpha_0\tiM_{x_n}^{i+1}\bphi=0.\]
  Hence, the sequence $\pare{\tiM_{x_n}^{\prime i+1}\bphi}_{i\in\N}$
  is linear recurrent of order $\beta$. Therefore,
  $x_n^{\beta}+\alpha_{\beta-1}+\cdots+\alpha_0$ divides
  $h_n$. Furthermore, because of the extra multiplication by
  $\tiM_{x_n}$ in the definition of this sequence, we know that the
  ideal of $\K[x_n]$ in one-to-one correspondence with its sets of
  linear recurrence relation is actually
  $\colid{\ideal{h_n}}{x_n}$. Thus it is spanned by $h_n$ if $x_n\nmid
  h_n$ and by $h_n/x_n$ otherwise.

  Let us denote $\sigma_0\ldrl\cdots\ldrl \sigma_{N-1}$ the monomials in $\Sigma$.
  If a monomial $\sigma_j\in\Sigma$ is such that $x_n\sigma_j\in
  \Sdrl\setminus\Sigma$, which implies that the $j$th column of
  $\tiM_{x_n}$ is $0$, then the coefficients of
  $\tiM_{x_n}^{\prime i+1}\bphi,\ldots,\tiM_{x_n}^{\prime i+\beta}\bphi$
  are
  all independent from the $j$th coefficient of
  $\bphi$. Thus, this coefficient does not appear in the second linear
  system and $c_0,\ldots,c_{d-1}$ do not depend on
  it. Hence, we can reduce the linear system by removing
  $\sigma_j$ from $\Sigma$ without changing the linear
  recurrence relation of smallest order that is satisfied.

  Now, if this monomial $\sigma_j$ is divisible by $x_n$, then there exists an index
  $i<j$ such that $\sigma_i=\sigma_j/x_n$. This implies that the $i$th
  column of the new matrix
  is zero. Thus, the previous argument can be repeated to remove $\sigma_i$
  from $\Sigma$ as well.

  By recurrence, at the end of this removal procedure, the set of
  monomials \emph{is} $\Sigma'$ and there was no change whatsoever in
  the linear recurrence relations satisfied by the modified
  sequence. Hence $d$ is the smallest integer such that there exist
  unique $c_0,\ldots,c_{d-1}\in\K$ such that
  \[\forall\,i\in\N,\quad
    \tiM_{x_n}^{\prime i+d}\bphi'
    +
    c_{d-1}\tiM^{\prime i+d-1}_{x_n}\bphi+\cdots
    +c_0\tiM_{x_n}^{\prime i}\bphi=0,
  \]
  in other words, $b=d$ and $a_0=c_0,\ldots,a_{b-1}=c_{d-1}$.
\end{proof}

\subsection{Non shape position case}
\label{ss:nonshape}
Next, we want to lift the assumption that
$\colid{I}{\vphi}$ is in shape position. In the \spFGLM algorithm,
this is easy to test, see \Cref{ss:spfglm}: The
minimal univariate polynomial in $x_n$ has the same degree as the
ideal if, and only if, the ideal is in shape
position. However, now, we do not know the degree of the polynomial
$h_n$ such that $\pare{\colid{I}{\vphi}}\cap\K[x_n]=\ideal{h_n}$.
Since for a generic choice of $\bmr$, we know that the \spFGLMcol
algorithm computes correctly $h_n$ on
line~\ref{ln:spcolbm}, the computation of the normal form at
the following line can be skipped. Now, the goal is to avoid computing the normal
forms of $(x_k-h_k(x_n))\vphi$ to ensure that $\colid{I}{\vphi}$ is in shape
position using the following lemma.
\begin{lemma}\label{lem:verifparam}
  Let $J$ be a zero-dimensional ideal of $\K[\bx]$. Let
  $\lambda\in\bar{\K}$ be generic. Then, for $1\leq k\leq n$,
  $J=\colid{J}{x_k+\lambda}$.
\end{lemma}
\begin{proof}
  Clearly
  $J\subseteq\colid{J}{x_k+\lambda}$ so it remains to prove the converse
  inclusion for generic $\lambda$. This is equivalent to proving that
  the converse inclusion
  does not hold for only finitely many possible
  values of $\lambda$.

  Let us assume that $\colid{J}{x_k+\lambda}\neq J$ and let
  $f\in\colid{J}{x_k+\lambda}$ not in $J$. Then
  $g=(x_k+\lambda)f\in J$. Thus, $g$
  vanishes on the finitely many points of the variety defined by
  $J$. If we assume that $J$ is radical, then $f$ does not vanish on
  at least one of these points but $g$ does. Since $x_k+\lambda$ is prime in $\K[\bx]$,
  this means that $x_k+\lambda$ vanishes on this point. Therefore,
  one of the points of the variety defined by $J$ has its $k$th coordinate which is
  $-\lambda$. Thus, this situation can only occur for finitely many choices of
  $\lambda$.

  If now, $J$ is not radical, then
  the same reasoning applies if one takes the multiplicities into
  account as well. Hence, only finitely many $\lambda\in\bar{\K}$ are such
  that $\colid{J}{x_k+\lambda}\neq J$.
\end{proof}

\begin{remark}\label{rk:verifparamcol}
  Thanks to \Cref{lem:verifparam} applied to
  $J=\colid{I}{\vphi}$, we can check if the polynomial $x_k-h_k(x_n)$
  computed by the \spFGLMcol algorithm is correct. We compute
  $x_k-h_k'(x_n)$ for the ideal $\colid{I}{x_k+\lambda}$ by
  \begin{enumerate}
  \item Building $\psi_k'$ the column-vector of
    $\nf{x_k^2\vphi}{\cGdrl}{\ldrl}$ restricted to $\Sigma$.
  \item Computing $\pare{\seqw_i^{(0)}}_{0\leq i<2 N}$ and
    $\pare{\seqw_i^{(k)}}_{0\leq i<N}$ with
    \Cref{algo:seqw} called on $\bmr$, $\bpsi_k+\lambda\bphi$ and
    $\bpsi_k'+\lambda\bpsi_k$.
  \end{enumerate}
  If $h_k=h_k'$  for
  generic $\lambda$, then $x_k-h_k(x_n)$ is in $\colid{I}{\vphi}$.
\end{remark}

Now let us assume that $\colid{I}{\vphi}$ is
not in shape position but its radical $\sqrt{\colid{I}{\vphi}}$ is.

\begin{proposition}
  Let $I$ be a positive-dimensional ideal of $\K[\bx]$, let $\cGdrl$ be its
  reduced \gb for $\ldrl$ and $\Sdrl$ be the associated staircase. Let
  $\vphi\in\K[\bx]\setminus I$ such that $\colid{I}{\vphi}$ is
  zero-dimensional, let $\cHlex$ be its reduced \gb for $\llex$ and
  let $h_n\in\cHlex\cap\K[x_n]$.

  If $\sqrt{\colid{I}{\vphi}}$ is in shape position, then one can
  compute its reduced \gb for $\llex$ calling the \spFGLMcol algorithm
  with the following modifications:
  \begin{enumerate}
  \item On line~\ref{ln:spcolbm}, $h_n$ is the squarefree part of the
    polynomial returned by the Berlekamp--Massey algorithm.
  \item On line~\ref{ln:spcolparam}, $h_k$ is obtained thanks
    to~\cite[Algo.~2]{Hyun2020163}, see also~\cite{BostanSS2003}.
  \end{enumerate}
\end{proposition}
\begin{proof}
  By \Cref{lem:elimcolon}, we already know that we can recover
  the minimal univariate polynomial in $x_n$ of
  $\colid{I}{\vphi}$. Extracting its squarefree part yields the one of
  $\sqrt{\colid{I}{\vphi}}$.

  Now, Algorithm~2 of~\cite{Hyun2020163} called on these sequences
  yields the polynomials with leading terms $x_1,\ldots,x_{n-1}$ in
  $\sqrt{J}$ for some ideal $J$. By construction of these sequences,
  $J=\colid{I}{\vphi}$.
\end{proof}

\begin{remark}\label{rk:verifparam}
  In practice, when an ideal $J$ is not in shape position, it is not easy to
  check that $\sqrt{J}$ is. Therefore, using \Cref{lem:verifparam} is the
  cornerstone of our probabilistic verification algorithm in
  \msolve~\cite{msolve,msolveweb} when $J$ is not in shape
  position but its radical might be, see~\cite[Sec.~4.4]{msolve}. We
  proceed as in \Cref{rk:verifparamcol}:

  \begin{enumerate}
  \item Compute the polynomials $x_k-g_k(x_n)$ in $\sqrt{J}$ for $1\leq
    k\leq n-1$, with $\deg g_k$ minimal.
  \item Compute the polynomials $x_k-g_k'(x_n)$ in $\sqrt{\colid{J}{x_k+\lambda}}$ for
    $\lambda$ picked at random and $1\leq k\leq n-1$, with $\deg g_k'$
    minimal.
  \item Check whether $g_k=g_k'$ for $1\leq k\leq n-1$.
  \end{enumerate}
  By \Cref{lem:verifparam}, for a generic $\lambda$,
  $J=\colid{J}{x_k+\lambda}$, hence both
  radical ideals are the same. Furthermore, if they are in shape position, then
  $g_k=g_k'$ for $1\leq k\leq n-1$. Therefore, any discrepancy must come from
  the fact that $\sqrt{J}$ is \emph{not} in shape position and the polynomials
  $x_k-g_k(x_n)$ and $x_k-g_k'(x_n)$ are meaningless.
\end{remark}

%%% Local Variables:
%%% mode: latex
%%% TeX-master: "new_saturation"
%%% End:

\section{Implementation and practical experiments}\label{s:implem}
We implemented \Cref{algo:F4SAT,algo:spfglmcolon} in
\msolve~\cite{msolve,msolveweb}, using the {\tt C} programming language. The
saturation examples we use come from classical benchmarks of real algebraic
geometry when it comes to compute {\em limits} of critical points of the
restriction of a polynomial map to some algebraic set depending on a parameter.
This is used, for instance, for computing sample points in singular real algebraic sets
\cite{Sa05} or computing their real dimension \cite{LaSa21} and boils down to the
computation of saturated ideals.

These computations were performed on a computing server with 1.48 TB of memory
and an Intel Xeon Gold 6244 @ 3.60GHz processor.

In \Cref{tab:F4SATpos,tab:F4SATzero}, we report on timings
for computing the \gb of the saturation $\satid{I}{\vphi}$ of an ideal $I$
\wrt $\vphi$ for $\ldrl$. In both tables, $I$ is positive-dimensional. However, in
\Cref{tab:F4SATpos}, $\satid{I}{\vphi}$ is also
positive-dimensional, while in \Cref{tab:F4SATzero} it is
zero-dimensional.

We optimized the \FfourSAT algorithm~(\Cref{algo:F4SAT}),
as discussed in \Cref{ss:F4SAToptim}.
Columns learn~1 and learn~2 are for the two learning
rounds of the tracer while column apply is for the apply
phase.

The columns \msolve correspond to our implementation of Rabinowitsch
trick in \msolve, with column prob.\ using the probabilistic linear
algebra while the column learn is for the learning phase of the tracer
and the column apply is for the apply phase (see \cite{msolve}). These last two
columns are to be compared with the learn 1 and 2, and the apply phases of
\FfourSAT. The one related to the probabilistic linear algebra is to be compared
to the apply phase of \FfourSAT.

\added{We compare the maximum degree of the polynomials that are
  handled by the different algorithms. In column \Fquatre, $D_I$ is the
  maximum degree when calling \Fquatre on the generators of $I$. This
  is useful if one were to perform this computation first
  before using \Cref{lem:staircase} and \Cref{algo:linalgcolon} to determine the
  the polynomials in the \gb of $\colid{I}{\vphi}$. For \FquatreSAT,
  the column $D_{\max}$ corresponds to the maximum degree during the first learning
  phase: i.e.\ either for reducing the S-polynomials or for searching new
  polynomials with support in the current staircase. Note that, for
  these examples, this maximum
  degree is actually always reached by the former computation and only
  sometimes for the latter as well. Finally, for \msolve with
  Rabinowitsch trick,
  the column
  $D_{\max}$ gives also the highest degree of the polynomials in $n+1$
  variables.
}  

We compare our implementations with Maple~\cite{maple2020} using probabilistic
linear algebra and \Singular~\cite{Singular}.

In both tables, examples SOS mean that we consider a polynomial $f$
which the sum of $p$ squares of polynomials of degree $d$ in $n$
variables. In \Cref{tab:F4SATpos},
$I=\ideal{\diff{f}{x_1},\ldots,\diff{f}{x_{n-1}}}$
or $I=\ideal{f,\diff{f}{x_1},\ldots,\diff{f}{x_{p-1}}}$, with $p\leq n-2$,
and
$\vphi=\diff{f}{x_n}$. In the former case, the saturated ideal
has dimension $1$, while, in the latter case, it
has dimension $n-p\geq 2$. The table is divided by increasing
dimension.
In \Cref{tab:F4SATzero},
$I=\ideal{f,\diff{f}{x_1},\ldots,\diff{f}{x_{n-1}}}$ and
$\vphi=\diff{f}{x_n}$.

In general, \FquatreSAT is the most efficient attempt to compute the saturations
(sometimes with a speed-up close to $10$), \msolve's \Fquatre with elimination
order is in general a bit faster than Maple on the probabilistic linear algebra.
When applying the tracing approach we can see that sometimes the search for the
correct steps to search for new elements in the saturation in \FquatreSAT has a
bigger impact (\ie learn 1 is slower than learn 2). In other cases, the exact
linear algebra applied in learn 2 to trace the full computation is the
bottleneck (\ie learn 2 is slower than learn 1). Nevertheless, the application
phase of \FquatreSAT is in general the fastest implementation, often by an order
of magnitude. This suggests that \FquatreSAT provides a very efficient method
for computing saturations of ideals over $\bbQ$ using the multi-modular tracer
approach. The few examples where \FquatreSAT is slower than \msolve's \Fquatre
or Maple are identified by the fact that \FquatreSAT finds saturation elements a
bit later than the Rabinowitsch trick-based implementations. Clearly, one could
test for saturation elements more often in learn 1, but this would have a bigger
impact on the running time. A future plan is to apply a more dynamic and
adaptable strategy of when to search for saturation elements.

{\tiny
  \begin{table}[h]
    \begin{center}
      {\tiny
      \begin{minipage}{\linewidth}
        \begin{tabular*}{\linewidth}{@{\extracolsep{\fill}}
          |l||r||r|r|r|r||r|r||r|r||r||r| @{\extracolsep{\fill}}
          }
          \hline
          % Examples
          % &\multicolumn{1}{c|}{f4sat (learn 1)}
          % &\multicolumn{1}{c|}{f4sat (learn 2)}
          % &\multicolumn{1}{c|}{f4sat (apply)}
          % &\multicolumn{1}{c||}{msolve (prob.)}
          % &\multicolumn{1}{c|}{msolve (learn)}
          % &\multicolumn{1}{c||}{msolve (apply)}
          % & Maple (prob.)
          % & Singular
          \multirow{2}{*}{Sys-SOS}
          &\multicolumn{1}{c||}{\Fquatre}    
          &\multicolumn{4}{c||}{\FquatreSAT}
          &\multicolumn{2}{c||}{\msolve}
          &\multicolumn{2}{c||}{\msolve}
          &Maple
          &\multirow{2}{*}{Singular}
          \\
            &\multicolumn{1}{c||}{\added{$D_I$}}
            &\multicolumn{1}{c|}{\added{$D_{\max}$}}
            &\multicolumn{1}{c|}{(learn1)}
            &\multicolumn{1}{c|}{(learn2)}
            &\multicolumn{1}{c||}{(apply)}
            &\multicolumn{1}{c|}{\added{$D_{\max}$}}
            &\multicolumn{1}{c||}{(prob.)}
            &\multicolumn{1}{c|}{(learn)}
            &\multicolumn{1}{c||}{(apply)}
            &\multicolumn{1}{c||}{(prob.)}
            &
          \\
          \hline
          &\multicolumn{11}{c|}{\added{positive-dimensional to $1$-dimensional}}\\
          d3-n6-p2 & \added{13} & \added{10} & 1.31 & 0.41 & 0.31
            &\added{15} & 0.77 & 2.40 & 0.40 & 1.12  & 52.2
          \\
          d3-n6-p3 & \added{16} & \added{13} & 43.7 & 5.55 & 1.84
            &\added{18} & 25.2 & 142 & 16.6 &  35.4  & 2,902
          \\
          d3-n6-p4 & \added{19} & \added{16} & 533 & 53.1 & 19.7
            &\added{21} &171 & 882 & 126 &  223  & 39,501
          \\
          d3-n6-p5 & \added{22} & \added{19} & 1,863 & 184 & 104
            &\added{24} & 276 & 1,145 & 183 & 394  & 42,854
          \\
          d4-n6-p2 & \added{20} & \added{16} & 972 & 107 & 77
            &\added{23} & 253 & 1,176 & 191 & 394  & 28,043
          \\
          d4-n6-p3 & \added{24} & \added{20} & 31,101 & 1,316 & 596 
            &\added{27}& 7,444 & 43,803 & 6,336 & 8,817 & --
          \\
          d2-n7-p6 & \added{14} & \added{12} & 5.13 & 1.82 & 0.77
            &\added{15} & 3.01 & 15.3 & 1.84 &  4.95  & 443
          \\
          d3-n7-p2 & \added{14} & \added{11} & 13.4 & 3.61 & 2.23
            &\added{16} & 9.59 & 54.1 & 5.29 &  12.5  & 872
          \\
          d3-n7-p3 & \added{17} & \added{14} & 1,263 & 164 & 32.4
            &\added{19} & 533 & 3,647 & 406 & 984  & --
          \\
          d3-n7-p4 & \added{20} & \added{17} & 22,296 & 2,235 & 469
            &\added{22} & 6,605 & 47,286 & 5,348 & 10,001 & --
          \\
          d3-n7-p5 & \added{23} & \added{20} & 126,006 & 137,724 & 2,881
            &\added{25} & 29,740 & 204,718 & 22,925 & 33,635 & -- 
          \\
          d2-n8-p5 & \added{12} & \added{10} & 11.7 & 8.37 & 1.79
            &\added{13} & 15.1 & 99.9 & 7.92 & 20.4 & 3,972
          \\
          d2-n8-p6 & \added{14} & \added{12} & 95.7 & 63.7 & 10.5
            &\added{15} & 54.3 & 387 & 33.8 & 63.1 & 15,950
          \\
          d2-n8-p7 & \added{16} & \added{14} & 265 & 79.6 & 22.2
            &\added{17} & 81.9 & 556 & 47.2 & 122  & 15,125
          \\
          d3-n8-p2 & \added{15} & \added{12} & 228 & 276 & 18.1
            &\added{17} & 98.3 & 787 & 71.7 & 135  & 15,252
          \\
          d3-n8-p3 & \added{18} & \added{15} & 25,593 & 3,716 & 471
            &\added{20} & 11,050 & 107,744 & 8,984 & 13,705 & --
          \\
          \hline
          &\multicolumn{11}{c|}{\added{positive-dimensional to $2$-dimensional}}\\
          d2-n8-p5 & \added{13} & \added{14} & 620    & 157   & 43
            &\added{18}  & 147  & 663 & 87       
            & 166 & 26,746\\
          d2-n9-p4 & \added{11} & \added{11} & 82    & 119   & 28
            &\added{10}  & 35  & 139   & 20       
            & 50 & 4,563\\
          d2-n9-p5 & \added{13} & \added{13} & 2,889  & 1,617   & 448
            &\added{12}  & 924  & 4574   & 726       
            & 735 & --\\
          d2-n9-p6 & \added{15} & \added{17} & 40,352 & 8,290 & 2,155
            &\added{20} & 8,775 & 40,921 & 7,010      
            & 6,969 & --\\
          d2-n10-p4 & \added{11} & \added{11} & 202 & 472  & 69
            &\added{10} & 101 & 530 & 64
            & 142 & 14,237\\
          d2-n10-p5 & \added{13} & \added{13} & 11,347 & 11,512 & 2,801
            &\added{12} & 3,373 & 19,465 & 2,710       
            & 3,148 & --\\
          \hline
          &\multicolumn{11}{c|}{\added{positive-dimensional to $3$-dimensional}}\\
          d2-n9-p4 & \added{11} & \added{11} & 284 & 100 & 22
            &\added{12} & 75 & 325 & 46       
            & 77 & 10,014\\
          d2-n9-p5 & \added{13} & \added{14} & 5,157 & 758 & 201
            &\added{17} & 747 & 3,422 & 333       
            & 845 & --\\
          d2-n10-p4 & \added{11} & \added{11} & 738 & 481 & 108
            &\added{10} & 180 & 981 & 109      
            & 185 & 28,688 \\
          d2-n10-p5 & \added{13} & \added{14} & 66,845 & 12,082 & 2,141
            &\added{15} & 25,707 & 60,054
            & 23,100  
            & 6,885 & --\\
          \hline
        \end{tabular*}
        \caption{Timings in seconds, \gb for $\ldrl$,
          positive-to-positive-dimensional case
          \label{tab:F4SATpos}}
      \end{minipage}
      }
    \end{center}
  \end{table}
}

{\small
  \begin{table}[h]
    \begin{center}
      {\tiny
      \begin{minipage}{\textwidth}
        \begin{tabular*}{\textwidth}{@{\extracolsep{\fill}}|l||r||r|r|r|r||r|r||r|r||r||r|@{\extracolsep{\fill}}}
          \hline
          % Examples
          % &\multicolumn{1}{c|}{f4sat (learn 1)}
          % &\multicolumn{1}{c|}{f4sat (learn 2)}
          % &\multicolumn{1}{c|}{f4sat (apply)}
          % &\multicolumn{1}{c||}{msolve (prob.)}
          % &\multicolumn{1}{c|}{msolve (learn)}
          % &\multicolumn{1}{c||}{msolve (apply)}
          % & Maple (prob.)
          % & Singular
              \multirow{2}{*}{Examples}
          &\multicolumn{1}{c||}{\Fquatre}
          &\multicolumn{4}{c||}{\FquatreSAT}
          &\multicolumn{2}{c||}{\msolve}
          &\multicolumn{2}{c||}{\msolve}
          &Maple
          &\multirow{2}{*}{Singular}
          \\
            &\multicolumn{1}{c||}{\added{$D_I$}}
            &\multicolumn{1}{c|}{\added{$D_{\max}$}}
            &\multicolumn{1}{c|}{(learn 1)}
            &\multicolumn{1}{c|}{(learn 2)}
            &\multicolumn{1}{c||}{(apply)}
            &\multicolumn{1}{l|}{\added{$D_{\max}$}}
            &\multicolumn{1}{c||}{(prob.)}
            &\multicolumn{1}{c|}{(learn)}
            &\multicolumn{1}{c||}{(apply)}
            &\multicolumn{1}{c||}{(prob.)}
            &
          \\
          \hline
          &\multicolumn{11}{c|}{\added{positive-dimensional to $0$-dimensional}}\\
          Steiner &\added{19} & \added{19} & 115 & 134 & 67.2
            &\added{24} & 204 & 614 & 153 & 239 & 3,642
          \\
          d3-n6-p3 & \added{18} & \added{14} & 82.4& 127 & 56.7
            &\added{15} & 51.5 & 191 & 32.6 & 67.4 & 8,226
          \\
          d3-n6-p4 & \added{21} & \added{18} & 1,592 & 1,776 & 810
            &\added{23}  & 2,123 & 5,284 & 1,720 & 3,585  & --
          \\
          d3-n6-p5 & \added{24} & \added{21} & 9,646 & 7,032 & 3,321
            &\added{33} & 7,485 & 16,711  & 6,466 & 7,226  & --
          \\
          d4-n6-p2 & \added{23} & \added{15} & 720 & 1,581 & 536
            &\added{22} & 120 & 520 & 60.6 & 135  & 24,532
          \\
          d4-n6-p3 & \added{27} & \added{20} & 45,749 & 38,657 & 18,123
            &\added{29} & 40,646 & 190,009 & 35,835 & 38,466  & --
          \\
          d2-n7-p6 & \added{15} & \added{6} & 18.9 & 41.85 & 10.8
            &\added{17} &  31.8 & 101 & 19.5 & 41.4 & 1,773
          \\
          d3-n7-p2 & \added{16} & \added{12} & 28.2 & 45.2 & 23.9
            &\added{12} & 5.02 & 11.6 & 2.63 & 8.09 & 961
          \\
          d3-n7-p3 & \added{19} & \added{14} & 1,462 & 2,688 & 937
            &\added{17} & 953 & 5,851 & 875 & 1,108  & --
          \\
          d3-n7-p4 & \added{22} & \added{18} & 48,907 & 65,035 & 22,844
            &\added{22} & 40,383 & 248,889 & 34,670 & 39,729  & --
          \\
          d2-n8-p4 & \added{11} & \added{9} & 2.68 & 5.04 & 1.89
            &\added{10} & 3.55 & 10.1 & 2.02 & 4.12  & 500
          \\
          d2-n8-p5 & \added{13} & \added{11} & 47.7 & 171.9 & 37.1
            &\added{13} & 62.9 & 270 & 45.3 & 48.8  & 8,333
          \\
          d2-n8-p6 & \added{15} & \added{13} & 287 & 820 & 169
            &\added{17} & 420 & 1,599 & 301 & 350  & 54,567
          \\
          d2-n8-p7 & \added{17} & \added{15} & 1,018 & 1,841 & 442
            &\added{25} & 907 & 3,198  & 683 & 871  & --
          \\
          d3-n8-p2 & \added{17} & \added{12} & 300 & 585 & 266
            &\added{16} & 32.4 & 105 & 20.4 & 50.7  & 9,812
          \\
          d3-n8-p3 & \added{20} & \added{14} & 18,152 & 42,436 & 11,285
            &\added{28} & 15,502 & 71,595 & 8,478 & 15,182  & --
          \\
          \hline
        \end{tabular*}
        \caption{Timings in seconds, \gb for $\ldrl$,
          positive-to-zero-dimensional case
          \label{tab:F4SATzero}}
      \end{minipage}
      }
    \end{center}
  \end{table}
}

In \Cref{tab:colon_spfglm}, we compare
\Cref{algo:spfglmcolon} for computing a \gb of the
zero-dimensional colon ideal
$\colid{I}{\vphi}=\satid{I}{\vphi}$ for $\llex$ with
\Maple~\cite{maple2020} using the \texttt{Groebner:-Basis} command
followed by the \texttt{Groebner:-FGLM} command. As in \Cref{tab:F4SATzero},
$I=\ideal{f,\diff{f}{x_1},\ldots,\diff{f}{x_{n-1}}}$ and
$\vphi=\pare{\diff{f}{x_n}}^M$ for $M$ large enough, with $f$ the sum of
$p$ squares of polynomials of degree $d$ in $n$ variables.
The columns $\Card\Sigma$
and $\Card\Sigma'$ correspond to the size of the set $\Sigma$ before
and after reductions as defined in \Cref{lem:sigma},
\Cref{rk:iterative_sigma} and
\Cref{prop:reduce_sigma}, while column $D'$ gives the degree
of the saturated ideal. Whether it is between $\Card\Sigma$ and
$\Card\Sigma'$ or between $\Card\Sigma'$ and $D'$, we can observe
ratios going up to around $5$. Therefore, it is clear that the algorithm would not be as
efficient if one were to work with $\Sigma$ directly. Still, it would
be even more beneficial to reduce further the size of $\Sigma'$ to be
as close as possible to $D'$.

The column \Fquatre gives the proportion of time to compute the \gb $\cGdrl$ of
$I$ for $\ldrl$ using the \Fquatre algorithm in \msolve, the column Sat.\ order
corresponds to the time for computing iteratively
$\nf{\pare{\frac{\partial f}{\partial x_n}}^M}{\cGdrl}{\ldrl}$ with
$M$ large enough. The column Matrix
corresponds to the proportion of time to compute all the normal forms
$\nf{x_n\sigma}{\cGdrl}{\ldrl}$ for $\sigma$ in $\Sigma'$ and
$x_n\sigma\in\ideal{\LM(\cGdrl)}$ as in
\Cref{lem:buildingmatrix}. The \FGLM column gives the proportion of time to
perform the \spFGLM algorithm with this matrix. Finally, the total column gives
the total time to perform all these computations in seconds, resulting in the
computation of the saturation of $I$ \wrt $\vphi$. Likewise the columns \Maple
give the time for \Maple in seconds using Rabinowitsch
trick~\cite{Rabinowitsch1930}. The column \texttt{Basis} computes the \gb for
$\ldrl$ of $I+\ideal{1-t\frac{\partial f}{\partial x_n}}$ while the column
\texttt{FGLM} computes the \gb of the same ideal for $\llex$ with
$x_n\llex\cdots\llex x_1\llex t$.

We can notice that the \spFGLMcol algorithm approach is most efficient when
either the change of order step is the most time-consuming or when
the ratios between $\Card\Sigma$, $\Card\Sigma'$ and $D'$ are the
smallest. In the former case, the algorithm benefits from the
regularity of the computation of the reduced \gb of $I$ for $\ldrl$ compared to
the one of $I+\ideal{1-t\vphi}$ in the Rabinowitsch trick approach. In
the latter case, when $\Sigma$ or $\Sigma'$ are large compared to
$D'$, the overhead in the linear algebra part becomes
overwhelming. Clearly, in a multi-modular approach, one would want to
consider an even smaller subset of $\Sigma'$ to perform the
computations, once $D'$ is known. All in all, we can see speed-ups
that are significant and sometimes higher than $10$.

{\begin{center}
  \begin{table}[h]
    \begin{center}
      {\tiny 
        \begin{minipage}{\linewidth}
          \centering
        \begin{tabular}
          {|p{1.2cm}||p{0.6cm}|p{0.6cm}|p{1cm}||
          p{0.4cm}|p{0.6cm}|p{0.6cm}|p{0.6cm}|p{0.7cm}||
          p{0.4cm}|p{0.4cm}|p{0.7cm}|}
          \hline
          &\multicolumn{3}{c||}{sizes}
          &\multicolumn{5}{c||}{\msolve}
          &\multicolumn{3}{c|}{\Maple \texttt{Groebner}}
          \\
          &$\Card{\Sigma}$
          &$\Card{\Sigma'}$
          &$D'$
          &\Fquatre
          &Sat.\ order
          &\!\!Mat.
          &\!\!\FGLM
          &Total
          &\!\texttt{Basis}
          &\texttt{FGLM}
          &Total
          \\
          \hline
          \nameSigSigpDeg
          {d2-n8-p5}
          {\phantom{0}5746}
          {\phantom{0}2636}
          {\phantom{0}1516}
          \FquatreSatMatrixSPFGLMTotalMapleBMapleFMapleT
          {%13
          60\%}
          {%4.2
          20\%}
          {%1.6
          \phantom{1}7\%}
          {%2.7
          13\%}
          {\phantom{000}21}
          {%54
          96\%}
          {%2
          \phantom{0}4\%}
          {\phantom{000}56}
          \\
          
          \nameSigSigpDeg
          {d2-n8-p6}
          {\phantom{0}7901}
          {\phantom{0}5100}
          {\phantom{0}3756}
          \FquatreSatMatrixSPFGLMTotalMapleBMapleFMapleT
          {%49
          35\%}
          {%12.86
          \phantom{0}9\%}
          {%8.5
          \phantom{1}6\%}
          {%69
          50\%}
          {\phantom{00}140}
          {%306
          88\%}
          {%42
          12\%}
          {\phantom{00}350}
          \\
          
          \nameSigSigpDeg
          {d2-n8-p7}
          {\phantom{0}8841}
          {\phantom{0}7340}
          {\phantom{0}6444}
          \FquatreSatMatrixSPFGLMTotalMapleBMapleFMapleT
          {%107
          33\%}
          {%30
          \phantom{0}9\%}
          {%19
          \phantom{1}6\%}
          {%170
          52\%}
          {\phantom{00}320}
          {%697
          79\%}
          {%190
          21\%}
          {\phantom{00}890}
          \\
          
          \nameSigSigpDeg
          {d2-n9-p5}
          {11748}
          {\phantom{0}4548}
          {\phantom{0}2308}
          \FquatreSatMatrixSPFGLMTotalMapleBMapleFMapleT
          {%82
          56\%}
          {%21
          14\%}
          {%12
          \phantom{1}8\%}
          {%33
          22\%}
          {\phantom{00}150}
          {%277
          68\%}
          {%129
          32\%}
          {\phantom{00}410}
          \\
          
          \nameSigSigpDeg
          {d2-n9-p6}
          {18829}
          {10372}
          {\phantom{0}6788}
          \FquatreSatMatrixSPFGLMTotalMapleBMapleFMapleT
          {%492
          40\%}
          {%92
          \phantom{0}7\%}
          {%97
          \phantom{1}8\%}
          {%554
          45\%}
          {\phantom{0}1200}
          {%2954
          91\%}
          {%283
          \phantom{0}9\%}
          {\phantom{0}3200}
          \\
          
          \nameSigSigpDeg
          {d2-n9-p7}
          {24332}
          {17540}
          {13956}
          \FquatreSatMatrixSPFGLMTotalMapleBMapleFMapleT
          {%1425
          33\%}
          {%220
          \phantom{0}5\%}
          {%312
          \phantom{1}7\%}
          {%2407
          55\%}
          {\phantom{0}4400}
          {%10000
          83\%}
          {%2100
          17\%}
          {12000}
          \\
          
          \nameSigSigpDeg
          {d2-n10-p4}
          {\phantom{0}9724}
          {\phantom{0}1996}
          {\phantom{00}652}
          \FquatreSatMatrixSPFGLMTotalMapleBMapleFMapleT
          {%27.839
          67\%}
          {%11.217
          27\%}
          {%2.037
          \phantom{1}5\%}
          {%0.64
          \phantom{0}1\%}
          {\phantom{000}42}
          {%67
          99\%}
          {%0.29
          \phantom{0}1\%}
          {\phantom{000}68}
          \\
          
          \nameSigSigpDeg
          {d2-n10-p5}
          {22408}
          {\phantom{0}7372}
          {\phantom{0}3340}
          \FquatreSatMatrixSPFGLMTotalMapleBMapleFMapleT
          {%473
          52\%}
          {%102
          11\%}
          {%77
          \phantom{1}9\%}
          {%250
          28\%}
          {\phantom{00}900}
          {%1112
          97\%}
          {%40
          \phantom{0}3\%}
          {\phantom{0}1200}
          \\
          
          \nameSigSigpDeg
          {d2-n10-p6}
          {40946}
          {19468}
          {11404}
          \FquatreSatMatrixSPFGLMTotalMapleBMapleFMapleT
          {%4119
          42\%}
          {%703
          \phantom{0}7\%}
          {%936
          \phantom{1}9\%}
          {%4081
          42\%}
          {\phantom{0}9900}
          {%15221
          92\%}
          {%1314
          \phantom{0}8\%}
          {17000}
          \\
          
          \nameSigSigpDeg
          {d3-n5-p3}
          {\phantom{0}3034}
          {\phantom{0}1320}
          {\phantom{00}672}
          \FquatreSatMatrixSPFGLMTotalMapleBMapleFMapleT
          {%0.418
          39\%}
          {%0.357
          33\%}
          {%0.091
          \phantom{1}9\%}
          {%0.20
          19\%}
          {\phantom{0000}1.1}
          {%3.91
          97\%}
          {%0.14
          \phantom{0}3\%}
          {\phantom{0000}4}
          \\
          
          \nameSigSigpDeg
          {d3-n5-p4}
          {\phantom{0}3750}
          {\phantom{0}2616}
          {\phantom{0}1968}
          \FquatreSatMatrixSPFGLMTotalMapleBMapleFMapleT
          {%1.15
          27\%}
          {%1.14
          27\%}
          {%0.464
          11\%}
          {%1.51
          35\%}
          {\phantom{0000}4.3}
          {%41.006
          95\%}
          {%2.152
          \phantom{0}5\%}
          {\phantom{000}43}
          \\
          
          \nameSigSigpDeg
          {d3-n6-p3}
          {10773}
          {\phantom{0}3792}
          {\phantom{0}1632}
          \FquatreSatMatrixSPFGLMTotalMapleBMapleFMapleT
          {%30.017
          60\%}
          {%5.6
          11\%}
          {%3.78
          \phantom{1}8\%}
          {%10.51
          21\%}
          {\phantom{000}50}
          {%74
          96\%}
          {%2.9
          \phantom{0}4\%}
          {\phantom{000}77}
          \\
          
          \nameSigSigpDeg
          {d3-n6-p4}
          {16271}
          {\phantom{0}9192}
          {\phantom{0}5952}
          \FquatreSatMatrixSPFGLMTotalMapleBMapleFMapleT
          {%186
          37\%}
          {%24
          \phantom{0}5\%}
          {%32
          \phantom{1}6\%}
          {%257
          52\%}
          {\phantom{00}500}
          {%1185
          89\%}
          {%144
          11\%}
          {\phantom{0}1300}
          \\
          
          \nameSigSigpDeg
          {d3-n6-p5}
          {18897}
          {14862}
          {12432}
          \FquatreSatMatrixSPFGLMTotalMapleBMapleFMapleT
          {%140.411
          12\%}
          {%65.76
          \phantom{0}5\%}
          {%75.367
          \phantom{1}6\%}
          {%934.23
          77\%}
          {\phantom{0}1200}
          {%4625
          82\%}
          {%1013
          18\%}
          {\phantom{0}5600}
          \\
          
          \nameSigSigpDeg
          {d3-n7-p3}
          {35117}
          {10320}
          {\phantom{0}3840}
          \FquatreSatMatrixSPFGLMTotalMapleBMapleFMapleT
          {%649.401
          52\%}
          {%108.165
          \phantom{0}9\%}
          {%105.188
          \phantom{1}8\%}
          {%389.10
          31\%}
          {\phantom{0}1300}
          {%1056
          95\%}
          {%58
          \phantom{0}5\%}
          {\phantom{0}1100}
          \\
          
          \nameSigSigpDeg
          {d3-n7-p4}
          {62104}
          {29760}
          {16800}
          \FquatreSatMatrixSPFGLMTotalMapleBMapleFMapleT
          {%2327
          19\%}
          {%437
          \phantom{0}4\%}
          {%1352
          11\%}
          {%8070
          66\%}
          {12000}
          {%28339
          91\%}
          {%2817
          \phantom{0}9\%}
          {31000}
          \\
          
          \nameSigSigpDeg
          {d4-n5-p3}
          {15881}
          {\phantom{0}7560}
          {\phantom{0}4104}
          \FquatreSatMatrixSPFGLMTotalMapleBMapleFMapleT
          {%87.802
          41\%}
          {%11.897
          \phantom{0}6\%}
          {%11.436
          \phantom{1}5\%}
          {%100.75
          48\%}
          {\phantom{00}200}
          {%586.718
          94\%}
          {%34.274
          \phantom{0}6\%}
          {\phantom{00}620}
          \\
          
          \nameSigSigpDeg
          {d4-n5-p4}
          {19274}
          {14088}
          {11016}
          \FquatreSatMatrixSPFGLMTotalMapleBMapleFMapleT
          {%324
          32\%}
          {%38.482
          \phantom{0}4\%}
          {%42.328
          \phantom{1}4\%}
          {%617.61
          60\%}
          {\phantom{0}1000}
          {%4034
          86\%}
          {%638
          14\%}
          {\phantom{0}4700}
          \\
          
          \nameSigSigpDeg
          {d4-n6-p2}
          {41189}
          {\phantom{0}8424}
          {\phantom{0}1944}
          \FquatreSatMatrixSPFGLMTotalMapleBMapleFMapleT
          {%437.330
          74\%}
          {%26.584
          \phantom{0}5\%}
          {%54.202
          \phantom{1}9\%}
          {%70.38
          12\%}
          {\phantom{00}590}
          {%217
          97\%}
          {%7
          \phantom{0}3\%}
          {\phantom{00}224}
          \\
          
          \nameSigSigpDeg
          {d4-n6-p3}
          {81068}
          {32184}
          {14904}
          \FquatreSatMatrixSPFGLMTotalMapleBMapleFMapleT
          {%1741
          16\%}
          {%356
          \phantom{0}3\%}
          {%1343
          12\%}
          {%7664
          69\%}
          {11000}
          {%25229
          92\%}
          {%2006
          \phantom{0}8\%}
          {27000}
          \\
          
          \nameSigSigpDeg
          {d5-n4-p3}
          {\phantom{0}7235}
          {\phantom{0}4540}
          {\phantom{0}3040}
          \FquatreSatMatrixSPFGLMTotalMapleBMapleFMapleT
          {%1.210
          13\%}
          {%2.212
          24\%}
          {%0.974
          11\%}
          {%4.69
          52\%}
          {\phantom{0000}9}
          {%168
          93\%}
          {%12
          \phantom{0}7\%}
          {\phantom{00}180}
          \\
          
          \nameSigSigpDeg
          {d5-n5-p3}
          {54787}
          {27360}
          {15360}
          \FquatreSatMatrixSPFGLMTotalMapleBMapleFMapleT
          {%386
          \phantom{0}8\%}
          {%182
          \phantom{0}4\%}
          {%368
          \phantom{1}7\%}
          {%4155
          81\%}
          {\phantom{0}5100}
          {%19861
          92\%}
          {%1806
          \phantom{0}8\%}
          {22000}
          \\
          \hline
        \end{tabular}
        \caption{Timings in seconds, \gb for $\llex$,
          positive-to-zero-dimensional case.\label{tab:colon_spfglm}}
      \end{minipage}
      }
    \end{center}
  \end{table}
\end{center}}

%%% Local Variables:
%%% mode: latex
%%% TeX-master: "new_saturation"
%%% End:

\bibliographystyle{abbrv}
\bibliography{biblio}

\end{document}